\RequirePackage{fix-cm}
\documentclass[twocolumn]{svjour3}          
\smartqed  
\usepackage{natbib}
\usepackage[colorlinks,linkcolor=blue,citecolor=blue,urlcolor=blue]{hyperref}
\usepackage{graphicx,epsfig,times}
\usepackage{amsmath,amsfonts,amssymb}
\usepackage{bbm}
\usepackage{subfigure,cases,multirow}
\usepackage{algpseudocode,algorithm,algorithmicx}

\journalname{International Journal of Computer Vision~~~}

\newcommand{\cP}{\mathcal P}
\newcommand{\cB}{\mathcal B}
\newcommand{\cH}{\mathcal H}

\newcommand{\cW}{\mathcal W}

\newcommand{\cC}{\mathcal C}
\newcommand{\cU}{\mathcal U}
\newcommand{\cA}{\mathcal A}
\newcommand{\cM}{\mathcal M}
\newcommand{\cF}{\mathcal F}
\newcommand{\cL}{\mathcal L}

\newcommand{\cK}{\mathcal K}
\newcommand{\cT}{\mathcal T}
\newcommand{\cO}{\mathcal O}
\newcommand{\cD}{\mathcal D}
\newcommand{\cE}{\mathcal E}

\newcommand{\bR}{\mathbb R}

\newcommand{\bS}{\mathbb S}

\newcommand{\0}{ \mathbf 0 }

\newcommand{\Id}{\mathbf I_d}
\newcommand{\vt}{\mathbf{v}_\theta}
\newcommand{\x}{\mathbf x}
\newcommand{\y}{\mathbf y}

\newcommand{\p}{\mathbf p}
\newcommand{\q}{\mathbf q}
\newcommand{\s}{\mathbf s}
\newcommand{\g}{\mathbf g}

\newcommand{\fv}{\mathbf v}
\newcommand{\fu}{\mathbf u}
\newcommand{\fw}{\mathbf w}

\newcommand{\bx}{\bar{\mathbf x}}
\newcommand{\bz}{\bar{\mathbf z}}
\newcommand{\bp}{\bar{\mathbf p}}
\newcommand{\bq}{\bar{\mathbf q}}
\newcommand{\bs}{\bar{\mathbf s}}
\newcommand{\by}{\bar{\mathbf y}}

\newcommand{\ba}{\bar{\mathbf a}}
\newcommand{\bc}{\bar{\mathbf c}}

\newcommand{\bu}{\bar{\mathbf u}}
\newcommand{\bv}{\bar{\mathbf v}}
\newcommand{\bw}{\bar{\mathbf w}}

\DeclareMathOperator\interp{I}

\newtheorem{defi}{Definition}[section]
\newtheorem{coro}{Corollary}[section]
\newtheorem{lema}{Lemma}[section]

\begin{document}

\title{Global Minimum for a Finsler Elastica Minimal Path Approach
}


\author{Da Chen \and Jean-Marie Mirebeau \and Laurent D. Cohen
}


\institute{Da~Chen\and Laurent~D.~Cohen \at
              University Paris Dauphine, PSL Research University\\
              CNRS, UMR 7534, CEREMADE, 75016 Paris, France\\
               \email{chenda@ceremade.dauphine.fr}\\
              \email{cohen@ceremade.dauphine.fr}\\
           \and
           Jean-Marie Mirebeau\\
           Laboratoire de math\'ematiques d'Orsay, Universit\'e Paris-Sud, CNRS,\\
           Universit\'e Paris-Saclay, 91405 Orsay, France\\
           \email{jean-marie.mirebeau@math.u-psud.fr}\\
}

\date{Received: 19-December-2015/ Accepted: 17-November-2016}

\maketitle

\begin{abstract}
In this paper, we propose a novel curvature penalized minimal path model via an orientation-lifted Finsler metric and the Euler elastica curve. The original minimal path model computes the globally minimal geodesic by solving an Eikonal partial differential equation (PDE). Essentially, this first-order model is unable to penalize curvature which is related to the path rigidity property in the classical active contour models. 
To solve this problem, we present an Eikonal PDE-based Finsler elastica minimal path approach to address the curvature-penalized geodesic energy minimization problem.  We were successful at adding the curvature penalization to the classical geodesic energy \citep{caselles1997geodesic,cohen1997global}. The basic idea of this work is to interpret the  Euler elastica bending energy via a novel  Finsler elastica metric  that embeds a curvature penalty.  This metric is non-Riemannian, anisotropic and asymmetric,  and is defined over an orientation-lifted space by adding to the image domain the orientation as an extra space dimension. Based on this orientation lifting,  the proposed minimal path model can benefit from both the curvature and orientation of  the paths. Thanks to the fast marching method, the global minimum of the curvature-penalized geodesic energy can be computed efficiently.

We introduce two anisotropic image data-driven speed functions  that are computed by steerable filters.  Based on these orientation-dependent speed functions, we can apply the proposed Finsler elastica minimal path model to the applications of  closed contour detection, perceptual grouping  and tubular structure extraction. Numerical experiments on both synthetic and real images  show that these applications of the proposed  model indeed obtain promising results.

\end{abstract}

\section{Introduction}
\label{sec:intro}
Snakes or active contour models have been studied considerably, and used for object segmentation and feature extraction for almost  three decades, since the pioneering work of the snakes model proposed by \citet{kass1988snakes}. A snake is a parametrized curve $\Gamma$ (locally) that minimizes the energy:
\begin{equation*}
E(\Gamma)=\int_0^1 \left(w_1\,\|\Gamma^\prime(t)\|^2 +w_2\,\|\Gamma''(t)\|^2+ P\big(\Gamma(t)\big)\right)dt
\end{equation*}
with appropriate boundary conditions at the endpoints $t = 0$ and $t=1$. $\Gamma^\prime$ and $\Gamma''$ are the first- and second-order derivatives of the curve $\Gamma$ respectively. The positive constants  $w_1$ and $w_2$ relate to the elasticity  and rigidity of the curve $\Gamma$ and, hence, weight its internal force. This approach models contours as  curves $\Gamma$ locally minimizing an objective energy functional $E$ that consists of an internal and an external force. The internal force terms depend on the first- and second-order derivatives of the curves (snakes), and, respectively, account for a prior of small length and of low curvature of the contours. The external force is derived from a potential function $P$,  which depends on image features such as the  gradient magnitude, and it is designed to attract the active contours or snakes to the image features of interest such as object boundaries.

The drawbacks of the classical active contours or snakes model~\citep{kass1988snakes} are its sensitivity to initialization, the difficulty of handling topological changes, and the difficulty of minimizing the strongly non-convex path energy as discussed by~\citet{cohen1997global}.  Regarding initialization, the  active contours model requires an initial guess that is close to the desired image features, and preferably enclosing them because energy minimization tends to shorten the snakes. The introduction of an expanding balloon force allows the model to be less demanding on the initial curve given inside the objective region \citep{cohen1991active}.  The issue of topology changes leads, on the other hand, to the development of active contour methods, which represent object boundaries as the zero level set of the solution to a PDE \citep{osher1988fronts, caselles1993geometric, malladi1994evolutionary,caselles1997geodesic,yezzi1997geometric}. 

The difficulty of minimizing the non-convex snakes energy \citep{kass1988snakes} leads to important practical problems because the curve optimization procedure often becomes stuck at a local minimum of the energy functional, which makes the results sensitive to curve initialization and image noise. This limitation is still the case for the level set approach on  geodesic active contours~\citep{malladi1995shape,caselles1997geodesic}. To address this local minimum sensitivity issue,~\citet{cohen1997global} proposed an Eikonal PDE-based minimal path model to find the global minimum of the geodesic active contours energy~\citep{caselles1997geodesic},  in which the penalty associated to the second-order derivative of the curve was removed from the snakes energy.  Thanks to this simplification, a fast, reliable and globally optimal numerical method allows to find the energy minimizing curve with prescribed endpoints, namely, the fast marching method \citep{sethian1999fast}, which is based on the formalism of viscosity solutions to Eikonal PDE.  These mathematical and algorithmic guarantees of Cohen and Kimmel's minimal path model~\citep{cohen1997global} have important practical consequences, leading to various approaches for image analysis and medical imaging~\citep{peyre2010geodesic,cohen2001multiple}.

In the basic formulations of the minimal paths-based interactive image segmentation models \citep{appleton2005globally,appia2011active,mille2014combination}, the common proposal is that the object boundaries can be delineated  by a set of minimal paths constrained by  user-provided points. 
In~\citep{li2007vessels,benmansour2011tubular}, vessels  were extracted under the form of minimal paths over the radius-lifted space. Therefore, each extracted minimal path consists of both the centerline positions and the corresponding thickness values of a vessel.

To reduce the user intervention, \citet{benmansour2009fast} proposed a  growing minimal path model for object segmentation with unknown endpoints. This model can recursively detect keypoints, each of which can be used as a new source point for the fast marching front propagation. Thus this model requires only one user-provided point to start the keypoints detection procedure. \citet{kaul2012detecting} applied the growing minimal path model to detect complicated curves with arbitrary topologies and  developed  criteria to stop the growth of the minimal paths.  \citet{rouchdy2013geodesic} proposed a geodesic voting model for  vessel tree extraction by a voting score map that is constructed from a set of geodesics with  a common source point. 

Recently, improvements of the minimal path model have been devoted to extend the isotropic Riemannian metrics  to the more general anisotropic Riemannian metrics by taking into account the orientation of the curves~\citep{bougleux2008anisotropic,jbabdi2008accurate,benmansour2011tubular}. Such orientation enhancement can solve some shortcuts problems suffered by the isotropic minimal path models~\citep{cohen1997global,li2007vessels}. \citet{kimmel2001optimal} designed an orientation-lifted Riemannian metric for the application of path planning,  providing an alternative way to take advantage of the orientation information. This  isotropic orientation-lifted Riemannian metric~\citep{kimmel2001optimal} was built over a higher dimensional domain by adding  an extra  orientation space to the image domain.

\citet{bekkers2015pde} considered a data-driven extension of the  sub-Riemannian metric on SE(2), which shows that the sub-Riemannian structure outperforms the isotropic Riemannian structures on SE(2) in  retinal vessel tracking.  The numerical solver used by~\citet{bekkers2015pde} is based on a PDE approach with an upwind discretization and iterative evolution scheme, which requires expensive computation time.
To solve this problem, \citet{sanguinetti2015sub} used the  fast marching method~\citep{mirebeau2014anisotropic}  as the Eikonal solver  to track the sub-Riemannian geodesics.
The sub-Riemannian geodesic model~\citep{Petitot:2003im} reintroduced curvature penalization to the geodesic energy, similar to the Euler elastica bending energy~\citep{nitzberg19902} considered in this paper, yet it differs in two ways: firstly, the Euler elastica bending energy involves the squared curvature, a stronger penalization than Petitot's sub-Riemannian geodesic energy which is roughly linear in the curvature. Secondly, minimal geodesics for Petitot's sub-Riemannian model occasionally feature \emph{cusps}, which sometimes may not be desirable for the applications of interest. In contrast, Euler elastica curves can  avoid such cusps.

\begin{figure*}[!t]
\centering
\includegraphics[width=17cm]{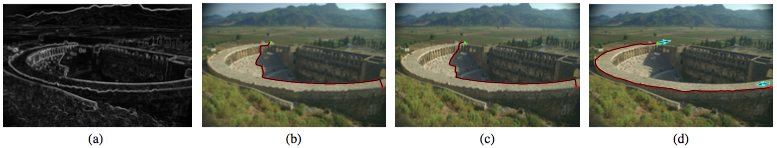}\,
\caption{(\textbf{a}) Edge saliency map. (\textbf{b}) Minimal path with isotropic Riemannian metric. (\textbf{c}) Minimal path with anisotropic Riemannian metric. (\textbf{d}) Minimal path with the proposed Finsler elastica metric. The red curves are the extracted minimal paths with the initial source positions and end positions indicated by the red  and green dots respectively. The arrows in (\textbf{d}) indicate the tangents.} 
\label{fig:Example}
\end{figure*}

\citet{schoenemann2012linear} proposed a model to address  the problems of using curvature regularization for region-based image segmentation by a graph-based combinational optimization method. This curvature regularization model can find  a solution that corresponds to the  globally optimal segmentation which is initialization-independent, and has proven to  obtain promising segmentation and inpainting results  especially for objects with long and thin structures.
\citet{ulen2015shortest} proposed a curvature and torsion regularized shortest path model for tubular structure segmentation, where the  curvature and the torsion were approximately computed by B-Splines. The solution of their energy functional including  curvature and torsion penalization terms can be efficiently obtained by using line graphs and Dijkstra's algorithm~\citep{dijkstra1959note}.

\citet{tai2011fast} presented an efficient method to solve  the minimization problem of the Euler elastica  energy, and demonstrated that this fast method can be applied  to  image denosing, inpainting, and zooming. Other approaches of interest using the curvature penalization include the  image segmentation models such as~\citep{elzehiry2010fast,schoenemann2011elasticratio,zhu2013image} and the image  inpainting model introduced by~\citet{shen2003euler}.

\subsection{Motivation}
In contrast with the classical snakes energy~\citep{kass1988snakes}, Eikonal PDE-based minimal path methods are first-order models, which do not penalize the second-order derivative of a curve, i.e., the curvature, and therefore do not enforce the smoothness  of the geodesic, leading sometimes to undesired result as shown in Fig.~\ref{fig:Example}, in which we would like to extract a boundary that is as smooth as possible between the two given points indicated by red and green dots.  In Fig.~\ref{fig:Example}a we show the edge saliency map. Figs.~\ref{fig:Example}b and \ref{fig:Example}c are the minimal paths obtained by using the isotropic Riemannian metric \citep{cohen1997global} and the anisotropic Riemannian metric \citep{bougleux2008anisotropic}, respectively,   both of which are unable to  find expected smooth boundaries and suffer from the shortcut problem due to the lack of curvature penalization in these metrics. In contrast,  the minimal path model presented in this paper reintroduces the curvature, in the form of weighted Euler elastica curves as studied in \citep{nitzberg19902, mumford1994elastica}. Therefore, the geodesics extracted by the proposed metric can catch the smooth object boundary, as shown in Fig.~\ref{fig:Example}d,  with arrows indicating the corresponding tangents at the given positions.

\subsection{Contributions}
\label{subsec:Contribution}
The contribution of this paper is three fold:
\begin{enumerate}
	\item Firstly, we propose a novel globally minimized minimal path model, namely, the Finsler elastica minimal path model,  with curvature penalization and  Finsler metric. We establish the connection between the Euler elastica curves and the minimal paths with respect to a Finsler elastica metric.
With an adequate numerical implementation, leveraging orientation lifting, asymmetric Finsler metrics and anisotropic fast marching method, the proposed model still allows to find the globally minimizing curves with prescribed points and tangents.  

    \item As a second contribution,  we present the mathematical convergence analysis of the proposed Finsler elastica metrics, and of the associated Finsler elastica minimal paths.
Furthermore, we discuss numerical options for geodesic distance and minimal paths computation, and settle for an adaptation of the fast marching method proposed by \citet{mirebeau2014efficient}.
     
	\item Finally, we  provide two types of  image data-driven speed functions that are computed  by steerable filters. These speed functions  are therefore orientation dependent, by which we apply the proposed Finsler elastica minimal path model to interactive closed contour detection, perceptual grouping and tubular structure extraction.

\end{enumerate}
More contributions have been added regarding the original conference paper presented in \citep{chen2015global},  such as the applications of interactive closed contour detection for image segmentation and perceptual grouping, respectively.

\subsection{Paper Outline}
In Section~\ref{sec:background} we briefly introduce the existing  minimal path models, the concept of Finsler metric, and algorithms for distance computation and path extraction. The relationship between the Euler elastica bending energy and the Finsler elasica metric is analyzed in Section~\ref{sec:FinslerMinimaPathsModel}. In Section \ref{sec:Velocity} we introduce two data-driven speed functions which are dependent of orientations. The  applications  of the Finsler elastica minimal path model are presented in Section \ref{sec:ContourDetection}.  Experiments and  Conclusion are presented in Sections \ref{sec:Experiments} and \ref{sec:Conclusion}, respectively.

\section{Background on the Minimal Path Models}
\label{sec:background}
\subsection{Cohen-Kimmel Minimal Path Model}
\label{subsec:TypicalMetrics}
Let $\Omega\subset\bR^n$ ($n=2, 3$) denote the image domain.
The classical Cohen-Kimmel minimal path model~\citep{cohen1997global} was designed to  find the global minimum  of  the geodesic energy which measures the  length of Lipschitz  paths $\gamma:[0,1]\to\Omega$ as follows
\begin{equation}
\label{eq:geodesicCohen}
\cL^{{\rm I}}(\gamma)=\int_0^1\Big(w+P\big(\gamma(t)\big)\Big)\|\gamma^\prime(t)\|\,dt,
\end{equation}
where $\|\cdot\|$ denotes the canonical Euclidean norm. $w$ is a positive constant and $P:\Omega\to\bR^+$ is a potential function which  usually depends on the image gradient magnitudes  or intensities as suggested by~\citet{cohen1997global}.

A geodesic $\cC_{\s,\x}$, joining the source point $\s$ to a point $\x\in\Omega$,  is a path that globally minimizes the length~\eqref{eq:geodesicCohen}: 
 \begin{equation*}
  \cC_{\s,\x}=\arg\min_{\gamma\in\cA_{\s,\x}}\{\cL^{\rm I}(\gamma)\},	
 \end{equation*}
where $\cA_{\s,\x}$ is the collection of all Lipschitz  paths $\gamma$ with  $\gamma(0)=\s$ and $\gamma(1)=\x$.

The minimal action map $\cU:\Omega\to\bR^+$ associated to a source point $\s$, is the global minimum  of the length~\eqref{eq:geodesicCohen} among all paths involved in the collection $\cA_{\s,\x}$:
\begin{equation}
\label{eq:geodesicDistance_IR}
\cU(\x):=\inf\{\cL^{\rm I}(\gamma);\,\gamma\in\cA_{\s,\x}\},\quad \forall\x\in\Omega,
\end{equation}

The Cohen-Kimmel minimal path model considers  an isotropic Riemannian metric,  which is defined by
\begin{equation}
\label{eq:CKMetric}
\cF^{\rm I}(\x,\fu):=\sqrt{\langle\fu,\cM_{\rm I}(\x)\,\fu\rangle},\quad \forall \x\in\Omega,\,\forall \fu\in\bR^n, 
\end{equation}
where $\langle\cdot,\cdot\rangle$ denotes the scalar product over $\bR^n$ and  $\cM_{\rm I}$ is a  symmetric positive definite tensor field which is proportional to the identity matrix $\Id$:
\begin{equation}
\label{eq:CKTensor}
\cM_{\rm I}(\x)=\big(w+P(\x)\big)^2\Id,	\quad \forall \x \in\Omega.
\end{equation}
The minimal action map $\cU$ defined in~\eqref{eq:geodesicDistance_IR} satisfies the following Eikonal PDE~\citep{cohen1997global}:
\begin{equation}
\label{eq:EikonalPDE_IR}
\begin{cases}
\|\nabla \cU(\x)\|=w+P(\x),\quad \forall \x\in\Omega\backslash\{\s\},\\
\cU(\s)=0,	
\end{cases}	
\end{equation}
where $\nabla \cU$ is the Euclidean gradient of $\cU$ with respect to the position in the domain. 

A geodesic $\hat\cC_{\x,\s}$ parameterized by its arc-length can be obtained through the gradient descent ordinary differential equation (ODE) till reaching the point $\s$:
\begin{equation}
\label{eq:GeodesicODE_IR}
\begin{cases}
\hat\cC_{\x,\s}^\prime(t)\propto-\nabla\cU(\hat\cC_{\x,\s}(t)),\\	
\hat\cC_{\x,\s}(0)=\x,
\end{cases}
\end{equation}
where $\propto$ denotes positive collinearity. Then the geodesic $\cC_{\s,\x}$ is obtained by reversing the path $\hat\cC_{\x,\s}$ with $\cC_{\s,\x}(0)=\s$ and $\cC_{\s,\x}(1)=\x$.

\subsection{Anisotropic Riemannian Metric Extension}
\label{subsec:AR}
\citet{bougleux2008anisotropic} and~\citet{jbabdi2008accurate}   extended  the  isotropic Riemannian metric~\eqref{eq:CKMetric} to the anisotropic Riemannian case  invoking a symmetric positive definite tensor field
\begin{equation}
\label{eq:ARTensor}
\cM_{\rm A}(\x)=\sum_{i=1}^n P_i(\x)\,\fv_i(\x)\,\fv^{\rm T}_i(\x),\quad \forall\,\x\in\Omega,	
\end{equation} 
where $n$ is the dimension of the domain $\Omega$ and the vector fields $\fv_i$ are related to the image data. $P_i:\Omega\to\bR^+$  are the image data-driven  functions associated to $\fv_i$.

Based on the tensor field $\cM_{\rm A}$~\eqref{eq:ARTensor}, the anisotropic Riemannian metric $\cF^{\rm A}$ can be expressed as:
\begin{equation}
\label{eq:ARMetric}
\cF^{\rm A}(\x,\fu):=\sqrt{\langle\fu,\cM_{\rm A}(\x)\,\fu\rangle},\quad \forall \x\in\Omega,\,\forall \fu\in\bR^n.
\end{equation} 
The minimal action map $\cU$ associated to the anisotropic Riemannian metric $\cF^{\rm A} $ satisfies the Eikonal PDE:
\begin{equation}
\label{eq:EikonalPDE_AR}
\begin{cases}
\|\nabla \cU(\x)\|_{\cM_{\rm A}^{-1}(\x)}=1,\quad \forall\,\x\in\Omega\backslash\{\s\},\\
\cU(\s)=0,	
\end{cases}	
\end{equation}
where  we denote that $\|\fu\|_{\cM}=\sqrt{\langle\fu,\cM\,\fu\rangle}$. 

The geodesic $\cC_{\s,\x}$ is obtained by reversing the geodesic $\hat\cC_{\x,\s}$ which is  the solution  to the following  ODE
\begin{equation}
\label{eq:GeodesicOED_AR}
\begin{cases}
\hat\cC^\prime_{\x,\s}(t)\propto -\cM_{\rm A}^{-1}(\hat\cC_{\x,\s}(t))\nabla\cU(\hat\cC_{\x,\s}(t)),\\
\hat\cC_{\x,\s}(0)=\x.	
\end{cases}
\end{equation}

\citet{benmansour2011tubular} proposed an anisotropic radius-lifted Riemannian metic,    as a special case of the metric $\cF^{\rm A}$, to address the problem of extracting both the centerline positions and the thickness values of  a tubular structure simultaneously. 

\subsection{Isotropic orientation-lifted Riemannian Metric Extension}
\label{subsec:IOL}

In order to take into account the local orientation in the image, it is possible to include orientation information in the energy minimization. For this purpose, the image domain $\Omega\subset\bR^2$  can be extended to an orientation-lifted space $\bar\Omega$ by product with an abstract orientation space $\bS^1$~\citep{kimmel2001optimal}, i.e., $\bar\Omega=\Omega\times\bS^1\subset \bR^3$ and the problem is to find a geodesic in the new lifted space $\bar\Omega$. Each point $\bx$  in the orientation-lifted path is thus a pair composed of a point $\x$ in the image domain $\Omega$ and an orientation $\theta$, i.e., $\bx=(\x,\theta)$.

For any point $\bx$ and any vector $\bu=(\fu,\nu)\in\bR^2\times\bR$, where $\nu$ denotes the orientation variation, the isotropic orientation-lifted Riemannian metric $\cF^{\rm IL}$ is expressed by
\begin{equation}
\label{eq:IsotropicLiftedMetric}
\cF^{\rm IL}(\bx,\bu):=\frac{1}{\Phi_{\rm IL}(\bx)}\sqrt{\|\fu\|^2+\rho|\nu|^2},
\end{equation}
where $\Phi_{\rm IL}:\bar\Omega\to\bR^+$ is an orientation-dependent speed function relying on the image data and $\rho>0$ is a constant.

Following  the form of \eqref{eq:CKMetric}, the metric $\cF^{\rm IL}$ defined in~\eqref{eq:IsotropicLiftedMetric} can be written by using a tensor field $\cM_{\rm IL}$
\begin{equation}
\label{eq:Tensor_IOL}
\cM_{\rm IL}(\bx)=\frac{1}{\Phi_{\rm IL}^{2}(\bx)}\,\mathbf D_\rho,
\end{equation}
where $\mathbf D_\rho=\text{diag}(1,1,\rho)$ is a diagonal matrix. Note that $\rho$ can be extended to a positive scalar function that is dependent of the image data. For simplicity, we set $\rho$ to be a positive  constant in this paper, as suggested in \citep{kimmel2001optimal,pechaud2009extraction}.

The curve length of  an orientation-lifted curve $\gamma:=(\Gamma,\theta):[0,1]\to\bar\Omega$ with respect to  the isotropic orientation-lifted Riemannian metric $\cF^{\rm IL}$  can be expressed as
\begin{align}
\label{eq:geodesicIL}
\cL^{\rm IL}(\gamma)&=\int_0^1\cF^{\rm IL}(\gamma(t),\,\gamma^\prime(t))\,dt\nonumber\\
&=\int_0^1\frac{1}{\Phi_{\rm IL}(\gamma(t))}\sqrt{\langle \gamma^\prime(t),\,\cM_{\rm IL}(\gamma(t))\,\gamma^\prime(t)\rangle}\,dt\nonumber\\
&=\int_0^1\frac{1}{\Phi_{\rm IL}(\gamma(t))}\sqrt{\|\Gamma^\prime(t)\|^2+\rho|\theta^\prime(t)|^2}\,dt,
\end{align}
where $\gamma^\prime(t)=(\Gamma^\prime(t),\theta^\prime(t))\in\bR^3$, for all $t\in[0,1]$.

The  idea of orientation lifting was applied  to  tubular structure extraction by~\citet{pechaud2009extraction} with additional radius lifting, typically accounting for the radius  of a  tubular structure in the processed image.  Orientation lifting often improves the results from~\citep{cohen1997global,li2007vessels}, but suffers from the fact that for the curve length of an orientation-lifted curve $\gamma(t)=(\Gamma(t),\theta(t))$, nothing in~\eqref{eq:geodesicIL} constrains the path direction $\Gamma^\prime(t)$ to align with the orientation vector associated to  $\theta(t)$, where $t\in[0,1]$. In other words, this isotropic orientation-lifted Riemannian metric $\cF^{\rm IL}$ cannot penalize the curvature of the geodesic, a point which is addressed in this paper.

\subsection{General  Minimal Path Model and Finsler Metric} 
\label{subsec:GeneralMinimalPath}
The general minimal path problem is posed on a bounded  domain $\Omega\subset\bR^n$ equipped with a  metric $\cF(\x,\fu)$ depending on  positions $\x\in\Omega$ and orientations $\fu\in\bR^n$. This metric $\cF$ defines at each point $\x$ an asymmetric norm
\begin{equation}
\cF_\x(\fu):=\cF(\x,\fu).
\end{equation}
These norms must be positive $\cF_\x(\fu) > 0$ whenever $\fu\neq \0$, $1$-homogeneous, and obey the triangular inequality. However,  in general we allow them to be asymmetric: $\exists\fu\in\bR^n$ such that 
\begin{equation}
\cF_\x(\fu)\neq  \cF_\x(-\fu).
\end{equation}
One can  measure the curve length of a regular curve $\gamma$ with respect to  the metric $\cF$:
\begin{equation}
\label{eq:pathLength}
\cL_\cF(\gamma)=\int_0^1\cF\big(\gamma(t),\,\gamma^\prime(t)\big)	\, dt,
\end{equation}
The minimal action map $\cU(\x)$ is defined by:
\begin{equation}
\label{eq:geodesicDistance}
\cU(\x):=\inf\{\cL_\cF(\gamma);\,\gamma\in\cA_{\s,\x}\},
\end{equation} 
which is the unique viscosity solution to an Eikonal PDE~\citep{lions1982generalized,sethian2003ordered}:
\begin{equation}
\label{eq:EikonalPDE}
\begin{cases}
\cF^{*}_\x\big(\nabla \cU(\x)\big)=1, \forall\x\in\Omega\backslash\{\s\},\\
\cU(\s)=0,
\end{cases}	
\end{equation}
where $\cF_\x^*$ is  the dual norm of $\cF_\x$  defined for all $\fu\in\bR^n$ by
\begin{equation}
\label{eq:dualNorm}
\cF^{*}_\x(\fu):=\sup_{\fv\neq\0}\frac{\langle\fu, \fv\rangle}{\cF_\x(\fv)}.
\end{equation}
Based on the definition of the dual norm in \eqref{eq:dualNorm}, the corresponding optimal direction map $\Psi$ is then obtained by
\begin{equation}
\label{eq:OptimalDirection}
\Psi(\x,\fu):=\arg\max_{\fv\neq \0}\frac{\langle\fu,\fv\rangle}{\cF_{\x}(\fv)},\quad \forall\,\x\in\Omega,\,\forall \fu\in\bR^n.
\end{equation}
Again, the geodesic $\cC_{\s,\x}$ is obtained by reversing the geodesic $\hat\cC_{\x,\s}$ with $\cC_{\s,\x}(0)=\s$ and $\cC_{\s,\x}(1)=\x$, where $\hat\cC_{\x,\s}$ is tracked through the following ODE involving the minimal action map $\cU$ and the optimal direction map $\Psi$
\begin{equation}
\label{eq:geodesicODE}
\begin{cases}
\hat\cC_{\x,\s}^\prime(t)\propto -\Psi\Big(\hat\cC_{\x,\s}(t),\nabla\cU\big(\hat\cC_{\x,\s}(t)\big)\Big),\\
\hat\cC_{\x,\s}(0)=\x.
\end{cases}
\end{equation}
Numerically, the ODE expressed in \eqref{eq:geodesicODE} is solved by using  Heun’s or Runge-Kutta's methods, or more robustly using the numerical method proposed by~\citet{mirebeau2014anisotropic}. 

The metric $\cF$ considered in this paper combines a symmetric part, defined in terms of a symmetric  positive definite tensor field $\cM$, and an asymmetric part involving a vector field $\vec\omega$:
\begin{equation}
\label{eq:generalFinslerMetric}
\cF(\x, \fu):=\sqrt{\langle\fu,\cM(\x)\,\fu\rangle}-\langle\vec\omega(\x),\,\fu\rangle,
\end{equation}
for all $\x\in\Omega$ and any vector $\fu\in\bR^n$.
The fields $\cM$ and $\vec\omega$ should obey the following \emph{smallness} condition to ensure that the Finsler metric $\cF$ is positive: 
\begin{equation}
\label{eq:constraintFinsler}
 \langle\vec\omega(\x), \cM^{-1}(\x)\,\vec\omega(\x)\rangle < 1,\quad\forall\x\in\Omega.
\end{equation}
Equation~\eqref{eq:generalFinslerMetric} defines an anisotropic Finsler metric in general. This is an anisotropic Riemannian metric if the vector field $\vec\omega$ is identically zero, and an isotropic metric if in addition the tensor field $\cM$ is proportional to a diagonal matrix like $\mathbf D_\rho$ in~\eqref{eq:Tensor_IOL}.
Therefore, the general Eikonal PDE~\eqref{eq:EikonalPDE} and the geodesic back tracking ODE~\eqref{eq:geodesicODE} reduce respectively to equations~\eqref{eq:EikonalPDE_IR} and~\eqref{eq:GeodesicODE_IR} for the isotropic case, and respectively to  equations~\eqref{eq:EikonalPDE_AR} and~\eqref{eq:GeodesicOED_AR} for the anisotropic case.

\subsection{Computation of the Minimal Action Map} 
\label{subsec:fastmarching}

In order to estimate the minimal action map $\cU$, presented in~\eqref{eq:geodesicDistance} and \eqref{eq:EikonalPDE},  a discretization grid $Z$ of the image domain $\Omega$ is introduced - or of the extended domain $\bar\Omega$ in the case of an orientation-lifted metric. For each point $\x\in Z$, a small mesh $S(\x)$ of a neighbourhood of $\x$ with vertices in $Z$ is constructed. 
For example,  $S(\x)$ can be the square formed by the four neighbours of  $\x$ in the classical fast marching method~\citep{sethian1999fast} on a regular orthogonal grid used in~\citep{cohen1997global}. In contrast with  Sethian's classical fast marching method which solves the discrete approximation of the Eikonal PDE itself, an approximation of the action map $\cU$, with  the initial source point $\s$, is calculated by solving the  fixed point system~\citep{tsitsiklis1995efficient}:
\begin{equation}
\label{eq:FixedPoint}
	\begin{cases}
		\cU(\x) = \Lambda \cU(\x), \quad \text{for all } \x \in Z\backslash\{\s\},\\
		\cU(\s) = 0, 
	\end{cases}
\end{equation}
where the  involved Hopf-Lax update operator is defined by:
\begin{equation}
\label{eq:Hopf-Lax}
\Lambda \cU(\x):=\min_{\y\in\partial S(\x)}\Big\{\cF(\x,\x-\y)+\interp_{S(\x)}\cU(\y)\Big\},
\end{equation}
where $\interp_{S(\x)}$ denotes the piecewise linear interpolation operator on the mesh $S(\x)$, and $\y$ lies on the boundary of $S(\x)$. 
The equality $\cU(\x) = \Lambda \cU(\x)$, replacing in~\eqref{eq:FixedPoint} the Eikonal PDE: $\cF_\x(\nabla \cU(\x))=1$ of~\eqref{eq:EikonalPDE}, is a discretization of Bellman's optimality principle, which is similar in spirit to the Tsitsiklis approach~\citep{tsitsiklis1995efficient}. 
It reflects the fact that the minimal geodesic $\cC_{\s,\x}$, from $\s$ to $\x$, has to cross the mesh boundary $\partial S(\x)$ at least once at some point $\y$; therefore it is the concatenation of a geodesic $\cC_{\s,\y}$ from $\s$ to $\y$, which length is approximated by piecewise linear interpolation, and a very short geodesic $\cC_{\y,\x}$ from $\y$ to $\x$, approximated by a segment of geodesic curve length $\cF(\x,\x-\y)$. 
The $N$-dimensional fixed point system \eqref{eq:FixedPoint}, with $N=\#Z$ the number of grid points, can be solved by a single pass fast marching method as described in Algorithm~\ref{alg:FM} in Appendix.

The classical fast marching methods~\citep{sethian1999fast,tsitsiklis1995efficient} using the square formed neighbourhood $S$ have difficulty to  deal with the computation of geodesic distance maps with respect to anisotropic metrics, especially when the anisotropy gets large.  An adaptive construction method of such stencils $S$ was introduced in~\citep{mirebeau2014anisotropic} for  anisotropic 3D Riemannian metric, and in \citep{mirebeau2014efficient} for anisotropic 2D Finsler metric, providing that the stencils or mesh $S(\x)$ at each point $\x\in\Omega$ or $\bar\Omega$ satisfies some geometric acuteness property depending on the local metric $\cF(\x,\cdot)$. Such adaptive stencils based fast marching methods lead to breakthrough improvements in terms of computation time and accuracy for strongly anisotropic geodesic metrics.  When the above mentioned geometric properties do not hold, the fast marching method is in principle not applicable, and slower multiple pass methods must be used instead, such as the Adaptive Gauss Siedel Iteration (AGSI) of \citet{Bornemann:2006hc}. The present paper involves a 3D Finsler metric \eqref{eq:metricFinsler_A}, for which we constructed stencils by adapting the 2D Finsler metric construction method proposed by \citet{mirebeau2014efficient}. Although these stencils lack the geometric acuteness condition, we found that the fast marching method still provided good approximations of the paths, while vastly improving computation performance. In Section~\ref{subsec:implementations}, the complexity will be discussed more. 

Note that whenever we mention fast marching method in the next sections, we mean the fast marching method with adaptive stencils proposed by \citet{mirebeau2014efficient}.

\section{Finsler Elastica Minimal Path Model}
\label{sec:FinslerMinimaPathsModel}
In this section, we present the core contribution of this paper: the orientation-lifted Finsler metric embedding curvature penalty term, defined  over the orientation-lifted domain $\bar\Omega=\Omega\times\bS^1\subset\bR^3$, where $\bS^1=[0,2\pi$) denotes the angle space with periodic boundary conditions and $\Omega\subset\bR^2$ denotes the image domain.

\begin{figure*}[t]
\centering
\subfigure[]{\includegraphics[width=8cm]{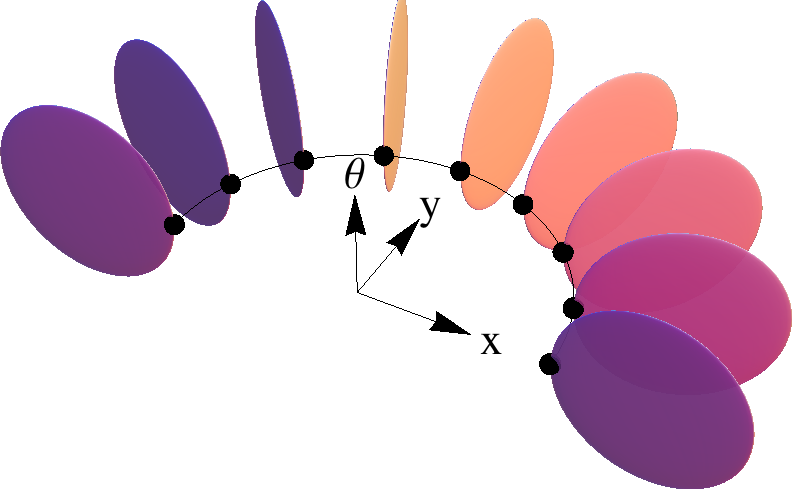}}
\subfigure[]{\includegraphics[width=8cm]{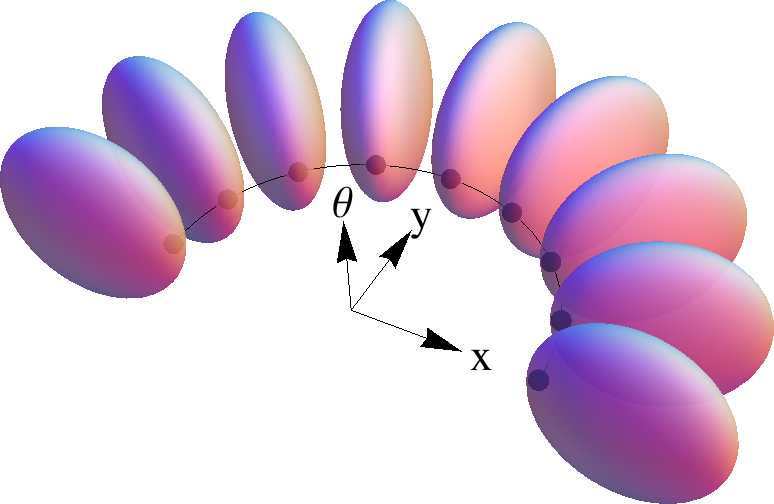}}
\caption{Visualization for the metrics 
$\cF^\infty$ and $\cF^\lambda$ with $\alpha=1$ by Tissot's indicatrix. (\textbf{a}) Tissot's indicatrix for the metric $\cF^\infty$~\eqref{eq:metricFinsler_A} with $\alpha=1$ are flat 2D disks embedded in 3D space, aligned with the direction $\vt$ (several directions $\theta$ are shown). (\textbf{b}) Tissot's indicatrix for the Finsler elastica metrics $\cF^\lambda$ are ellipsoids, which flatten and approximate a limit disk as $\lambda \to \infty$.}
\label{fig:Tissot}
\end{figure*}

\subsection{Geodesic Energy Interpretation of the Euler Elastica Bending Energy  via a Finsler Metric}
\label{sec:GeodesicInterpretation}

The Euler elastica curves were  introduced to  the field of computer vision by~\citet{nitzberg19902} and~\citet{mumford1994elastica}. They minimized the following bending energy:
\begin{equation}
\label{eq:DataBendEnergy}
\cL(\hat\Gamma):=\int_0^L\frac{1}{\Phi_0\big(\hat\Gamma(s)\big)}\big(1+\alpha\kappa^2(s)\big)\,ds,
\end{equation}
where $\hat\Gamma:[0,L]\to\Omega$ is a regular curve with non-vanishing velocity, $\kappa$ is the curvature of curve $\hat\Gamma$,  $L$ is the  Euclidean curve length of $\hat\Gamma$,  and  $s$ is the arc-length.   Parameter $\alpha>0$ is a  constant.  $\Phi_0$ is an image data-driven speed function.

For the sake of simplicity,  we set  $\Phi_0\equiv 1$, yielding the  simplified Euler elastica bending energy
\begin{equation}
\label{eq:EulerElastica}
\ell(\hat\Gamma)=\int_0^L\big(1+\alpha\kappa^2(s)\big)\,ds,
\end{equation}  
where the general case will be studied in Section~\ref{subsec:weightedFinsler}.

The goal of this section is to cast the Euler  elastica bending energy $\ell$~\eqref{eq:EulerElastica} in the form of curve length with respect to a relevant asymmetric Finsler metric.  
We firstly transform the elastica problem  into finding a geodesic in an orientation-lifted space. For this purpose,  we denote by 
\begin{equation}
\label{eq:orienVector}
\vt=(\cos\theta,\sin\theta)	
\end{equation}
the  unit vector corresponding to $\theta \in \bS^1$. 

Let $\Gamma:[0,1]\to\Omega$ be a regular curve with non-vanishing velocity and  $\gamma=(\Gamma,\zeta):[0,1]\to\bar\Omega$ be its canonical orientation lifting. For any $t\in[0,1]$, $\zeta(t)$ is defined as being such that:
\begin{equation}
\label{eq:Assumption}
\fv_{\zeta(t)}: = \frac{\Gamma^\prime(t) }{ \| \Gamma^\prime(t)\|}.
\end{equation} 
According to the definition of $\fv_\zeta$ in \eqref{eq:orienVector}, one has
\begin{align}
\label{eq:OrtogonalVector}
\left(\frac{\Gamma^\prime(t)}{\|\Gamma^\prime(t)\|}\right)^\perp&= (\fv_{\zeta(t)})^\perp\nonumber\\
&=(-\sin(\zeta(t)),\cos(\zeta(t))),	
\end{align}
where $\fu^\perp$ denotes  the  vector that is perpendicular to a vector $\fu$. It is known that
\begin{equation}
\label{eq:CurvaEquation}
\frac{d}{dt}\left(\frac{\Gamma^\prime(t)}{\|\Gamma^\prime(t)\|}\right) = \kappa(t) \|\Gamma^\prime(t)\|\left(\frac{\Gamma^\prime(t)}{\|\Gamma'(t)\|}\right)^\perp,
\end{equation}
where $\kappa$ is the curvature of path $\Gamma$. Then we have the following equations:
\begin{align*}
\frac{d}{dt}\fv_{\zeta(t)}&=\frac{d}{dt}(\cos(\zeta(t)),\sin(\zeta(t)))\\
&=\zeta^\prime(t)\big(-\sin(\zeta(t)),\cos(\zeta(t))\big)\\
&=\zeta^\prime(t)\left(\frac{\Gamma^\prime(t)}{\|\Gamma^\prime(t)\|}\right)^\perp.
\end{align*}
Thus, using \eqref{eq:Assumption} and \eqref{eq:CurvaEquation}, we have
\begin{equation}
\label{eq:KeyStepCurva}
\zeta^\prime(t)\left(\frac{\Gamma^\prime(t)}{\|\Gamma'(t)\|}\right)^\perp=\kappa(t) \|\Gamma^\prime(t)\|\left(\frac{\Gamma^\prime(t)}{\|\Gamma'(t)\|}\right)^\perp,
\end{equation}
which yields to 
\begin{equation}
\label{eq:Curveture}
\zeta^\prime(t) = \kappa(t) \|\Gamma^\prime(t)\|,\quad \forall t\in[0,1].
\end{equation}
Using equations \eqref{eq:EulerElastica}, \eqref{eq:Assumption} and \eqref{eq:Curveture},  one has
\begin{align}
\label{eq:InduceRelation}
\ell(\Gamma)&=\int_0^L\big(1+\alpha\kappa^2(s)\big)\,ds\nonumber\\
&=\int_0^1\left(1+\alpha\,\frac{|\zeta^\prime(t)|^2}{\|\Gamma^\prime(t)\|^2}\right)\|\Gamma^\prime(t)\|\,dt\nonumber\\
&=\int_0^1\left(\|\Gamma^\prime(t)\|+\alpha\,\frac{|\zeta^\prime(t)|^2}{\|\Gamma^\prime(t)\|} \right)\,dt,
\end{align}
where the  Euclidean arc-length is defined as 
\begin{equation*}
ds = \|\Gamma'(t)\| dt.
\end{equation*}
By the definition of $\gamma$, for any $t \in [0,1]$ we have $\gamma^\prime(t)=(\Gamma^\prime(t),\zeta^\prime(t))$ and 
\begin{equation}
\label{eq:InterpBendEnergy}
\ell(\Gamma)=\int_0^1 \cF^\infty\big(\gamma(t),\gamma^\prime(t)\big)\,dt,
\end{equation}
where we define the Finsler metric $\cF^\infty$ on the orientation-lifted domain $\bar \Omega$ by 
\begin{equation}
\label{eq:inftyFinsler}
\cF^\infty(\bx, \bu):=
\begin{cases}
\|\fu\|+\alpha\frac{|\nu|^2}{\|\fu\|},&\text{ if } \fu \propto \vt,\\
+\infty,&\text{otherwise}.
\end{cases}
\end{equation}
for any orientation-lifted point $\bx=(\x, \theta)\in\bar\Omega$,  any vector  $\bu=(\fu,\nu)\in \bR^2\times\bR$ in the tangent space, and where $\propto$ denotes positive collinearity. 
Note that any other lifting $\tilde\gamma(t) = (\Gamma(t), \tilde \zeta(t))$ of $\Gamma(t)$ would by construction of~\eqref{eq:inftyFinsler} have infinite energy, i.e., $\ell(\Gamma)=\infty$.

\subsection{$\lambda$ Penalized Asymmetric Finsler Elastica Metric $\cF^\lambda$}
\label{subsec:uniformFinsler}

The Finsler metric $\cF^\infty$~\eqref{eq:inftyFinsler} is too singular to compute the global minimum of $\ell$~\eqref{eq:EulerElastica} by directly applying the numerical algorithm such as the fast marching method~\citep{mirebeau2014efficient}.  Hence we introduce a family of orientation-lifted Finsler elastica metrics over the orientation-lifted domain $\bar\Omega$,  depending on a penalization parameter $\lambda\gg1$ as follows:
\begin{equation}
\label{eq:metricFinsler_A}
\cF^\lambda(\bx,\bu):=\sqrt{\lambda^2\|\fu\|^2+2\alpha\lambda|\nu|^2}-(\lambda-1)\langle\vt,\fu\rangle,	
\end{equation}
for any orientation-lifted point $\bx=(\x,\theta)\in\bar\Omega$ and any vector $\bu=(\fu,\nu)\in \bR^2\times\bR$, and where $\vt=(\cos\theta,\sin\theta)$ is the unit vector associated to $\theta$  which denotes the position of  $\bx$ in the orientation space $\bS^1$ . 

As $\lambda\to\infty$, the $\lambda$ penalized Finsler elastica metric $\cF^\lambda$ can be expressed as:
\begin{align}
\cF^\lambda(\bx,\bu)=&\sqrt{\lambda^2\|\fu\|^2+2\alpha\lambda|\nu|^2}-(\lambda-1)\langle\vt,\fu\rangle\nonumber\\
=& \lambda \|\fu\| \sqrt{1+ \alpha \frac {2|\nu|^2}{\lambda \|\fu\|^2} } - (\lambda-1) \langle\fv_\theta,\fu\rangle\nonumber\\
=& \lambda \|\fu\| (1+ \frac {\alpha|\nu|^2}{\lambda \|\fu\|^2} + \cO(\lambda^{-2}))-(\lambda-1) \langle\fv_\theta, \fu\rangle\nonumber\\
=&\|\fu\|+\frac{\alpha|\nu|^2}{\|\fu\|}+(\lambda-1)(\|\fu\|-\langle\fv_{\theta},\fu\rangle)\nonumber\\
&+\cO(\lambda^{-1})
\end{align}
The term $\|\fu\|-\langle\fv_{\theta},\fu\rangle$ will vanish if vector $\fu$ is positively proportional to $\fv_\theta$. Therefore, one has for any $\bx$ and any $\bu$
\begin{equation*}
\cF^\lambda(\bx,\bu)\to\cF^\infty(\bx,\bu),\quad  \text{as } \lambda\to\infty.
\end{equation*}
The  metric $\cF^\lambda$ \eqref{eq:metricFinsler_A} has precisely the required form formulated in \eqref{eq:generalFinslerMetric},  with tensor field $\cM:=\cM_{\rm F}$ as:
\begin{equation}
\label{eq:Finsler_M}
\cM_{\rm F}(\bx)={\rm diag}(\lambda^2,\,\lambda^2,\,2\alpha\lambda),
\end{equation}
and vector field $\vec\omega:=\vec\omega_{\rm F}$
\begin{equation}
\label{eq:Finsler_w}
\vec\omega_{\rm F}(\bx) = (\lambda-1) (\vt,0),
\end{equation}
for any $\bx=(\x,\theta)\in\bar\Omega$. Note that the definiteness constraint \eqref{eq:constraintFinsler} is satisfied: 
\begin{equation*}
\langle\vec\omega_{\rm F}(\bx), \cM_{\rm F}^{-1}(\bx)\, \vec \omega_{\rm F}(\bx)\rangle = (1-\lambda^{-1})^2 < 1,\quad  \forall\,\bx\in\bar\Omega.
\end{equation*}

The anisotropy ratio characterizes the distortion between different orientations induced by a metric. Letting  $\bx=(\x,\theta)$, $\bw=(\fw,\nu_{\fw})\in\bR^2\times\bR$ and $\bv=(\fv,\nu_{\fv})\in\bR^2\times\bR$, the anisotropy  ratio  $\mu(\cF^\lambda)$ of the Finsler elastica metric $\cF^\lambda$ \eqref{eq:metricFinsler_A} can be defined by:
\begin{equation}
\label{eq:anisotropicRatio}
\mu(\cF^\lambda):=\sup_{\bx\in\bar\Omega}\left\{\max_{\|\bw\|=\|\bv\|=1 }\left\{ \frac{\cF^\lambda_{\bx}(\bw)}{\cF^\lambda_{\bx}(\bv)}\right\}\right\},
\end{equation} 
where the norm $\cF^\lambda_{\bx}(\cdot)=\cF^\lambda(\bx,\cdot)$. As an example, for the  Finsler elastica metric $\cF^\lambda$~\eqref{eq:metricFinsler_A} with $\lambda\geq 2$ and $\alpha=1$, we can show that  $\mu(\cF^\lambda)$ in \eqref{eq:anisotropicRatio} is obtained by choosing $\bw=(-\vt,0)$ and $\bv=(\vt,0)$, so that $\mu(\cF^\lambda)= 2\lambda-1$.

Moreover, one can  define the \emph{physical} anisotropy ratio of the Finsler elastica metric $\cF^\lambda$ by replacing  by $\bw_s=(\fw,0)$ and $\bv_s=(\fv,0)$ the variables  $\bw$ and $\bv$ in \eqref{eq:anisotropicRatio}, respectively. In this case, for any $\alpha$, the physical anisotropy ratio is equal to $2\lambda-1$ and only depends on $\lambda$.

A crucial object for  studying and visualizing  the geometry distortion induced by a metric is Tissot's indicatrix defined as  the collection of unit balls in the tangent space. For point $\bx = (\x, \theta) \in \bar\Omega$ and $\lambda \in [1, \infty)$, we define the unit balls for metrics $\cF^\infty$  and $\cF^\lambda$ respectively by
\begin{equation}
\label{eq:Ball_infty}
B_{\bx}^\infty := \{\bu= (\fu, \nu)\in \bR^2 \times \bR;\,  \cF^\infty(\bx,\bu) \leq 1\},
\end{equation}
and
\begin{equation}
\label{eq:Ball_lambda}
B_{\bx}^\lambda := \{ \bu = (\fu, \nu)\in \bR^2 \times \bR;\,  \cF^\lambda(\bx,\bu) \leq 1\}.
\end{equation}
Then any vector  $\bu=(\fu,\nu)\in B_{\bx}^\infty$ can be characterized by
\begin{equation}
\label{eq:Character}
u_\perp = 0, \quad u_\parallel>0,	\quad \text{and }\quad u_\parallel+\alpha\,\frac{|\nu|^{2}}{u_\parallel} \leq 1,
\end{equation}
where we introduce  $u_\parallel$ and $u_\perp$ as follows:
\begin{equation*}
u_\parallel:= \langle\fu,\vt\rangle,\quad
u_\perp: = \langle\fu, \vt^\perp\rangle.
\end{equation*}
Using  \eqref{eq:Character}, one has
\begin{equation}
\label{eq:Charac_infty}   
\left(u_\parallel - \frac{1}{2}\right)^2 + \alpha\,|\nu|^{2} \leq \frac{1}{4}.
\end{equation}
Thus $B_{\bx}^\infty$ is a flat 2D ellipse embedded in the 3D tangent space, and containing the origin on its boundary. Particularly, when $\alpha=1$, $B_{\bx}^\infty$ turns to  a flat 2D disk of radius $1/2$  as shown in Fig.~\ref{fig:Tissot}a. 

On the other hand, when $\lambda<\infty$, a short computation shows that any vector $\bu=(\fu,\nu)\in B_{\bx}^\lambda$ is characterized by a quadratic equation
\begin{equation}
\label{eq:Charac_lambda}
\frac{\lambda}{2}\,u^{2}_\perp + a_\lambda \left(u_\parallel - \frac{b_\lambda}{2}\right)^2 +\alpha\,|\nu|^{2}\leq \frac{c_\lambda}{4},
\end{equation}
where $a_\lambda, b_\lambda, c_\lambda$ are all $1+\cO(1/\lambda)$. Hence $B_{\bx}^\lambda$ is an ellipsoid, for instance see Fig. \ref{fig:Tissot}b with $\alpha=1$, almost flat in the direction of $\fv_\theta^\perp$  due to the large factor $\lambda/2$, which converges to the flat disk $B_{\bx}^\infty$ in the Haussdorf distance as $\lambda \to \infty$.

Tissot's indicatrix is also the control set in the optimal control interpretation of the Eikonal PDE \eqref{eq:EikonalPDE}. The Haussdorf convergence of the control sets  guarantees that the minimal action map  and minimal paths for the metric $\cF^\lambda$ converge towards those of $\cF^\infty$ as $\lambda \to \infty$.  Elements of proof of convergence can be found in Appendix B.

\begin{figure*}[t]
\centering
\includegraphics[width=17cm]{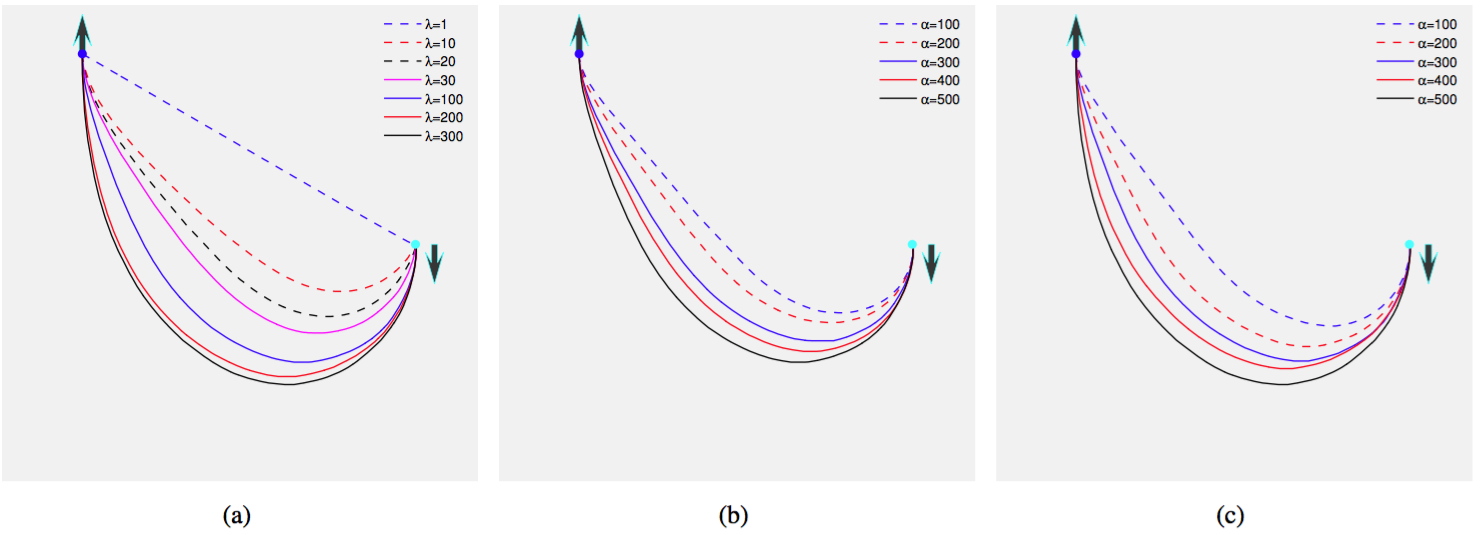}
\caption{Approximating Euler elastica curves by Finsler elastica minimal paths. (\textbf{a}) Finsler elastica minimal paths with $\alpha=500$ and different values of $\lambda$. (\textbf{b}) and (\textbf{c}) Finsler elastica minimal paths with $\lambda=100$ and $\lambda=300$ respectively, and  different values of $\alpha$.}
\label{fig:ApproximatedElasticas}
\end{figure*}


\subsection{Numerical Implementations}
\label{subsec:implementations}
Numerically, anisotropy is related to the problem stiffness, hence to its difficulty. In Table \ref{table:Cpu}, we show the computation time and the average number of Hopf-Lax updates required for each grid point by the adaptive stencils based fast marching method \citep{mirebeau2014efficient} for $\alpha=500$ and different values of $\lambda$ on a $300^2\times 108$ grid. This experiment was performed with a C++ implementation running on a standard 2.7GHz Intel I7 laptop with 16Gb RAM.
 
 We observe on Table \ref{table:Cpu} a logarithmic dependence of computation time and average number of the Hopf-Lax updates per grid point with respect to  anisotropy.  These observations agree with the  complexity analysis of the fast marching method presented in \citep{mirebeau2014efficient}, yielding the upper bound $\cO(N(\ln\mu)^3 + N \ln N)$, depending poly-logarithmically on the anisotropy ratio $\mu$~\eqref{eq:anisotropicRatio}, and quasi-linearly on the number $N$ of discretization points in the orientation-lifted domain $\bar\Omega$. 
In contrast, numerical methods such as \citep{sethian2003ordered} displaying a polynomial complexity $\cO(\mu^2 N\ln N)$ in the anisotropy ratio would be unworkable. The iterative AGSI method \citep{Bornemann:2006hc}, on the other hand, requires hundreds of evaluations of the Hopf-Lax operator \eqref{eq:Hopf-Lax}  per grid point to converge for large anisotropies, which also leads to prohibitive computation time, thus impractical. For $\lambda=30$ or $100$, the average numbers of the Hopf-Lax updates per grid required by the AGSI method are approximately $86$ and $182$, respectively, while the numbers of Hopf-Lax from the fast marching method are only $6.49$ and $7.27$, respectively, as demonstrated in Table~\ref{table:Cpu}.

In Fig.~\ref{fig:ApproximatedElasticas}a, we show  different Finsler elastica minimal paths, computed by the fast marching method \citep{mirebeau2014efficient},  with $\alpha=500$ (see~\eqref{eq:metricFinsler_A}) and different values of $\lambda$. The  arrows indicate the initial source and end points tangents. The cyan point denotes the initial position and the blue point indicates the end position. In Figs.~\ref{fig:ApproximatedElasticas}b and \ref{fig:ApproximatedElasticas}c, we show the Finsler  elastica minimal paths for different values of $\alpha$, with $\lambda=100$ and $\lambda=300$ respectively.  In this experiment, the angle resolution is $\theta_s=2\pi/108$ and  the image size is $300\times 300$. 
When $\lambda=1$, the metric $\cF^\lambda$ is constant over the domain $\bar\Omega$ and degenerates to the isotropic orientation-lifted metric $\cF^{\rm I}$  \eqref{eq:IsotropicLiftedMetric}, since the coefficient in front of the term $\langle\vt,\fu\rangle$ in~\eqref{eq:metricFinsler_A} vanishes. Hence the minimal geodesics are straight lines, see Fig.~\ref{fig:ApproximatedElasticas}a, that do not align with the prescribed endpoints tangents.  From Fig. \ref{fig:ApproximatedElasticas}, one can point out that as $\lambda$ and $\alpha$ increasing, curvature penalization forces the extracted paths to gradually align with the prescribed endpoints tangents and take the elastica shape.

\begin{table}[!t]
\centering
\setlength{\tabcolsep}{5pt}
\renewcommand{\arraystretch}{1.5}
\centering
\caption{Computation time (in seconds) and average number of Hopf-Lax updates required for each grid point by the fast marching method with  $\alpha=500$ and different values of $\lambda$ on a $300^2\times 108$ grid.}
\vspace{1mm}
\begin{tabular}{lccccccc}
\hline
$\lambda$  &1     & 10   & 20    & 30   &100  &200   &1000 \\
time      &13.9s & 25.3s  & 27.3s   & 27.7s  &31.7s  &33.9s   &36.8s\\
number     &3     & 5.49 & 6.06  & 6.49 &7.27 &7.82  &8.12 \\
\hline
\end{tabular}
\label{table:Cpu}
\end{table}

\subsection{Image Data-Driven Finsler Elastica Metric $\cP$}
\label{subsec:weightedFinsler}
We use $\Phi_0 \equiv 1$ in Section~\ref{sec:GeodesicInterpretation} for the sake of simplicity.   In the general case, the metric  $\cF^\infty$~\eqref{eq:inftyFinsler} and  its approximation $\cF^\lambda$~\eqref{eq:metricFinsler_A} should be respectively replaced by $\Phi_0^{-1}\cF^\infty$ and $\Phi^{-1}_0\cF^\lambda$.  Furthermore, in order to take into account the orientation information,  we use an orientation-dependent speed function $\Phi:\bar\Omega\to\bR^+$ to replace  $\Phi_0$. In this case, the data-driven Finsler elastica metric can be defined by
\begin{equation}
\label{eq:FinslerWeighted}
\cP(\bx,\bu):=\frac{1}{\Phi(\bx)}\cF^\lambda(\bx,\bu),\quad \forall\bx\in\bar\Omega,\,\forall\bu\in\bR^3, 
\end{equation}
and minimizing Eq.~\eqref{eq:DataBendEnergy} is approximated for large values of $\lambda$ by minimizing
\begin{align*}
\cL(\Gamma)&=\int_0^1\frac{1}{\Phi(\gamma(t))}\cF^\lambda(\gamma(t),\gamma^\prime(t))dt\\
&=\int_0^1\cP(\gamma(t),\gamma^\prime(t))dt,
\end{align*}
where $\gamma$ is the orientation-lifted curve of $\Gamma$.
The data-driven Finsler elastica metric $\cP$  is asymmetric in the sense that for most vectors $\bu\neq\0$, one has
\begin{equation}
\label{eq:Asymmetric}
\cP(\bx,\bu)\neq \cP(\bx,-\bu),\quad \forall\,\bx\in\bar\Omega.
\end{equation}
This asymmetric property can help to build a closed contour passing through a collection of orientation-lifted points as discussed in Section \ref{subsec:ClosedContour}.

The minimal action map $\cW_{\bs}$ associated to the data-driven Finsler elastica metric $\cP$ and an initial source point $\bs$ is the unique viscosity solution to the Eikonal PDE~\eqref{eq:EikonalPDE}~\citep{lions1982generalized,sethian2003ordered}:\begin{equation}
\label{eq:FinslerDataEikonalPDE}
\begin{cases}
\cP^*_{\bx}\big(\nabla \cW_{\bs}(\bx)\big)=1, &\forall \bx\in\bar\Omega\backslash\{\bs\} ,\\
\cW_{\bs}(\bs)=0,
\end{cases}	
\end{equation}
where  $\cP^*_{\bx}$ is the dual norm of $\cP_{\bx}(\cdot):=\cP(\bx,\cdot)$ defined by \eqref{eq:dualNorm}. The geodesic distance value $\cW_{\bs}(\bx)$ between $\bx$ and $\bs$ depends on both the curvature and the orientation information of the minimal paths.
When $\lambda$ is sufficiently large, the spatial and angular resolutions are sufficiently small, the fixed point system~\eqref{eq:FixedPoint} is properly solved,  and the minimal paths are properly extracted by~\eqref{eq:geodesicODE}.

\section{Computation of  Data-Driven Speed Functions by Steerable Filters}
\label{sec:Velocity}
In this section, we introduce two types of anisotropic speed  functions  for  boundary detection and tubular structure extraction, both of which are based on the steerable filters.

\subsection{Steerable Edge Detector}
\label{subsec:SteerableEdgeFilrer}
\citet{jacob2004design} proposed a new class of edge detection filters based on  the computational framework  and the steerable property. Letting $G_{\sigma}$ be a 2D isotropic Gaussian kernel with variance $\sigma$ and $\x=(x,y)$, the computational steerable filter ${\rm F} _\theta^{{\rm M}}$ with order $M$ can be expressed as
\begin{equation}
\label{eq:SteerableEdgeDetector}
{\rm F}^{{\rm M}}_\theta(\x)=\sum_{\tau=1}^M\sum_{\xi=0}^{\tau}\cK_{\tau,\xi}(\theta)\frac{\partial^{(\tau-\xi)}}{\partial x^{(\tau-\xi)} }\frac{\partial^{\xi}}{\partial y^{\xi}}G_{\sigma}(\x),
\end{equation}
where $\theta\in[0,2\pi)$ and $\cK_{\tau,\xi}$ are the orientation-dependent coefficients which can be computed in terms of some optimality criteria. 
Particularly, when $M=1$, the steerable filter ${\rm F}^{1}_\theta$ becomes the classical Canny detector~\citep{canny1986computational}. 
For higher order steerable filters, for example, $M=3$ or $M=5$, the orientation-dependent responses of the filters will be more robust to noise. Therefore, we set $M=5$ for the relevant experiments. Regarding the details of the computation of $\cK_{\tau,\xi}$, we refer to \citep{jacob2004design}.

A color image is regarded as a vector valued map $\mathbf I: \Omega\to\bR^3$, $\mathbf I(\x)=[I_1(\x),I_2(\x),I_3(\x)]$ for each $\x\in\Omega$.
A multi-orientation response function $h: \bar\Omega\to \bR^+$ of a color image $\mathbf I$ can be computed by the steerable filter ${\rm F}^{\rm M}_\theta$ \eqref{eq:SteerableEdgeDetector}
\begin{equation}
\label{eq:EdgeResponse}
 h(\x,\theta)=\frac{1}{3}\sum_{i=1}^3|I_i(\x)\ast {\rm F^M_\theta}(\x)|,\quad \forall\,\theta\in[0,2\pi).
\end{equation}
For a gray level image $I:\Omega\to\bR$, we have the simple computation of $h$:
\begin{equation}
\label{eq:GrayEdgeResponse}	
h(\x,\theta)=|I(\x)\ast{\rm F^M_\theta}(\x)|,\quad \forall\,\theta\in[0,2\pi).
\end{equation}
The multi-orientation response function $h$ is symmetric in the sense that for any orientation $\theta_\pi\in[0,\pi)$,  one has
\begin{equation*}
h(\x,\theta_\pi)=h(\x,\theta_\pi+\pi),\quad \forall\x\in\Omega.
\end{equation*}

\subsection{Multi-Orientation Optimally Oriented Flux Filter}
\label{subsec:TubularSpeed}
Optimally oriented flux filter is used to extract the local geometry of the image~\citep{law2008three}. The oriented flux of an image $I: \Omega\to\bR^+$, of dimension $2$, is defined by the amount of the image gradient projected along the orientation $\fv$ flowing out from a 2D circle $\cC_r$ at  point $\x=(x,y)\in\Omega$ with radius $r$:
\begin{equation}
\label{eq:OOF}
f(\x;r,\fv)=\int_{\partial \cC_{r}}(\nabla(G_{\sigma}\ast I)(\x+r\mathbf{n})\cdot \fv )(\fv \cdot \mathbf{n})\,ds,
\end{equation}
where $G_{\sigma}$ is a Gaussian with variance $\sigma$.  $\mathbf{n}$ is the outward normal of the boundary $\partial \cC_{r}$ and $ds$ is the infinitesimal length on $\partial\cC_{r}$.  According to the divergence theory, one has 
\begin{equation*}
f(\x;r,\fv)=\fv^{\rm T}\,\mathbf{Q}(\x,r)\,\fv,
\end{equation*}
for some symmetric matrix $\mathbf{Q}(\x,r)$:
\begin{equation}
\label{eq:OOF_Q}
\mathbf Q(\x,r)=
\begin{pmatrix} 
\partial_{xx}G_\sigma\,& \partial_{xy}G_\sigma \\
\partial_{yx}G_\sigma\,& \partial_{yy}G_\sigma
\end{pmatrix}
\ast\mathbbm{1}_{r}\ast I(\x),
\end{equation}
where $\mathbbm{1}_{r}$ is an indicator function of the circle $\cC_{r}$.

Let $\lambda_{1}(\cdot)\geq \lambda_{2}(\cdot)$ be the eigenvalues of symmetric matrix $\mathbf  Q(\cdot)$ and suppose that the intensities inside the tubular structures are darker than the background regions. If a point $\x$ is inside the tubular structure, one has $\lambda_{1}(\x,r^*(\x))\gg0$ and $\lambda_{2}(\x,r^*(\x))\approx 0$, where $r^*$ is the optimal scale map
\begin{equation}
\label{eq:OptimalScale}	
r^*(\x):=\arg\max_{r}\left\{\frac{1}{r}\lambda_1(\x,r)\right\},\quad \forall \x\in\Omega.
\end{equation}
The scale normalized factor $1/r$ in \eqref{eq:OptimalScale} provides the responses of the optimally oriented flux filter a scale invariant property as discussed in~\citep{law2008three}. 

As shown in~\citep{benmansour2011tubular}, the optimally oriented flux filter is steerable, which means that we can construct a multi-orientation  response function $g:\bar\Omega\to\bR^+$ for any $\theta\in [0,\,2\pi)$ by:
\begin{equation}
\label{eq:CentrelineResponse}
g(\x,\theta)=\max\big\{\fu^{\rm T}_{\theta}\,\mathbf{Q}(\x,r^*(\x))\,\fu_{\theta},\,0\big\},\quad \forall\,\x\in\Omega.
\end{equation}
where $\fu_\theta=(\cos\theta,\sin\theta)^{\rm T}$ is a unit vector associated to $\theta$.

In addition,  the vesselness map $V_{\rm n}:\Omega\to\bR$, which indicates the probability of a pixel $\x$ belonging to a vessel,  can be calculated by:
\begin{equation}
\label{eq:Vesselness}
V_{\rm n}(\x)=\max\left\{\max_{r}\left\{\frac{1}{r}\lambda_1(\x,r)\right\},\,0\right\}.
\end{equation}
The vesselness map $V_{\rm n}$ will be used to compute the isotropic Riemannian metric in the following section.

\subsection{Computation of Data-Driven Speed Functions}
\label{subsec:SpeedFunctions}
A requirement on the  speed function $\Phi$ used by the  data-driven Finsler elastica metric $\cP$ defined in~\eqref{eq:FinslerWeighted} is that it should depend on the position and the orientation.
Therefore,  based on  the orientation-dependent response function $h$, defined in Section \ref{subsec:SteerableEdgeFilrer}, the speed function $\Phi$  that is used for object boundary detection  can be computed by
\begin{equation}
\label{eq:Edge_Velocity}
\Phi(\x,\theta)=1+\eta\left(\frac{ h\big(\x,\text{mod}(\theta+\pi/2,2\pi)\big)}{\|h\|_\infty}\right)^p,
\end{equation}
for any $\x\in\Omega$ and any orientation $\theta\in[0,2\pi)$, where $\text{mod}(\theta+\pi/2,2\pi)$ is defined as $(\theta+\pi/2)\mod 2\pi$. 

Similarly,  we define the speed function $\Phi$ for tubular structure  extraction by using the  function $g$ \eqref{eq:CentrelineResponse}
\begin{equation}
\label{eq:Centreline_Velocity}
\Phi(\x,\theta)=1+\eta\left(\frac{ g\big(\x,\text{mod}(\theta+\pi/2,2\pi)\big)}{\|g\|_\infty}\right)^p,
\end{equation}
where  $\eta$ and $p$ are positive constants. In this paper,  we use $p=2$ for all the experiments.

\section{Closed Contour Detection and Tubular Structure Extraction}
\label{sec:ContourDetection}

We use the following convention in the remaining part of this paper: if $\bp=(\p,\theta)$ is a  point in $\bar\Omega$, then we use $\bp^\dag=\big(\p,\text{mod}((\theta+\pi),2\pi)\big)$ to denote the orientation-lifted point that has the same physical position $\p$ with $\bp$ but the opposite direction, where $\theta\in[0,2\pi)$.

\begin{figure*}[!t]
\centering
\includegraphics[width=17cm]{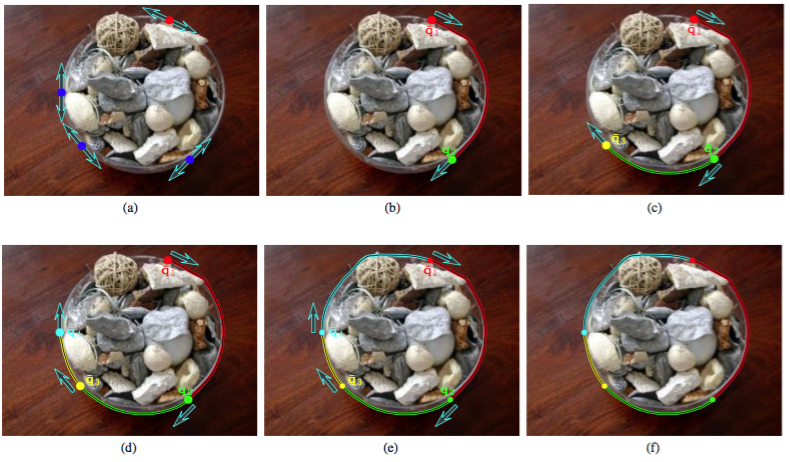}
\caption{Steps for the closed contour detection procedure. (\textbf{a}) Original image and all  of the vertices in $\cD$ denoted by dots and arrows. (\textbf{b}) The first pair of successive vertices $(\bq_1,\bq_2)$ is detected. (\textbf{c}) The third   vertice $\bq_3$ is detected. (\textbf{d}) The final vertice $\bq_4$ is detected and the closed contour detection procedure is terminated. (\textbf{e}) The minimal path joining $\bq_4$ and $\bq_1$ is tracked (cyan curve). (\textbf{f})  The final closed contour is obtained.}
\label{fig:Steps}
\end{figure*}

\subsection{Closed Contour Detection as a Set of  Piecewise Finsler Elastica Minimal Paths}
\label{subsec:ClosedContour}
In this section, we propose an interactive  closed contour detection method based on the Finsler elastica minimal paths constrained by a set $\cH$ of $m$ user-provided  physical points 
\begin{equation*}
\cH:=\{\x_i\in\Omega, i=1,2, ...,m;\,m\geq2\},
\end{equation*}
all of which are on the target object boundary. As discussed in Section \ref{sec:FinslerMinimaPathsModel}, the Finsler elastica metric $\cP$ is defined on the orientation-lifted space $\bar\Omega$.  Thus,  we build an orientation-lifted collection $\cD$ of $\cH$ by 
\begin{equation}
\label{eq:Setoflifting}
\begin{split}
\cD:=\Big\{&\bx_i=(\x_i,\theta_i),\,\bx^\dag_i=\big(\x_i,{\rm mod}(\theta_i+\pi,2\pi)\big);\\
&i=1,2, ...,m,~\x_i\in\cH,~{\rm and}~\theta_i\in[0,2\pi)\Big\},
\end{split}
\end{equation}
where the directions $\theta_i$ are manually specified in this paper. Corresponding to each physical point $\x_i\in\cH$, there exist two orientation-lifted points: $\bx_i$, $\bx^\dag_i\in\cD$, which have opposite tangents.
We show an example of these orientation-lifted points in Fig.~\ref{fig:Steps}a, where the physical positions and the corresponding tangents are denoted by dots and arrows, respectively. 

The basic idea of the proposed closed contour detection method is to identify a set  of ordered vertices from the collection $\cD$, and to join these detected vertices  by a set of piecewise Finsler elastica minimal paths.

We start the closed contour detection procedure by selecting a physical position from $\cH$, say  $\x_1$. The corresponding orientation-lifted points of $\x_1$ are denoted by $\bx_1,\bx^\dag_1\in\cD$. Once $\x_1$ is specified, we remove both  $\bx_1$ and $\bx^\dag_1$ from $\cD$. As shown in Fig.~\ref{fig:Steps}a,  $\bx_1$ and $\bx^\dag_1 $ are denoted by a red dot and two  arrows with opposite directions. 

Let $\ba^*\in\cD$ be the closest orientation-lifted point to $\bx_1$ in terms of the geodesic distance $\cW_{\bx_1}$ \eqref{eq:FinslerDataEikonalPDE} with respect to the data-driven  Finsler elastica metric $\cP$, i.e.
\begin{equation}
\label{eq:ClosestVertex}
\ba^*=\arg\min_{\bz\in\cD}\cW_{\bx_1}(\bz). 	
\end{equation}
Similarly to $\ba^*$,  the closest orientation-lifted point $\bc^*\in\cD$ of $\bx^\dag_1$ can be detected. By the detected points $\ba^*$ and $\bc^*$, the first pair of successive vertices $(\bq_1,\bq_2)$ is determined simultaneously using the following criterion:
\begin{equation}
\label{eq:FirstandSecondVertices}
(\bq_1,\bq_2)=
\begin{cases}
(\bx_1,\ba^*),\text{ if }\cW_{\bx_1}(\ba^*)<\cW_{\bx^\dag_1}(\bc^*),\\
(\bx^\dag_1,\bc^*),\text{ otherwise}.	
\end{cases}	
\end{equation}
In Fig.~\ref{fig:Steps}b, we illustrate the vertices $\bq_1$ and $\bq_2$  by the red and green dots with arrows, respectively. If the minimal action map $\cW_{\bx_1}$ (resp. $\cW_{\bx^\dag_1}$) is computed via the fast marching method \citep{mirebeau2014efficient}, $\ba^*$ (resp.  $\bc^*$) is the first vertex reached by the fast marching front, which is monotonically advancing. Once  the first pair of successive vertices $(\bq_1,\bq_2)$ is found, the geodesic $\cC_{\bq_1,\bq_2}$ (red curve in Fig.~\ref{fig:Steps}b) can be tracked by reversing the path obtained through the ODE~\eqref{eq:geodesicODE},  and both $\bq_2$, $\bq^\dag_2$ are removed from $\cD$. If the number of physical points is $2$, i.e. $m=2$, the closed contour detection procedure can be terminated. In order to form a closed contour, one can recover the geodesic $\cC_{\bq_2,\bq_1}$ using the minimal action map associated to  the source point $\bq_2$.

If $m>2$, the subsequent vertex $\bq_i$ with $i\geq 3$ is identified from the remaining points of $\cD$ by searching for the  nearest neighbour  of the vertex $\bq_{i-1}$ in terms of geodesic distance, i.e.
\begin{equation}
\label{eq:NextVertex}
\bq_i=\arg\min_{\bz\in\cD}\cW_{\bq_{i-1}}(\bz).
\end{equation}
After the detection of the vertex $\bq_i$, we remove both $\bq_i$ and $\bq^\dag_i$ from $\cD$. Again the geodesic $\cC_{\bq_{i-1},\bq_i}$ can be tracked, as denoted by the green curve in Fig.~\ref{fig:Steps}c. 

The procedure of detecting  the nearest neighbor from the set of remaining orientation-lifted points is recursively carried out according to the criterion \eqref{eq:NextVertex} until $m$ ordered vertices have been identified.  Then the geodesic $\cC_{\bq_m,\bq_1}$, which is denoted by the cyan curve in Fig.~\ref{fig:Steps}e,  is computed by simply allowing $\bq_m$ to be the source point and  $\bq_1$ to be the end  point.  The final closed contour, denoted by $\cC$, is defined as the concatenation of all of the detected Finsler elastica minimal paths as demonstrated in Fig.~\ref{fig:Steps}f.

In summary, the proposed closed contour detection procedure aims to seeking a set $\chi$ of $m$ pairs of successive orientation-lifted  points from $\cD$:
\begin{equation}
\label{eq:set_pairs}
\chi=\bigcup_{i=1}^{m-1}\big\{(\bq_i,\bq_{i+1})\big\}	\bigcup \big\{(\bq_m,\bq_1)\big\},
\end{equation}
and a closed contour $\cC$ contains a set of  Finsler elastica minimal paths, joining all the pairs of vertices in $\chi$. 
This method simply matches orientation-lifted points by pairs, joining a vertex to the remaining nearest neighbor with respect to the curvature-penalized geodesic distance, so as to form a closed contour located at the expected object boundaries. Note importantly, that the obtained piecewise geodesic contour is \emph{smooth} ($C^1$ differentiable) since the initial source and end orientation-lifted points of the consecutive geodesics have both matching positions $\q_i$ and orientations $\theta_i$.  In fact, we find a closed contour passing  through all the orientation-lifted points in a greedy manner. Instead of trying out all possible combinations of Finsler elastica minimal paths,  we use  a greedy searching strategy that is performed with a low complexity. The problem we solve here is similar to the NP-hard traveling salesman problem, where the cities are represented by the orientation-lifted points $\bq_i\in\cD$.

\subsection{Perceptual Grouping via  Curvature-Penalized Geodesic Distance}
\label{subsec:perceptual_grouping}
Perceptual grouping is relevant to the task of curve reconstruction and completion \citep{cohen2001multiple}. The geodesic distance based perceptual grouping model was firstly introduced by \citet{cohen2001multiple} using the concept of  saddle point. The basic idea of  this model is to identify each pair of points which has to be linked by a minimal path from a set of key points. Later on, \citet{bougleux2008anisotropic} improved this grouping model  by  using path orientations  and structure tensors. However, neither of  the mentioned grouping methods considered the curvature penalization.

In this section, we  address the perceptual grouping problem of finding $n$  closed contours, each of which is formed by  a set of piecewise  Finsler elastica minimal paths. Each closed contour $\cC_i$ passes through all the ordered orientation-lifted points involved  in $\cD_i\subseteq\cD$, where $\cD$ is defined in  \eqref{eq:Setoflifting} and $i=1,2,3,...,n$.

We initialize the proposed perceptual grouping procedure by specifying an initial  physical position $\x_1$, where the  orientation-lifted points of $\x_1$, denoted by  $\bx_1$ and $\bx^\dag_1$,  are involved in $\cD$. Both $\bx_1$ and $\bx^\dag_1$ are removed from $\cD$ after detecting the respective nearest vertices that correspond to $\bx_1$ and $\bx^*_1$  by \eqref{eq:ClosestVertex}.  As a consequence,  the first two vertices $\bq_1,\bq_2$ are identified using the criterion of \eqref{eq:FirstandSecondVertices}, and the geodesic $\cC_{\bq_1,\bq_2}$ is recovered by using the ODE~\eqref{eq:geodesicODE}.
Once  the vertices $\bq_1$ and $\bq_2$ are detected, we add $ \bq_1,\bq_2$ to $\cD_1$, remove $\bq_2,\bq^\dag_2$ from $\cD$ and compensate $\bq_1$ to $\cD$. 

Similar to the closed contour detection procedure,  the next vertex $\bq_k$ with $k\geq 3$ is found based on  the criterion of~\eqref{eq:NextVertex} and the detected vertex $\bq_{k-1}$. Following the detection of vertex $\bq_k$, we add $\bq_k$ to $\cD_1$, remove $\bq_k$, $\bq_k^\dag$ from $\cD$, and track the geodesic $\cC_{\bq_{k-1},\bq_k}$ that joins $\bq_{k-1}$ to $\bq_k$. 
This perceptual grouping  procedure is carried out by recursively searching for new vertices. Once the vertex $\bq_1$ is  detected again according to the criterion \eqref{eq:NextVertex}, we stop the construction of $\cD_1$ after removing $\bq_1$ from $\cD$, and recover the geodesic ending at vertex $\bq_1$. The desired closed contour $\cC_1$ can be obtained by concatenating  all the detected Finsler elastica minimal paths with source and end points in $\cD_1$.

We start to build the collection $\cD_2$ by choosing a new physical point  as initialization. This initial physical point is obtained from the remaining orientation-lifted points of $\cD$. Similar to the procedure of constructing $\cD_1$, we build the collection $\cD_2$ from the remaining points of $\cD$. The procedure of building the collections $\cD_i$ can be terminated when $n$ such collections have been identified or when the collection $\cD$ is empty.
One can note that the constructed collections $\cD_i$ follow
\begin{equation*}
\cD_i \cap \cD_j=	\varnothing,\quad \forall i\neq j.
\end{equation*}
In contrast to the closed contour detection method described in Section \ref{subsec:ClosedContour},  we do not enforce all of the orientation-lifted points in $\cD$ to be used in the perceptual grouping procedure.

\subsection{Tubular Structure Extraction}
\label{subsec:tubular_structure}
In this section, we apply the proposed Finsler elastica minimal path model to the tubular structure extraction,  where the centerlines of the tubular structures are represented by the Finsler elastica minimal paths. 

The minimal paths with the proposed data-driven Finsler elastica metric depend on the tangents of  both the  source point and the end point. To simplify the initialization procedure, we firstly compute the optimal orientation map, denoted by $\Theta:\Omega\to [0,\pi)$, which minimizes the multi-orientation function $g$ in \eqref{eq:CentrelineResponse}: 
\begin{equation}
\label{eq:OptimalOrientation}
	\Theta(\x)=\arg\min_{\theta\in[0,\pi)}\{ g(\x,\theta)\},\quad \forall \x\in\Omega.
\end{equation}
Once the optimal orientation map $\Theta$ is obtained, for  the source  position $\s\in\Omega$, one can obtain two orientation-lifted points $\bs=(\s,\Theta(\s))$ and $\bs^\dag$. Additionally, for  any end position $\p_i\in\Omega$ ($i=1,2,\cdots,n$), the corresponding orientation-lifted end points are defined by $\bp_i=(\p_i,\Theta(\p_i))$  and $\bp_i^\dag$.  

For each set of orientation-lifted end points $\{\bp_i,\bp^\dag_i\}$, we can extract four possible \emph{geodesics}, each of which joins a source point in $\{\bs,\bs^\dag\}$ to an end point in $\{\bp_i,\bp_i^\dag\}$. The goal in this section is to search for a geodesic $\cC^*_i$ with \emph{minimal}  geodesic curve length associated to the metric $\cP$, among all of the four possible geodesics.

Let us denote the  source  and  end points of the geodesic $\cC^*_i$ by $\ba^*$ and $\bc_i^*$, respectively. 
If the geodesic curve length is estimated by the fast marching method \citep{mirebeau2014efficient}, this procedure can be simplified as follows: starting the fast marching front propagation from both of the source points $\bs$ and $\bs^\dag$, the orientation-lifted point $\bc_i^*\in\{\bp_i,\bp_i^\dag\}$ is the first point that is reached by the front. The desired geodesic $\cC^*_i$ can be determined by reversing the path obtained through  the ODE \eqref{eq:geodesicODE}.  As a result, a set $\{\cC^*_i;1\leq i\leq n\}$ of all the desired geodesics can be extracted  from the same minimal action map generated by a single fast marching propagation.

In these applications, the geodesic distance maps associated to  the Finsler elastica metric are computed in a manner of early abort, i.e., once all the specified orientation-lifted endpoints are reached by the fast marching front,  the geodesic  distance computation will be terminated. This early abort trick can greatly reduce the computation time. 
It is similar to the partial front propagation described in \citep{deschamps2001fast} with a simple extension to multiple points.

\section{Experimental Results}
\label{sec:Experiments}
We show the advantages  of using curvature penalization for minimal paths in the following experiments involving a study of the proposed  metric itself, and  the comparative results against the isotropic Riemannian (IR) metric,  the anisotropic Riemannian (AR) metric  and the isotropic orientation-lifted Riemannian (IOLR) metric in the applications of closed contour detection and tubular structure extraction. Note that in this section, whenever we mention Finsler elastica metric, we mean the data-driven Finsler elastica metic.

\subsection{Riemannian Metrics Construction}
\label{subsec:ColorGradient}
We adopt the color image gradient  introduced by \citet{di1986note} to construct the IR  and AR metrics for closed boundaries detection of objects in color images.

Considering a color image $\mathbf I=(I_1,\,I_2,\,I_3)$ and a Gaussian kernel $G_\sigma$ with variance $\sigma$, the respective Gaussian-smoothed x- and y-derivatives of  $\mathbf I$ are defined by
\begin{equation}
\label{eq:GaussianDerivatives}
\mathbf{I}_x^\sigma:=\partial_xG_\sigma\ast\mathbf I\quad \text{and}\quad \mathbf{I}_y:=\partial_y G_\sigma\ast\mathbf I,
\end{equation}
where $\mathbf{I}^\sigma_x(\cdot)$ and $\mathbf{I}^\sigma_y(\cdot)$ should be understood as $1\times3$ vectors. Following~\citep{di1986note}, a tensor $\cE(\x)$  can be defined by
\begin{equation*}
\cE(\x)=
\begin{pmatrix}
\|\mathbf I^\sigma_x(\x)\|^2~~ &\langle\mathbf I^\sigma_x(\x),\, \mathbf I^\sigma_y(\x)\rangle\\
&\\   
\langle\mathbf I^\sigma_x(\x),\, \mathbf I^\sigma_y(\x)\rangle ~~&\|\mathbf I^\sigma_y(\x)\|^2
\end{pmatrix},\quad \forall\,\x\in\Omega.
\end{equation*}
It is known that  the tensor  $\cE(\x)$  can be decomposed in terms of its eigenvalues and eigenvectors for all $\x\in\Omega$ by
\begin{equation*}
\cE(\x)=\varphi_1(\x)\,\g_1(\x)\,\g^{\rm T}_1(\x)+\varphi_2(\x)\,\g_2(\x)\,\g^{\rm T}_2(\x),	
\end{equation*}
where $\varphi_1(\x)$, $\varphi_2(\x)$ are the eigenvalues of $\cE(\x)$ and $\g_1(\x)$,  $\g_2(\x)$ are the associated eigenvectors. Without loss of generality, we assume that $\varphi_1(\x)\leq \varphi_2(\x)$, $\forall\,\x\in\Omega$.  In this case, $\varphi_2$ denotes the color gradient magnitude and   $\g_2$ denotes the unit color gradient vector field.

Based on the scalar field $\varphi_2$, the IR metric $\cF^{{\rm I}}$  can be constructed for all $\x\in\Omega$ by
\begin{equation}
\label{eq:IRMetricSpeed}
\cF^{\rm I}(\x, \fu)=\Big(\beta_1+\beta_2\,\varphi^{p}_2(\x)\Big)^{-1}\|\fu\|,
\end{equation}
where $\beta_1$, $\beta_2$ and $p$ are  positive constants. In the following relevant experiments, we set $\beta_1=1$ and $p=2$.

The tensor field $\cM_{\rm A}$ for the AR metric $\cF^{\rm A}$ \eqref{eq:ARMetric} can be computed by
\begin{align}
\label{eq:ColorARTensor}
\cM_{\rm A}(\x)=&\exp(\tau\,\varphi_2(\x))\,\g_1(\x)\,\g_1^{\rm T}(\x)\nonumber\\
&+\exp(\tau\,\varphi_1(\x))\,\g_2(\x)\,\g_2^{\rm T}(\x),\quad \forall \x\in\Omega,
\end{align} 
where $\tau$ is a  negative  constant.

In the tubular structure extraction experiments,  for the construction of the IR metric $\cF^{\rm I}$~\eqref{eq:IRMetricSpeed}, we simply replace the scalar field $\varphi_2$  by the vesselness map $V_{\rm n}$ defined in~\eqref{eq:Vesselness}. 
Moreover, regarding the construction of the AR metric, we make use of a radius-lifted tensor field as introduced by~\citet{benmansour2011tubular}, instead of using the tensor field  $\cM_{\rm A}$ defined in~\eqref{eq:ColorARTensor}. In this case, the AR metric is regarded as the anisotropic radius-lifted Riemannian (ARLR) metric defined over the radius-lifted domain. In the following related experiments, we use the optimally oriented flux filter~\citep{law2008three} to compute the ARLR metric as suggested in~\citep{benmansour2011tubular}. For further details on the construction of the ARLR metric, we refer the reader to~\citep{benmansour2011tubular}.

The speed function $\Phi^{\rm IL}$  for the  IOLR metric $\cF^{\rm IL}$ \eqref{eq:IsotropicLiftedMetric} should be dependent of the orientations. Simply,  we compute the speed function  $\Phi^{\rm IL}$ by
\begin{equation}
\Phi^{\rm IL}(\x,\theta)=\Phi(\x,\theta),\quad\forall\,\theta\in[0,\pi),\,\forall\,\x\in\Omega,
\end{equation}
where  $\Phi$  is the orientation-dependent speed function defined in \eqref{eq:Edge_Velocity} and \eqref{eq:Centreline_Velocity}.  The parameter $\rho$ of the IOLR metric $\cF^{\rm IL}$ penalizes the variations of the orientation $\theta$, and is set as $\rho=\alpha$, where $\alpha$ is the parameter for the curvature term in the bending energy $\cL$ \eqref{eq:DataBendEnergy}.

\begin{figure*}[!t]
\centering
\includegraphics[width=17cm]{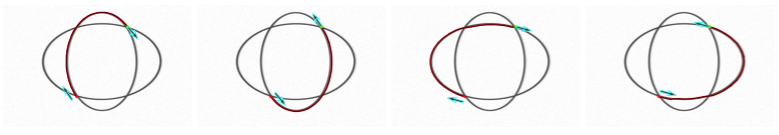}
\caption{Flexible Finsler elastica minimal paths extraction on ellipse-like curves. The red  and green dots denote the source and end positions, respectively. Arrows indicate the tangents.}
\label{fig:Asymetry_Ellipse1}
\end{figure*}

\begin{figure*}[!t]
\centering
\includegraphics[width=17cm]{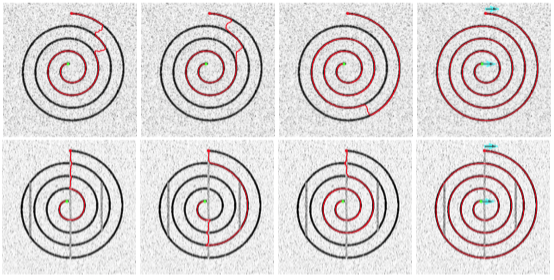}
\caption{Minimal paths extraction results on Spirals. \textbf{Columns 1-4} Minimal paths extracted by the IR metric, the ARLR metric, the IOLR metric and the  Finsler elastica metric, respectively. The red and green dots denote the source  and  end positions, respectively. The arrows indicate the tangents.}
\label{fig:Spiral}
\end{figure*} 

\subsection{Parameters Setting}
\label{subsec:Parameters}
Curvature penalization in the proposed Finsler elastica metric relies on two parameters, $\alpha$ and $\lambda$ \eqref{eq:metricFinsler_A}. The choice of $\lambda$ is dictated by algorithmic compromises. Indeed, minimal paths with respect to the Finsler elastica metric $\cP$  converge to the elastica curves in the limit $\lambda \to \infty$, hence a large value  of $\lambda$ is desirable. However, large values of $\lambda$ yield metrics with strong anisotropy ratio $\mu(\cP)$. As a result,  the numerical algorithm used in this paper, adapted from \cite{mirebeau2014efficient}, uses larger discretization stencils, which increases its numerical cost and reduces its locality. For instance, $\lambda = 30$ (resp.\ $100$ or $300$) leads to stencils with a radius of $4$ pixels (resp. $8$ or $13$). We typically use $\lambda=100$.

On the other hand, the parameter $\alpha$ is used  to weight the curvature penalty in the Finsler elastica metric $\cP$. In the course of the fast marching method, a large value of $\alpha$ makes the front to propagate slowly along the orientation dimension, implying that the obtained geodesics tend to be smooth, i.e., with low curvature. When $\alpha$ is very small, the extracted geodesics mainly depend on the image data-driven speed functions defined in Section \ref{subsec:SpeedFunctions}. Therefore, the choice of $\alpha$ should depend on the desired image features. Basically, we make use of the following heuristics.
There is a natural candidate $\alpha_*$ for the parameter $\alpha$, dictated by the physical units of the parameters, namely $\alpha_* = (R_* / \Phi_*)^2$, where $R_*$ is the smallest radius of curvature of the image features to be extracted, measured in pixels, and $\Phi_*$ is the typical value of the speed function $\Phi$ around these features. 

The angular resolution is set as $\theta_s=\pi/36$ for both the IOLR metric and the proposed Finsler elastica metric.
The parameter $\eta$ for image data-driven speed function $\Phi$ is  set for each tested image individually. The parameter $\beta_2$ that is used in the IR metric $\cF^{\rm I}$~\eqref{eq:IRMetricSpeed} is set as $\beta_2=2\eta$ for all the comparative experiments. 
Unless otherwise specified, we set the anisotropy ratio values for the AR metric and the ARLR metric to be 20.

\subsection{Smoothness and Asymmetry of the Finsler Elastica Minimal Paths}
\label{subsec:Asymmetry-SmoothnessExperiments}
The  Finsler elastica metric invoking orientation lifting and curvature penalty benefits from the smooth and asymmetric properties of  minimal paths. We demonstrate  these properties in a synthetic image as shown in Fig.~\ref{fig:Asymetry_Ellipse1}, where two ellipse-like shapes cross each other. The red   and  green dots are the source and end positions, respectively. The arrows  indicate the tangents at the corresponding positions. One can see that for the fixed source  and end positions, changing the corresponding tangents will give different minimal paths.  As shown in the first two columns of  Fig.~\ref{fig:Asymetry_Ellipse1},  the two minimal paths with the same  source and end positions can form a complete ellipse  shape.

In Fig. \ref{fig:Spiral}, we design a spiral that has high anisotropy.  The source  and end positions are placed at the ends of the spiral. In the top row we add high noise to the spiral, while in the bottom row we blur the spiral.
In columns 1-4, we show the  minimal paths  obtained from the IR metric, the ARLR metric, the IOLR metric and the  Finsler elastica metric, respectively. One can see that by using  the mentioned Riemannian metrics, the shortcuts occur as shown in columns 1 to 3. In the top row, the minimal path (shown in column 3) obtained by using the IOLR metric is improved compared against the respective paths from the IR  and  AR metrics. However, a segment of the spiral is  also missed due to the shortcuts problem.  In contrast, the minimal paths shown in column 4   extracted by the Finsler elastica metric can completely avoid the problem of  shortcuts  thanks to the curvature penalization  embedded in the proposed metric. In this experiment, we  use an anisotropy ratio value of $100$ for the ARLR metric. For the Finsler elastica metric, we set $\alpha=500$ to ensure the  geodesics to be smooth enough.

\begin{figure}[!t]
\centering
\includegraphics[width=8cm]{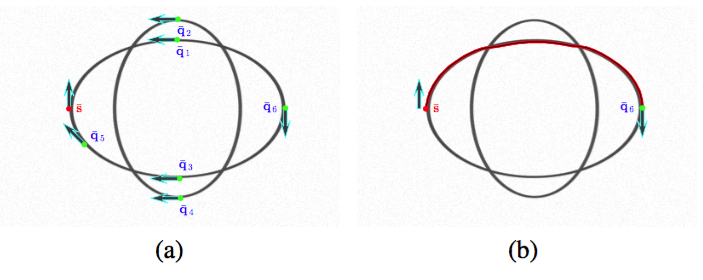}
\caption{Finding the closest orientation-lifted candidate to the source point $\bs$ in terms of the geodesic distance associated to the Finsler elastica metric. (\textbf{a})  The orientation-lifted candidates $\bar\q_i$, $i=1,2,\cdots,6$, and source point $\bs$ are demonstrated. (\textbf{b}) The closest candidate $\bar\q_6$ is detected. Red curve indicates the geodesic linking $\bs$ to $\bar\q_6$.}
\label{fig:Asymetry_Ellipse2}
\end{figure}

\begin{figure}[!t]
\centering
\includegraphics[width=8cm]{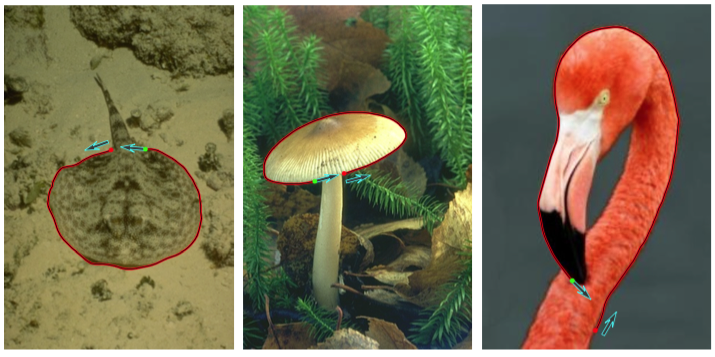}	
\caption{Finsler elastica minimal paths extraction results. The red  and  green dots indicate the  source and end positions,  respectively. The arrows indicate the corresponding tangents.}
\label{fig:BoundaryExtraction}
\end{figure}

In Fig.~\ref{fig:Asymetry_Ellipse2}a, we specify  six orientation-lifted candidates $\bar\q_i$, $i=1,2...,6$, which are denoted by green dots with arrows, and a source point $\bar{\mathbf s}$ indicated by red dot with arrow. Among all of these candidates, we aim to find the closest orientation-lifted candidate to the   source  point $\bs$, in terms of the geodesic distance associated to the metric $\cP$ defined in~ \eqref{eq:FinslerWeighted}. In Fig.~\ref{fig:Asymetry_Ellipse2}b,  it is shown that the closest orientation-lifted point to $\bs$ is the candidate $\bar\q_6$, even though the geodesic (red curve),  joining  the source point $\bs$ to the candidate $\bq_6$,  passes through the vicinity of the physical position of  $\bar\q_1$. Moreover, one can claim that the Euclidean distance  between the physical positions of $\bs$ and $\bar\q_6$ is the largest among the  Euclidean distance values between the physical positions of  $\bs$ and any remaining candidate $\bq_i$.  This experiment demonstrates the asymmetric and smooth properties of the proposed Finsler elastica minimal path model.

In Fig.~\ref{fig:BoundaryExtraction}, we show the extracted Finsler elastica minimal paths on three images, where each pair of the prescribed source and end positions is very close to each other in terms of Euclidean distance. For each image, we expect to detect a long boundary between the two given orientation-lifted points.
It can be observed that the extracted  minimal paths are able to catch the desired boundaries. In Fig.~\ref{fig:BoundaryExtraction}, the images shown in columns 1 and 2 are from the Berkeley Segmentation Dataset \citep{arbelaez2011contour} and the image in column 3 is from  the Weizmann dataset \citep{alpert2012image}.

\begin{figure*}[!htb]
\centering
\includegraphics[width=3.4cm]{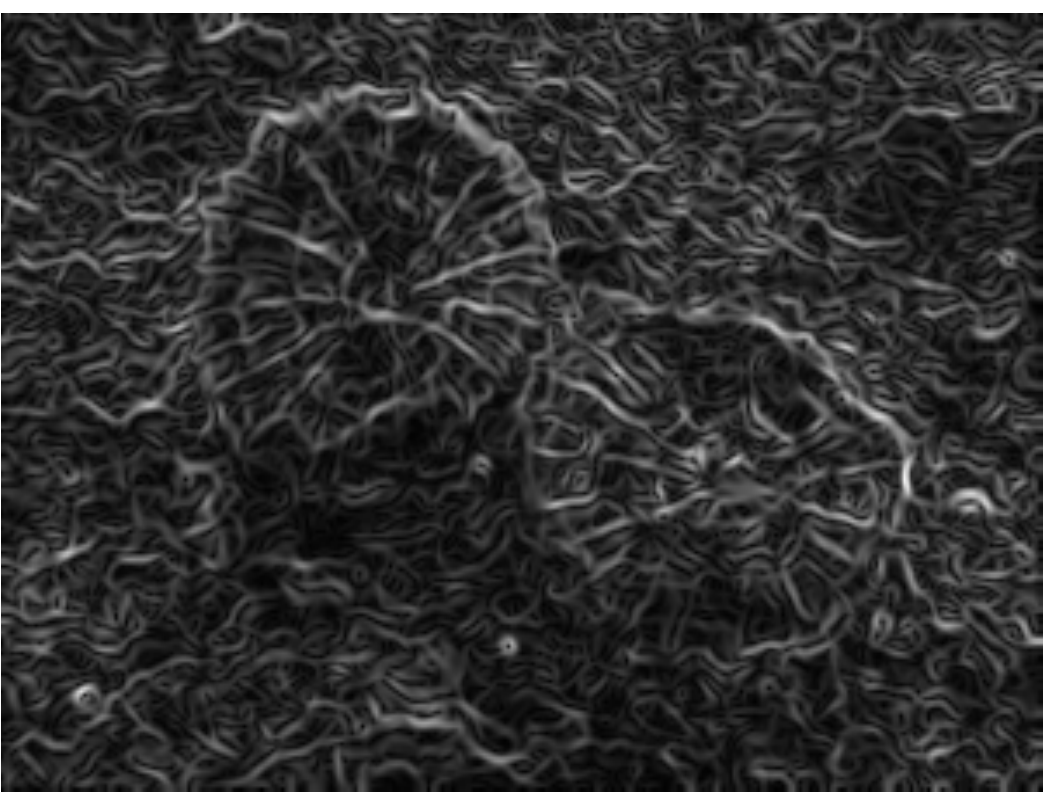}
\includegraphics[width=3.4cm]{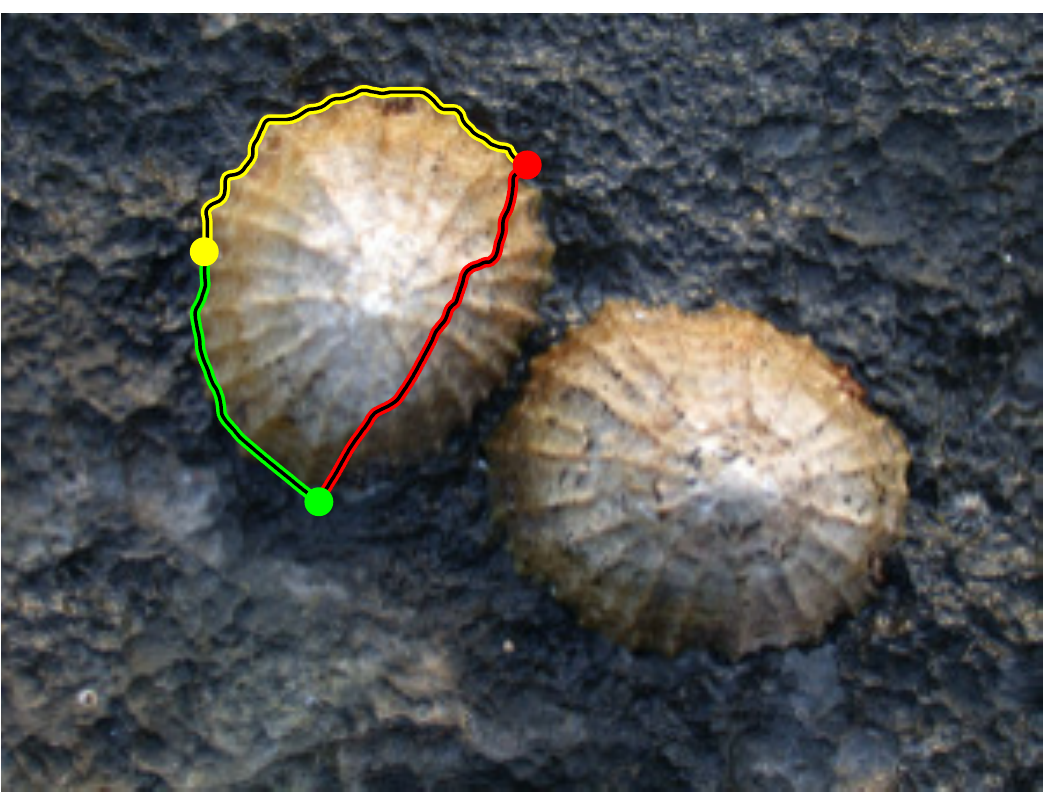}
\includegraphics[width=3.4cm]{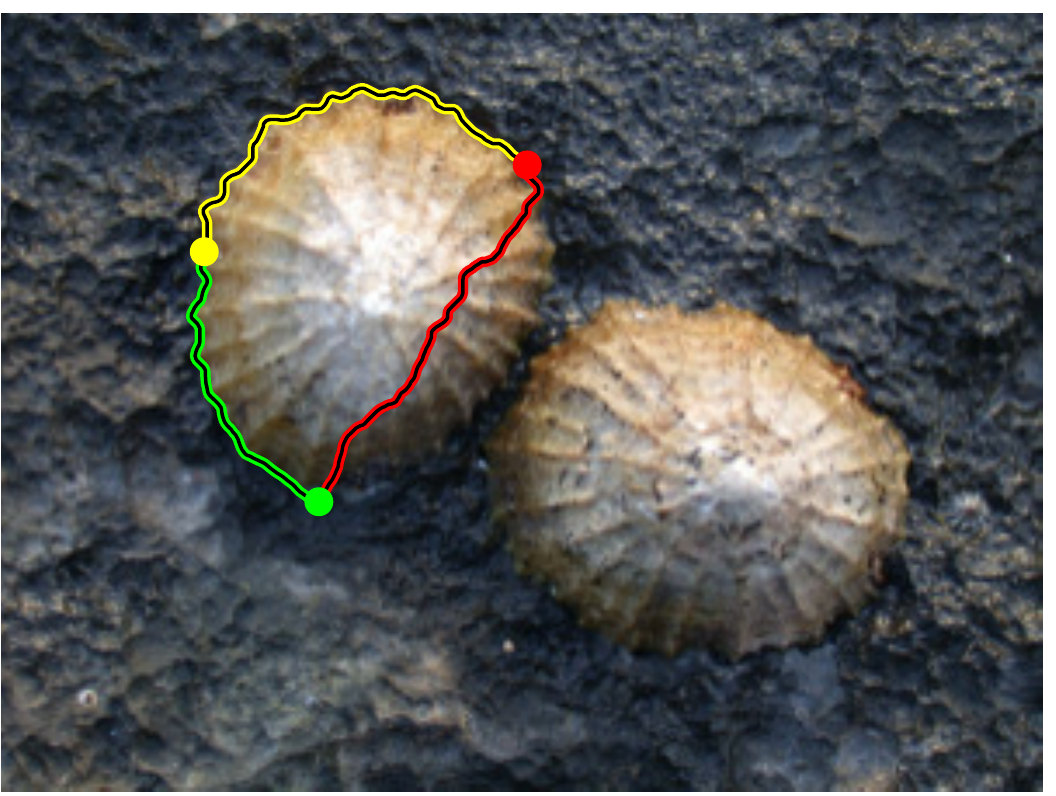}
\includegraphics[width=3.4cm]{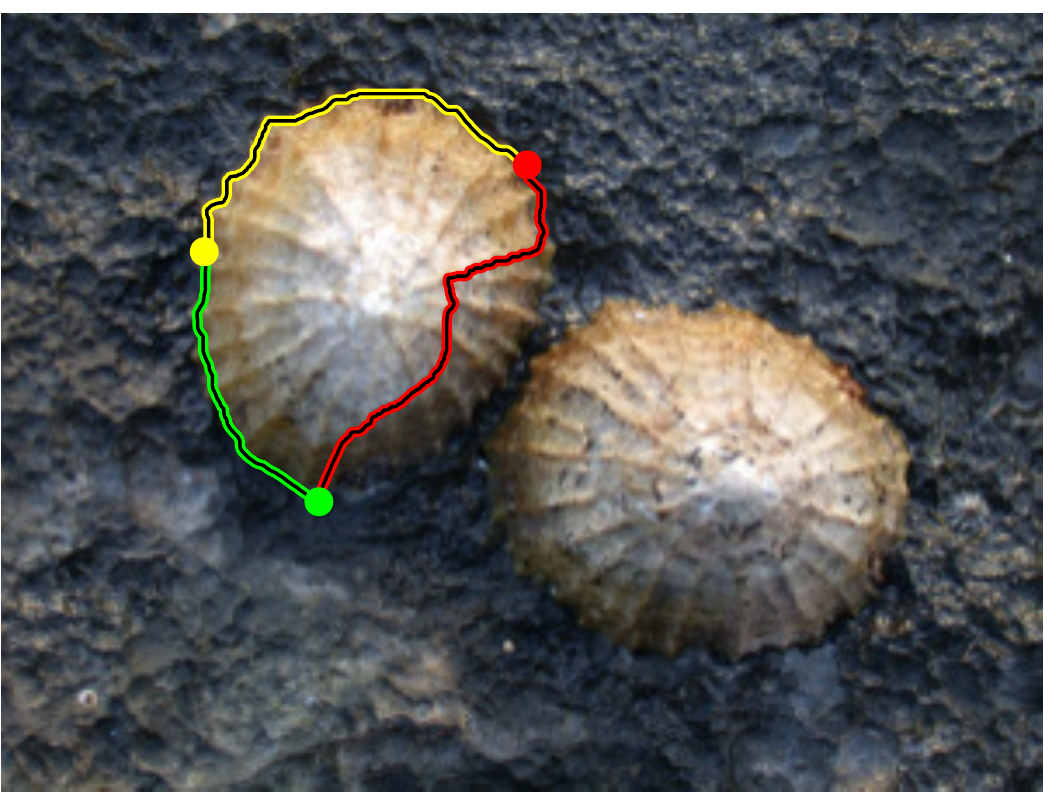}
\includegraphics[width=3.4cm]{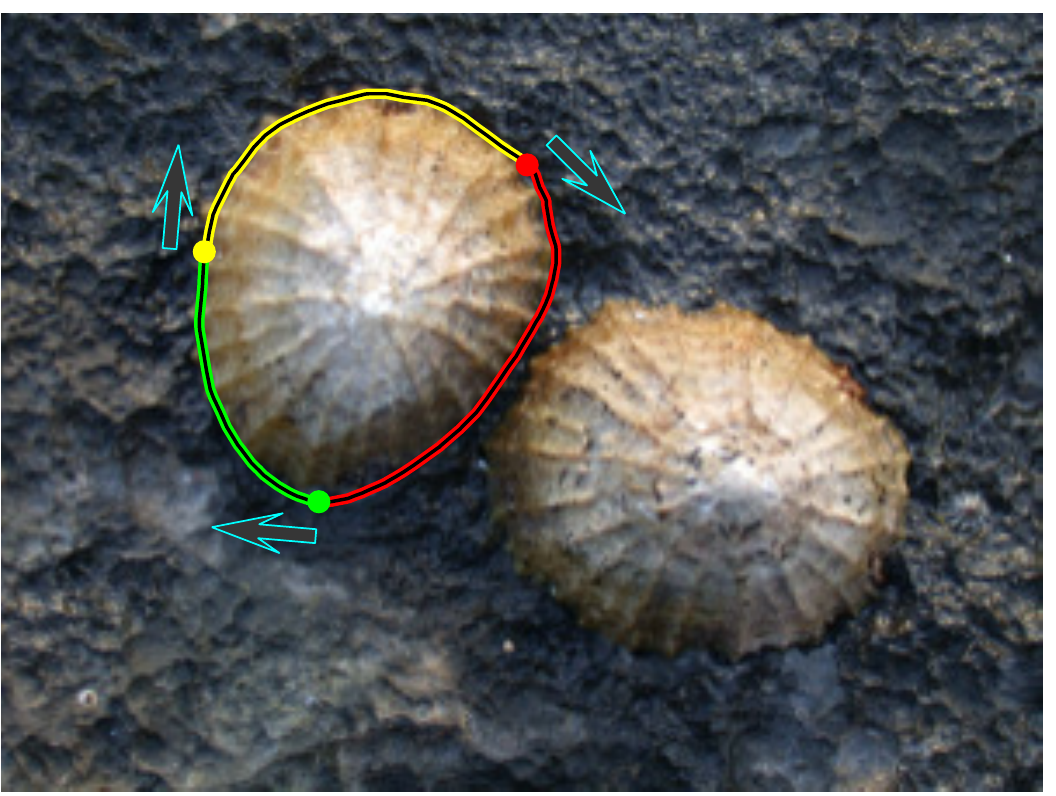}	\\
\includegraphics[width=3.4cm]{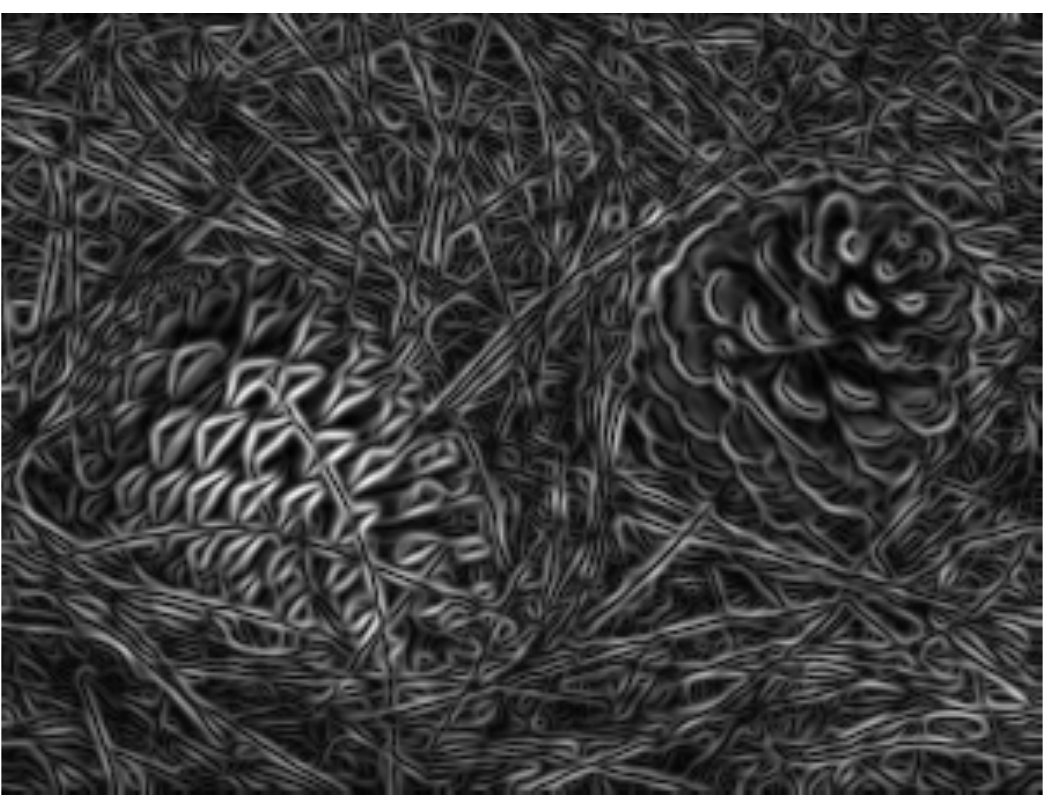}
\includegraphics[width=3.4cm]{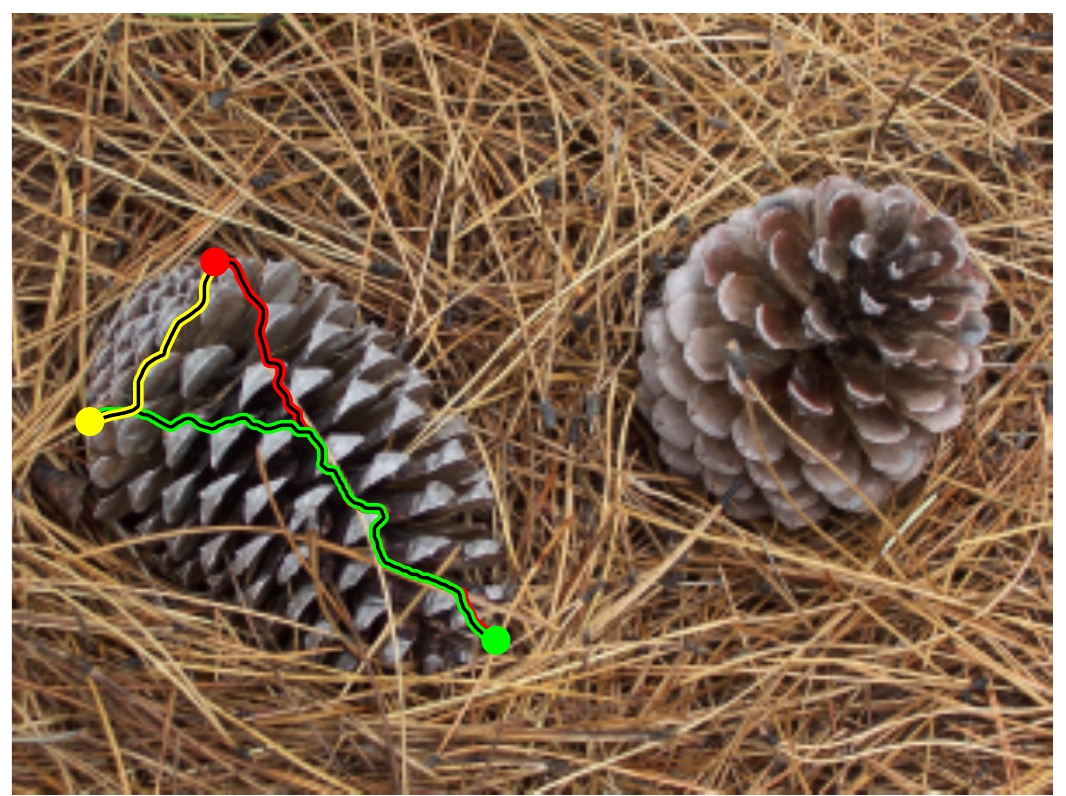}
\includegraphics[width=3.4cm]{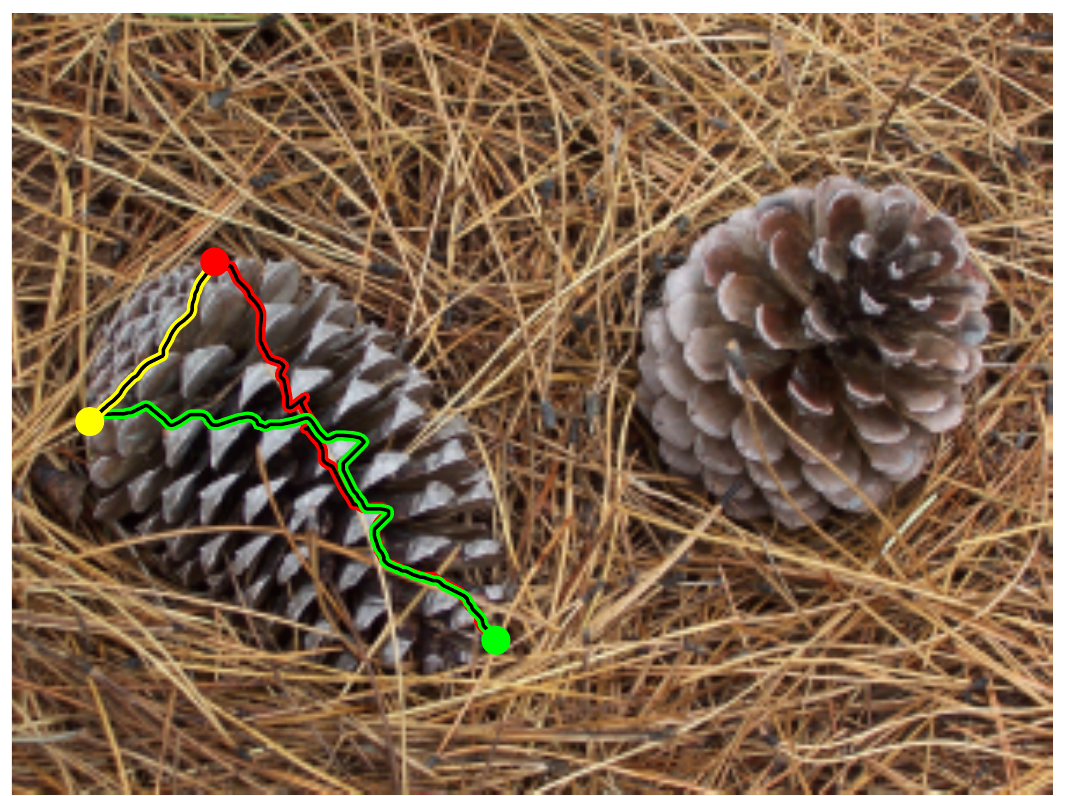}
\includegraphics[width=3.4cm]{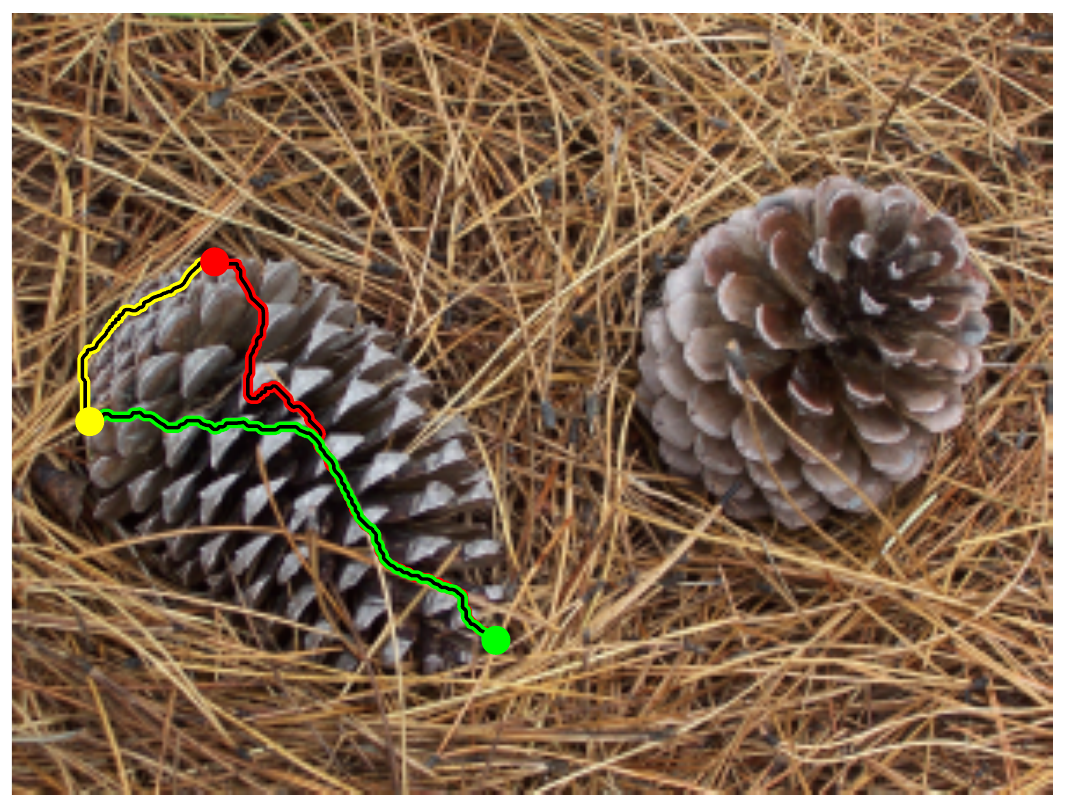}
\includegraphics[width=3.4cm]{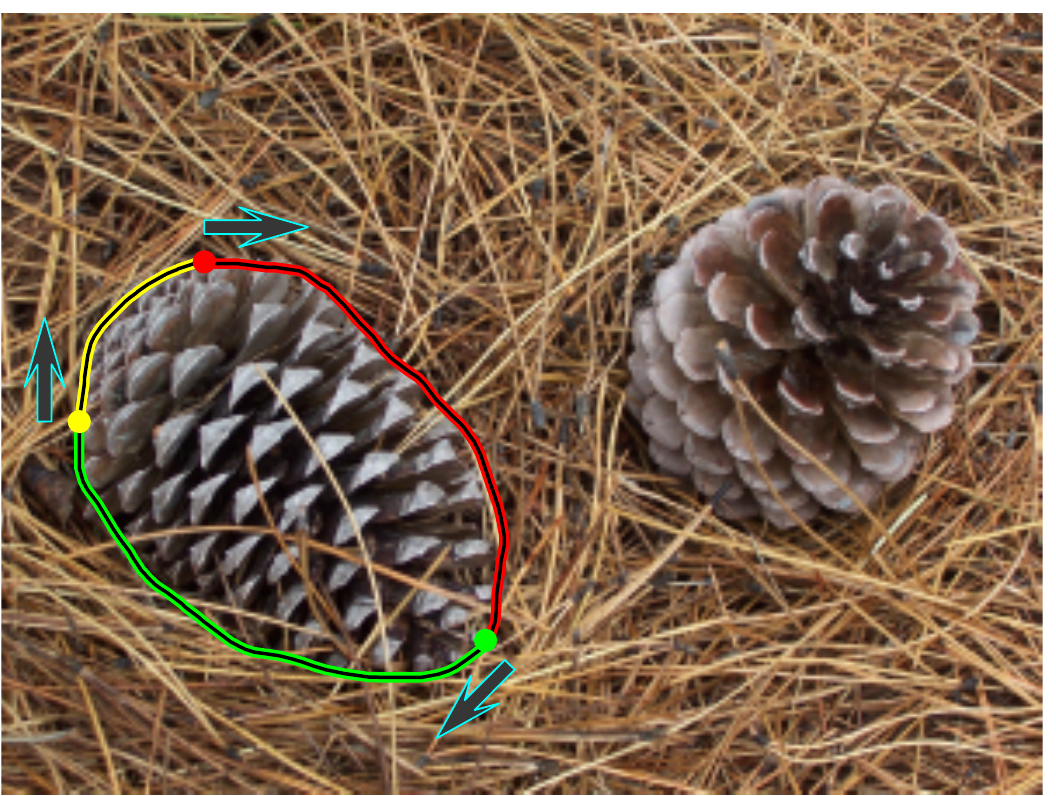}	\\
\includegraphics[width=3.4cm]{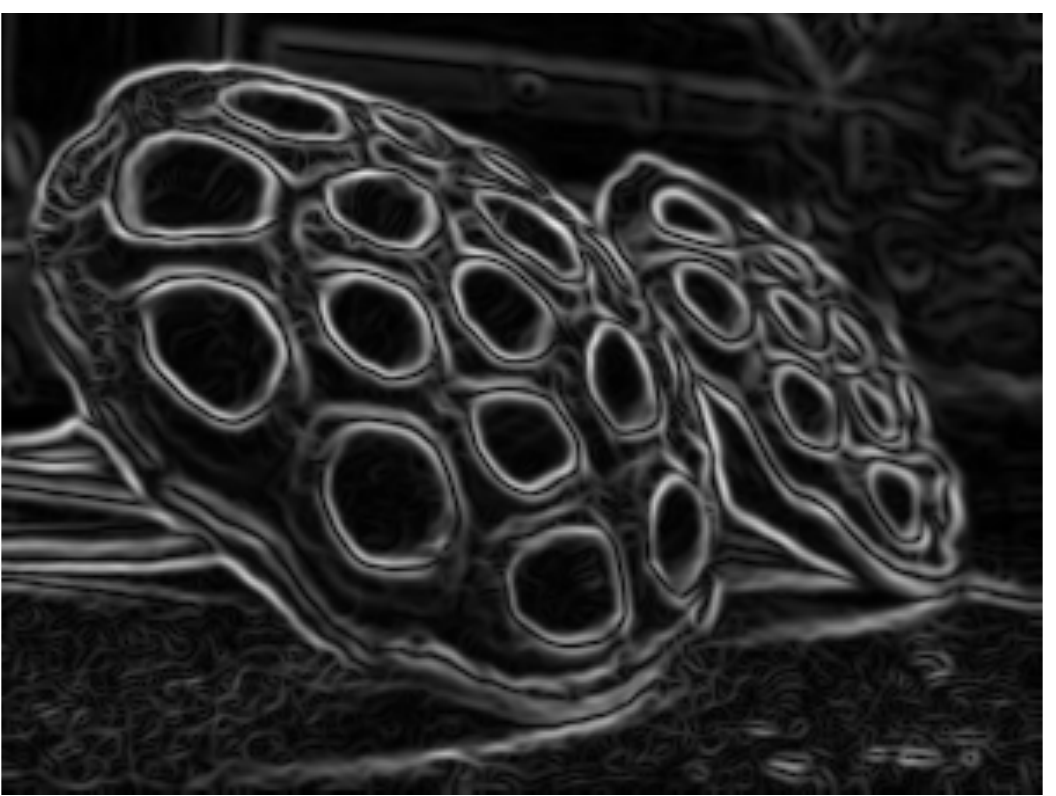}
\includegraphics[width=3.4cm]{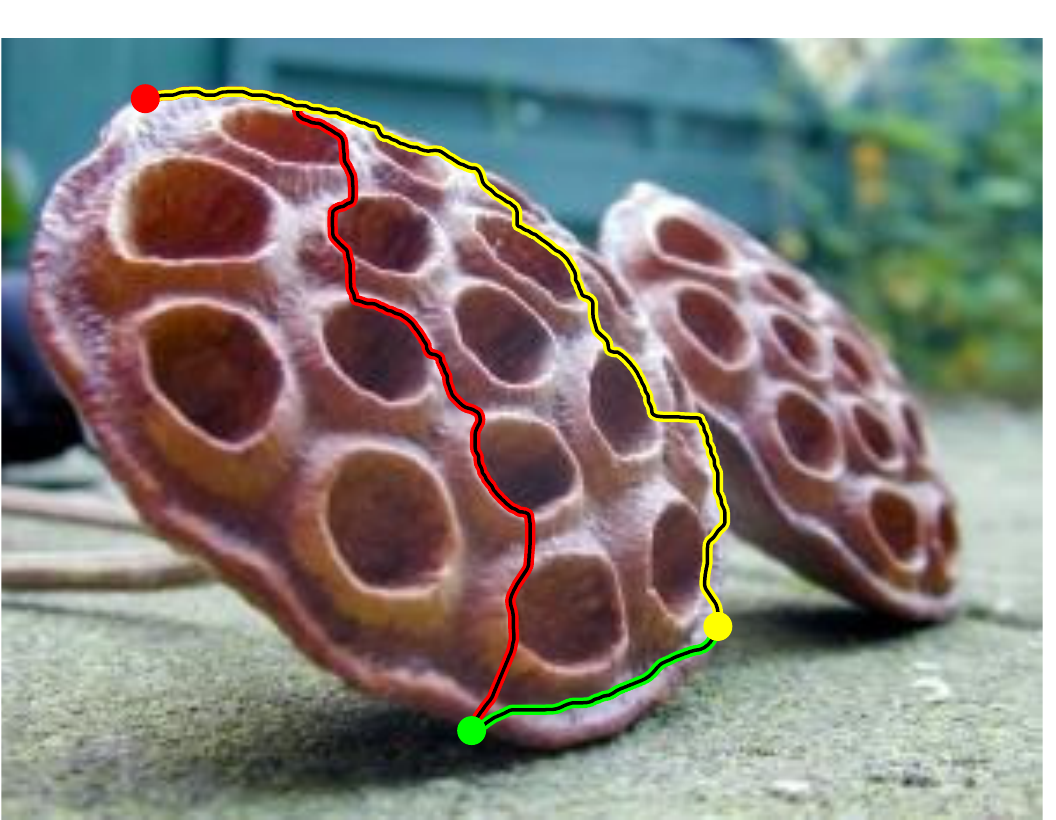}
\includegraphics[width=3.4cm]{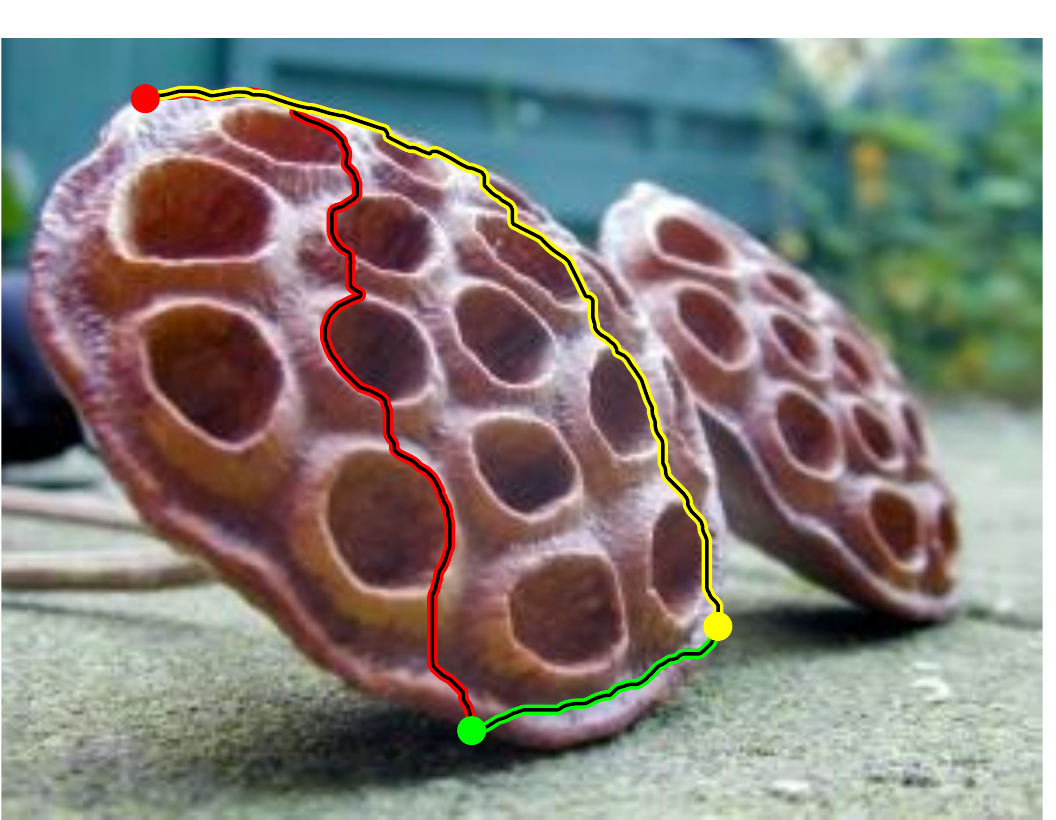}
\includegraphics[width=3.4cm]{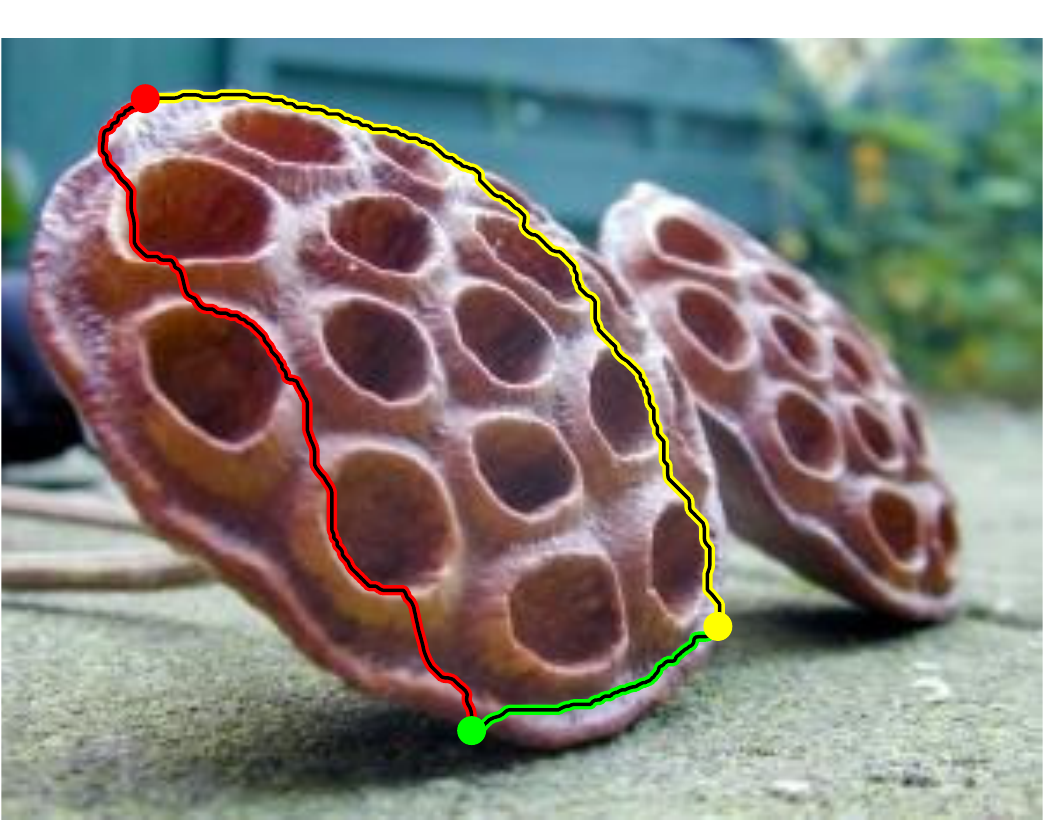}
\includegraphics[width=3.4cm]{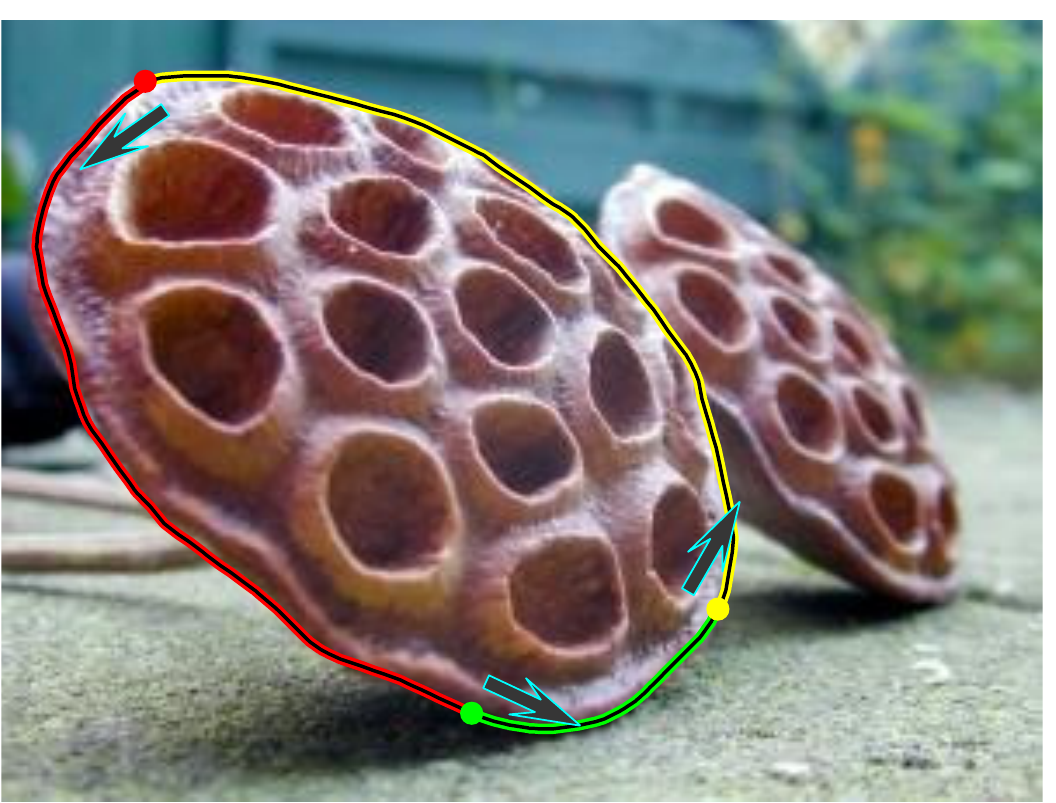}	\\
\includegraphics[width=3.4cm]{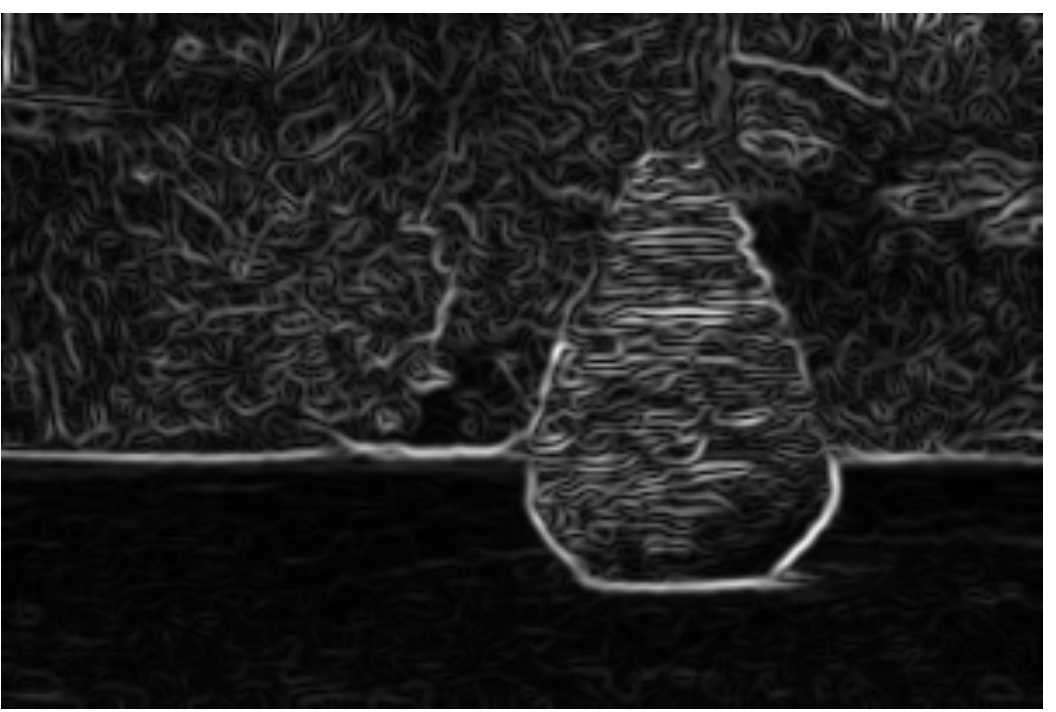}
\includegraphics[width=3.4cm]{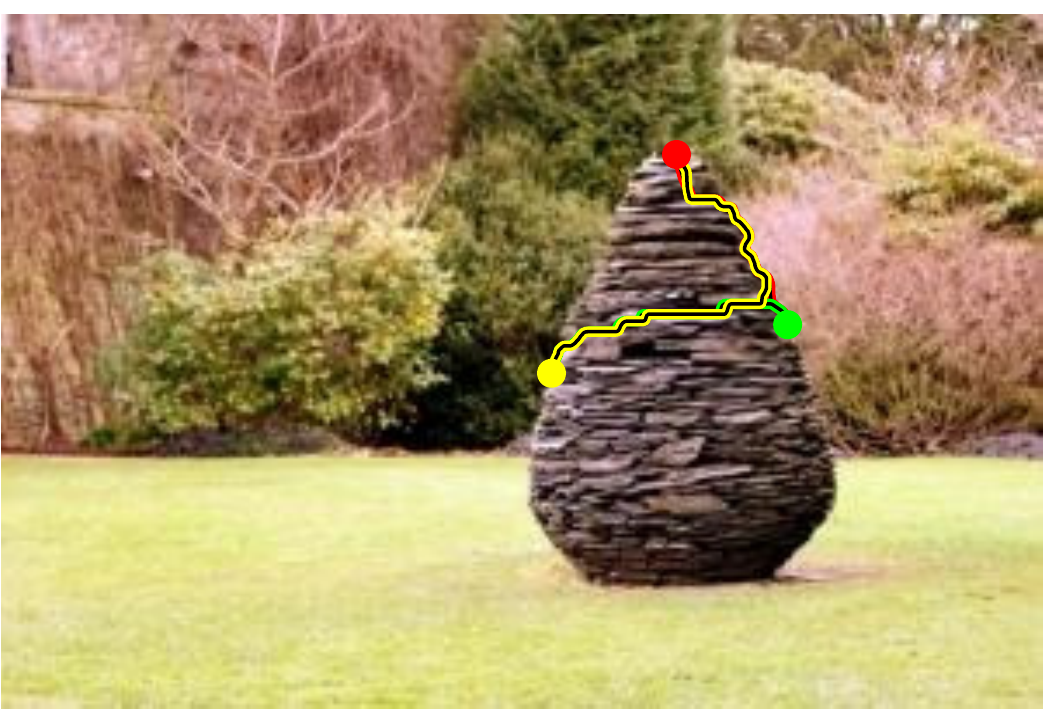}
\includegraphics[width=3.4cm]{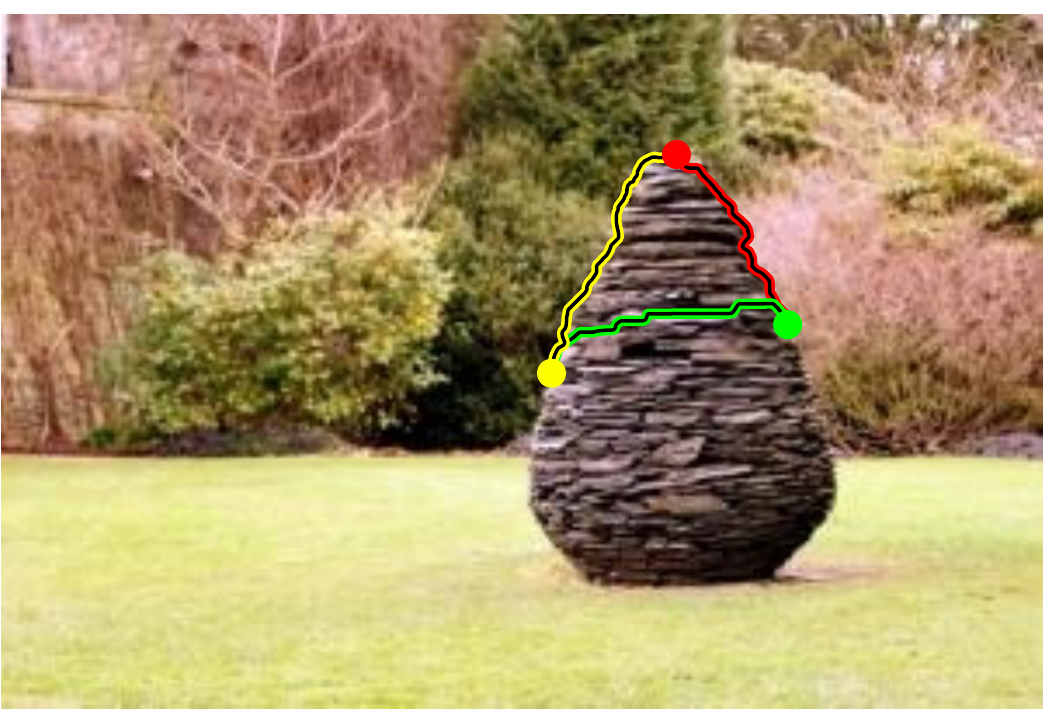}
\includegraphics[width=3.4cm]{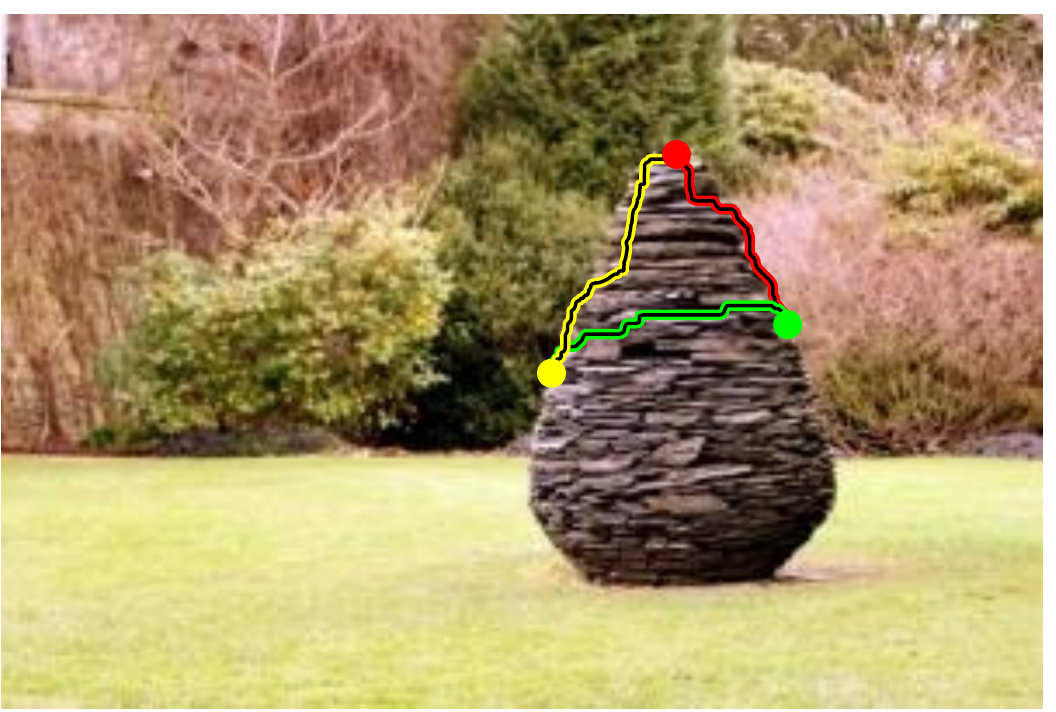}
\includegraphics[width=3.4cm]{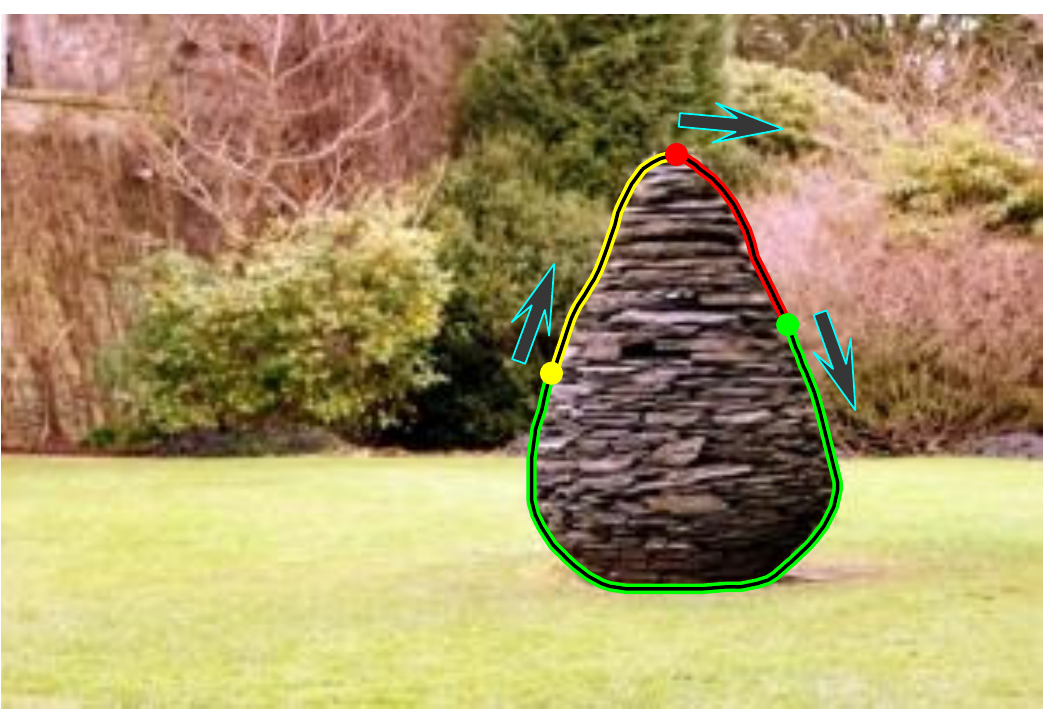}	\\
\includegraphics[width=3.4cm]{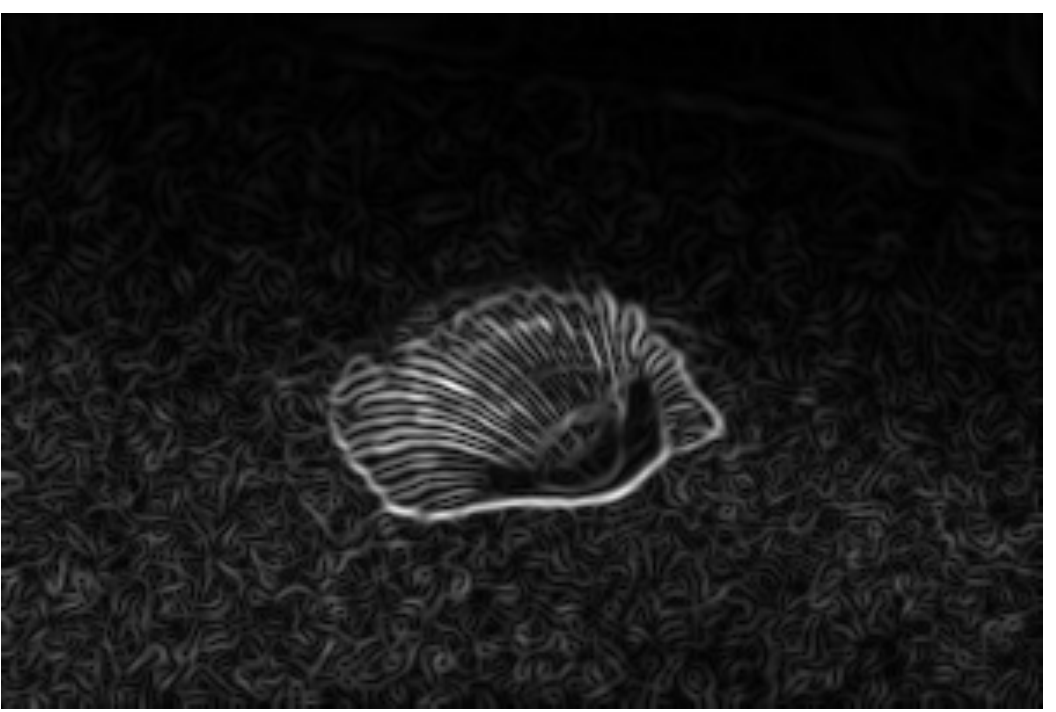}
\includegraphics[width=3.4cm]{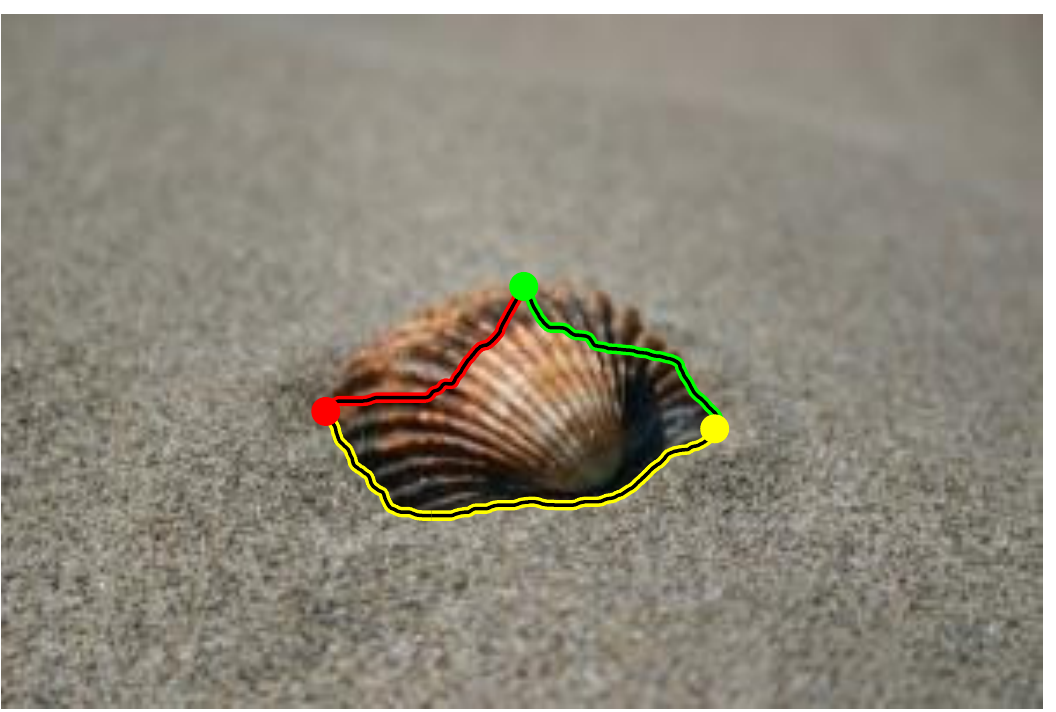}
\includegraphics[width=3.4cm]{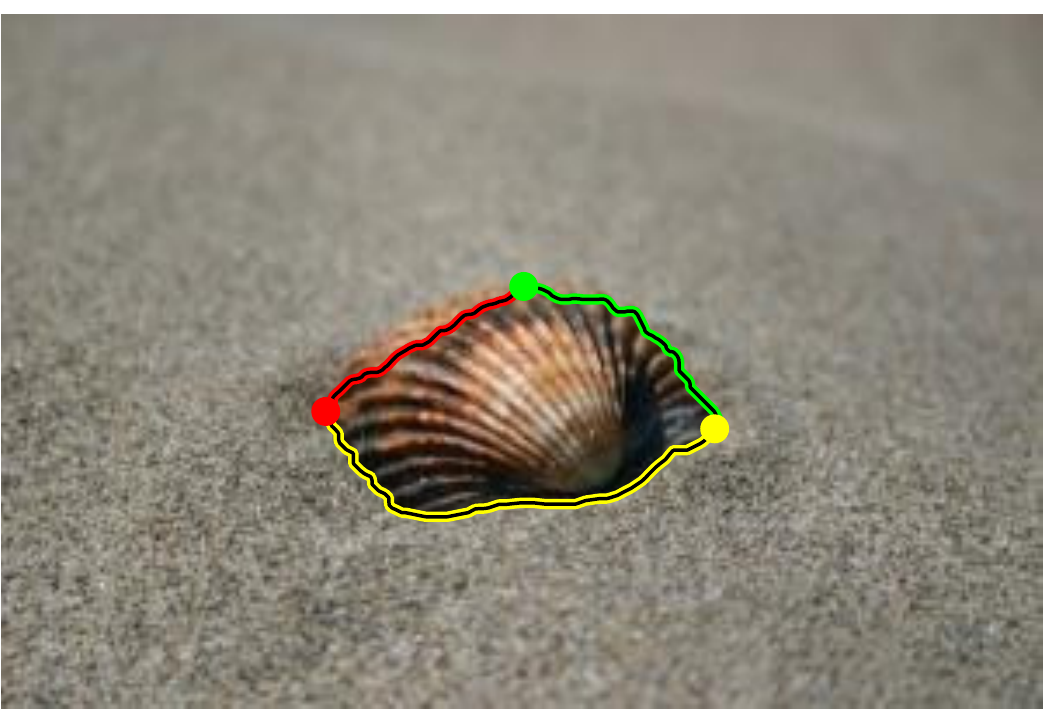}
\includegraphics[width=3.4cm]{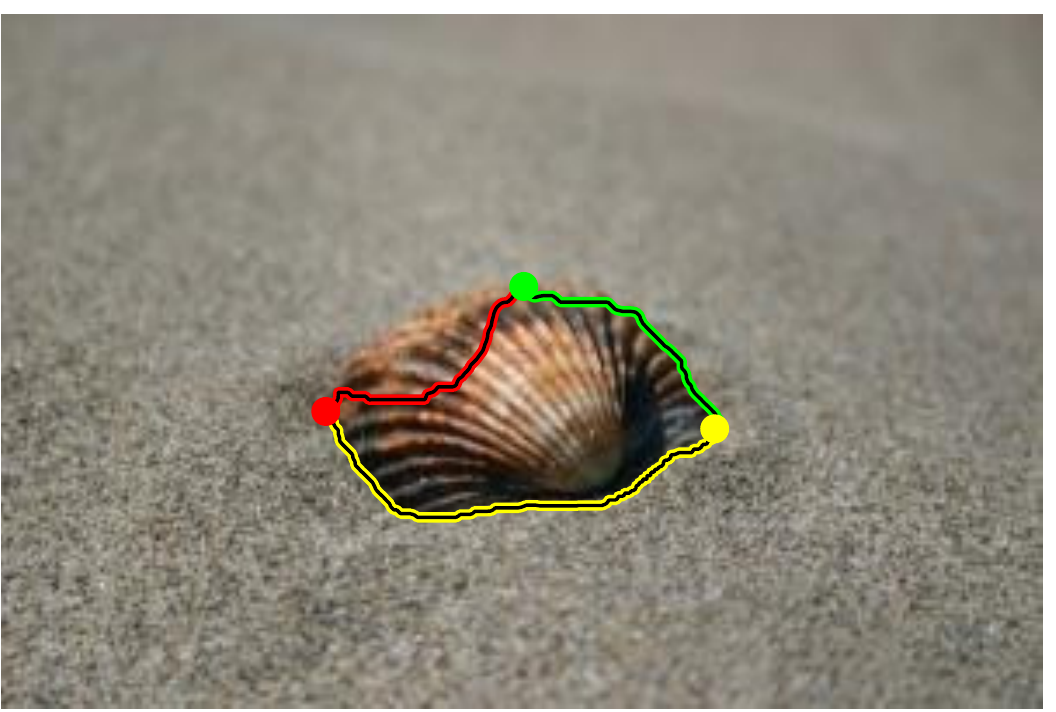}
\includegraphics[width=3.4cm]{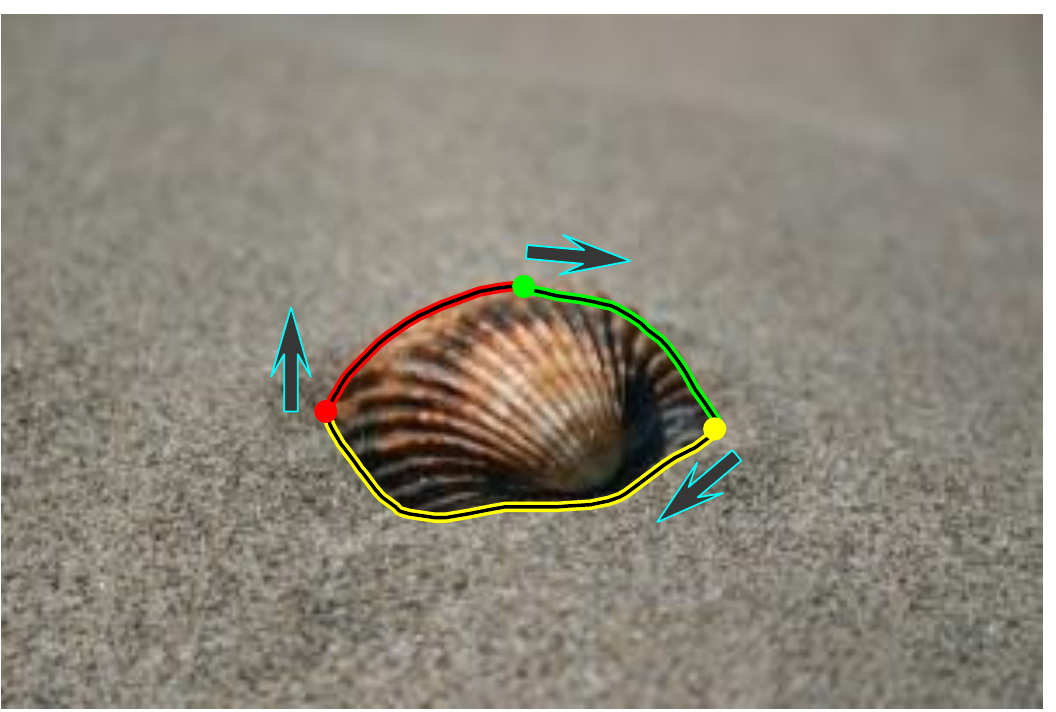}	\\
\includegraphics[width=3.4cm]{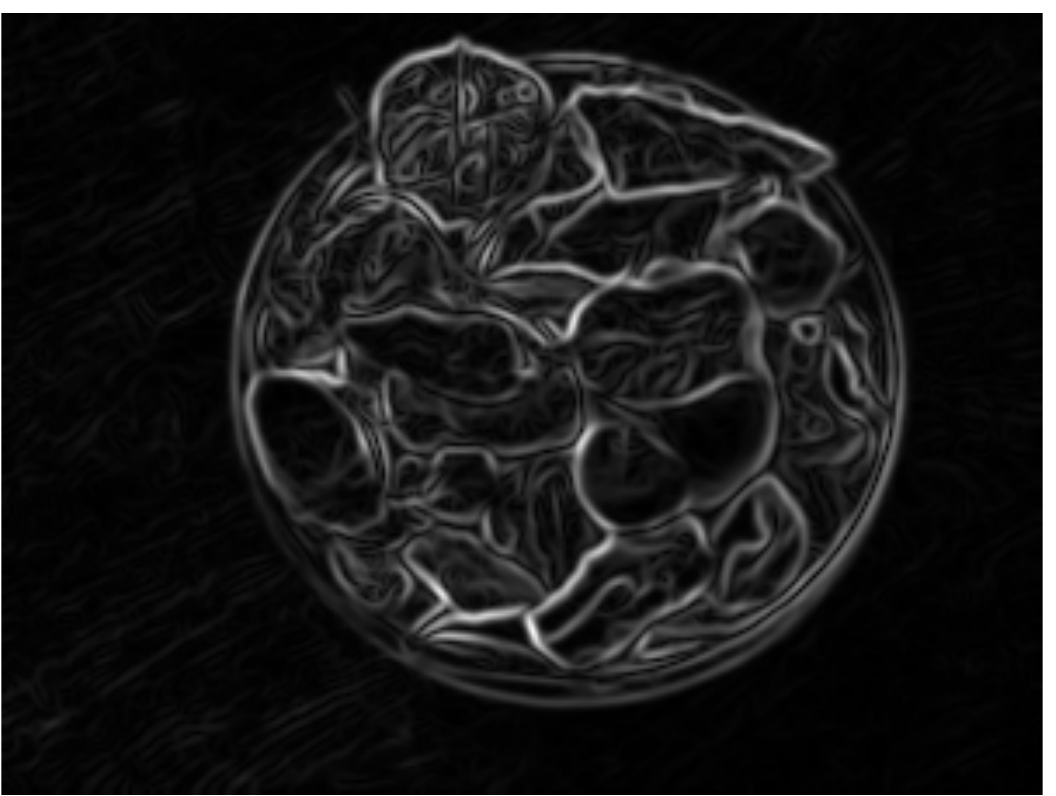}
\includegraphics[width=3.4cm]{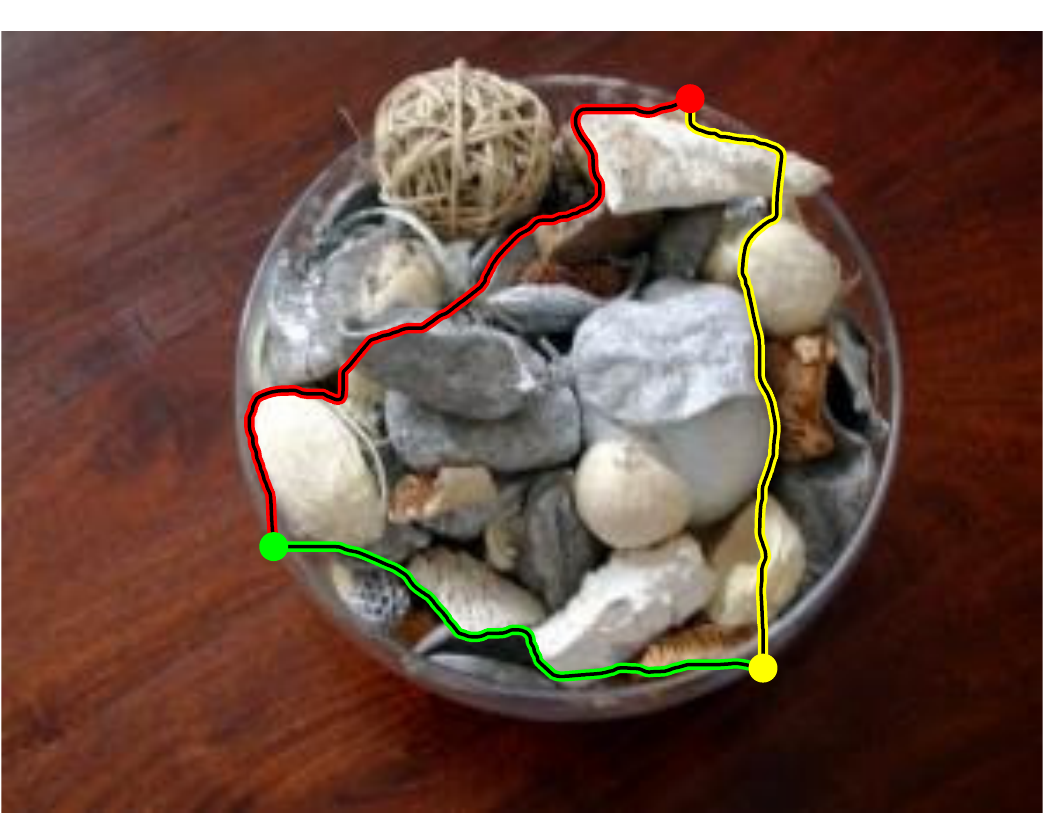}
\includegraphics[width=3.4cm]{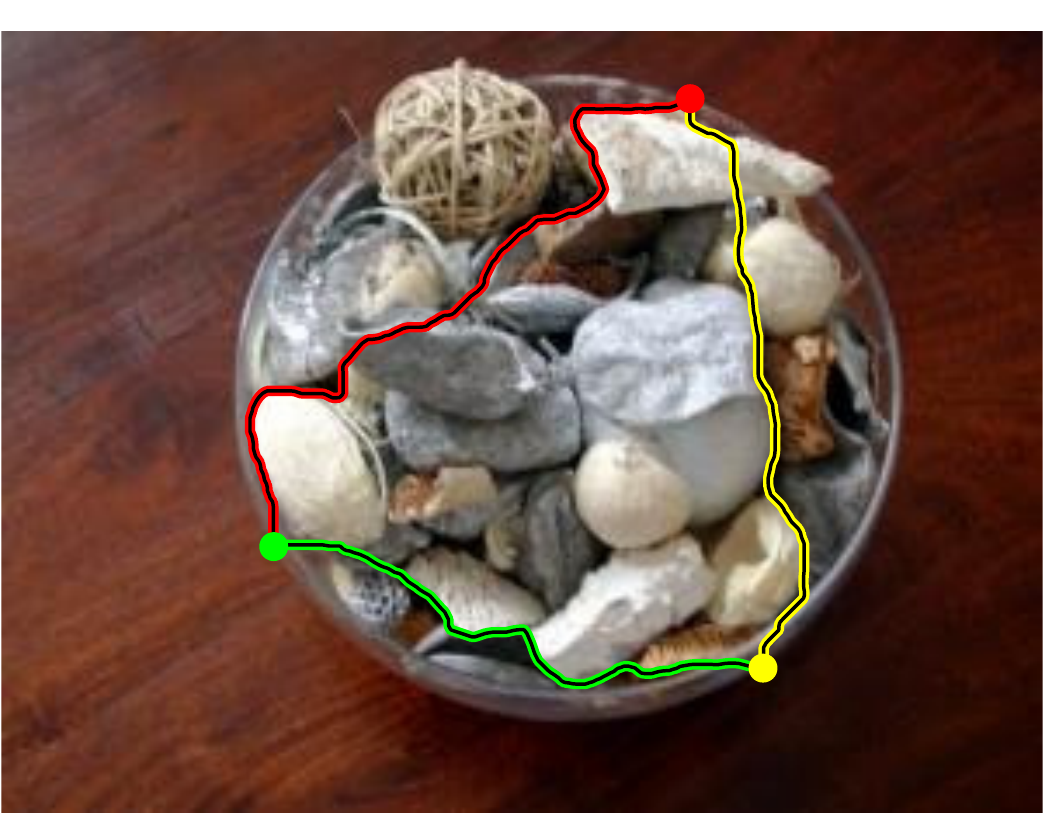}
\includegraphics[width=3.4cm]{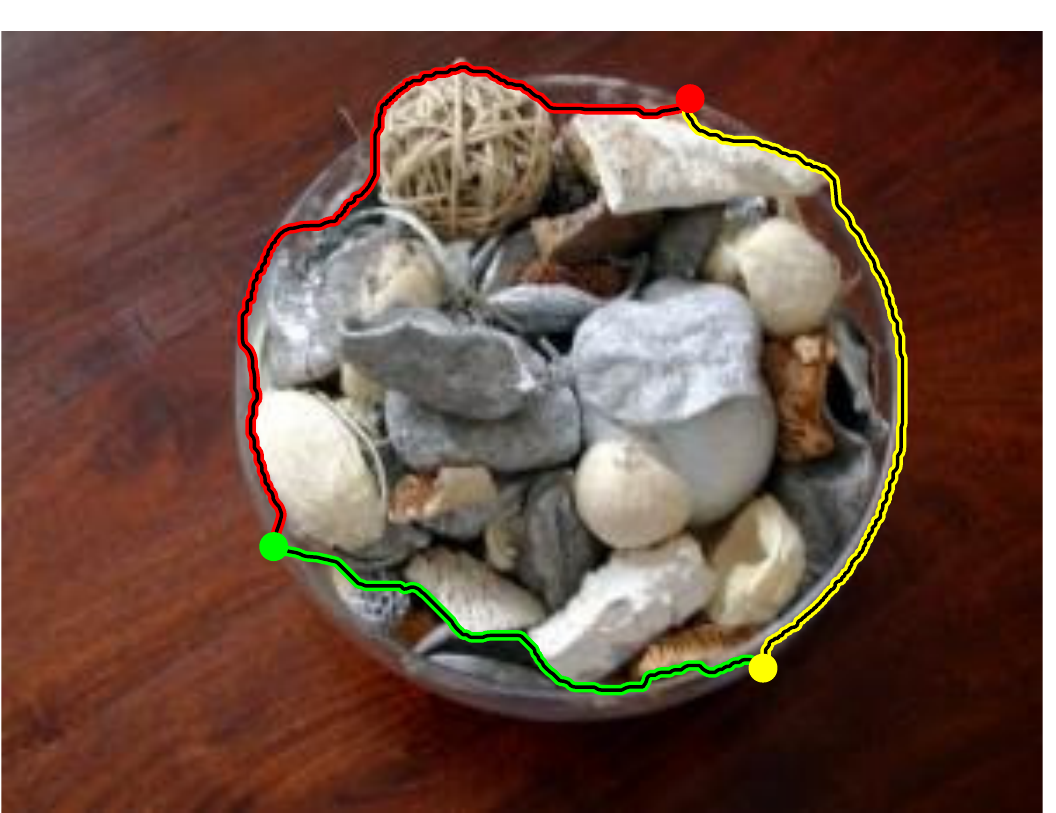}
\includegraphics[width=3.4cm]{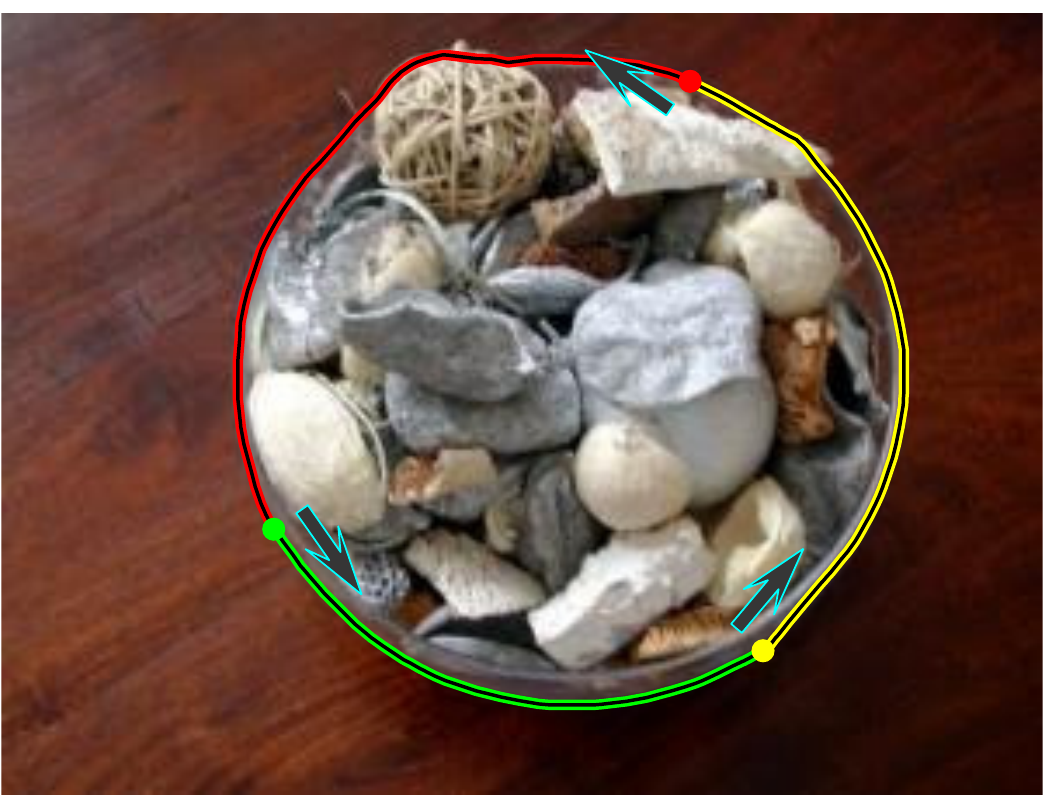}	\\
\includegraphics[width=3.4cm]{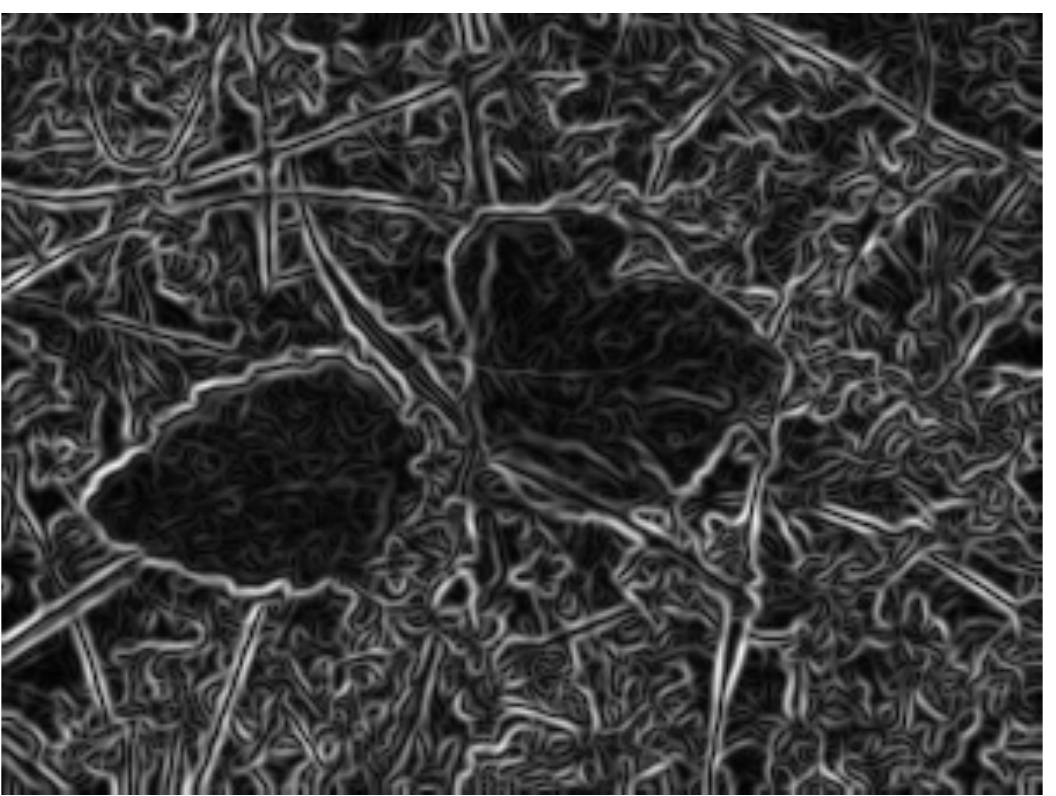}
\includegraphics[width=3.4cm]{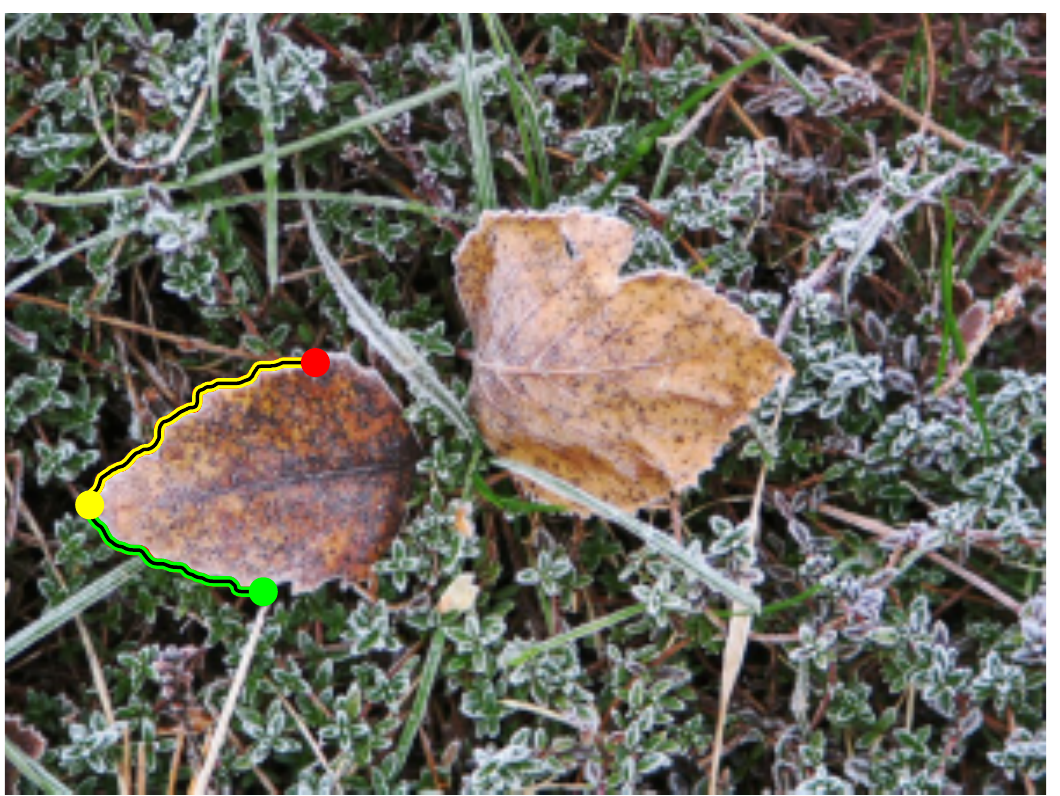}
\includegraphics[width=3.4cm]{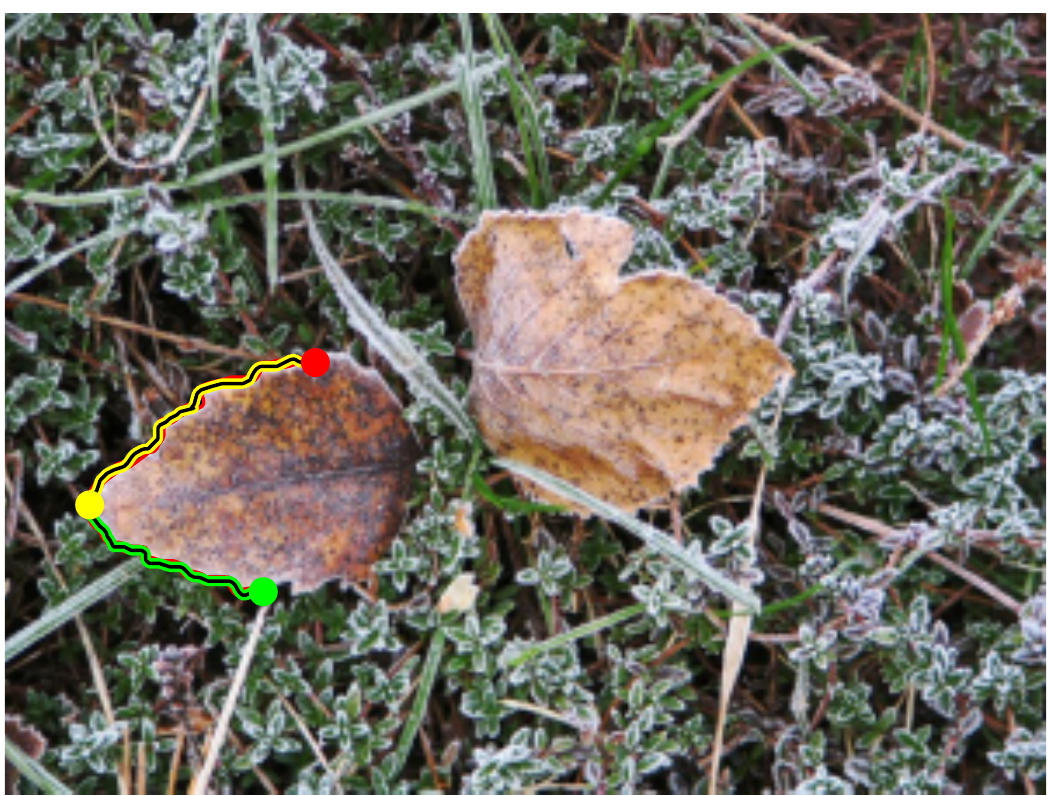}
\includegraphics[width=3.4cm]{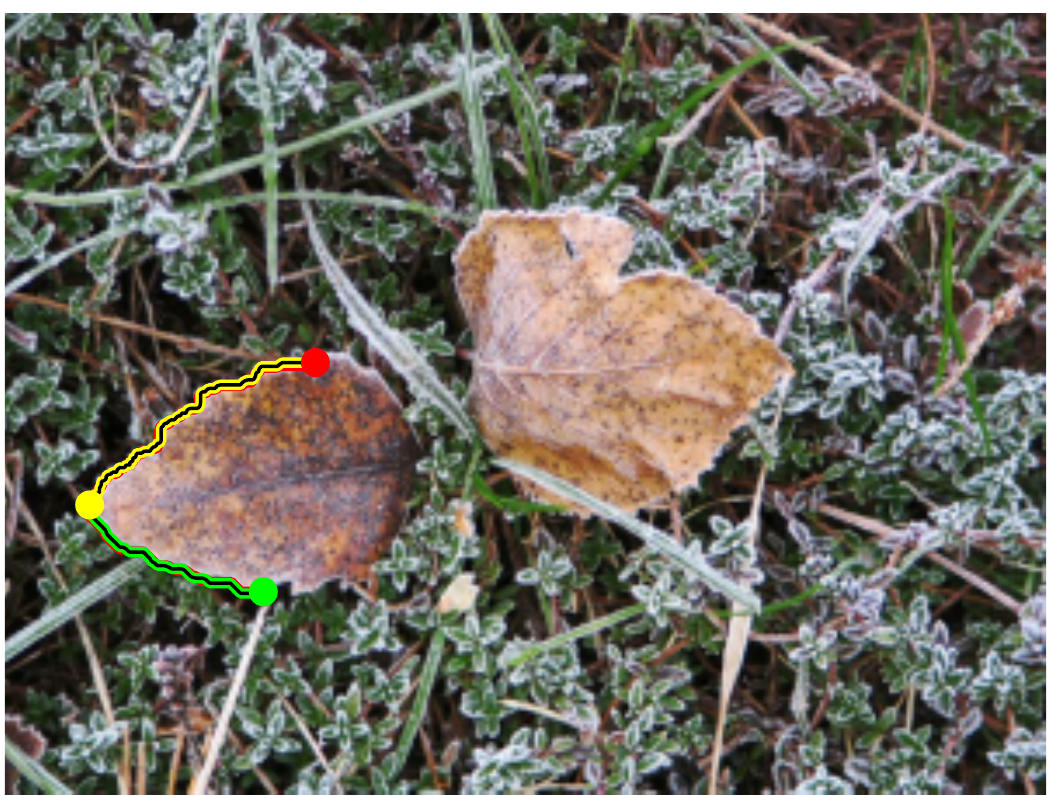}
\includegraphics[width=3.4cm]{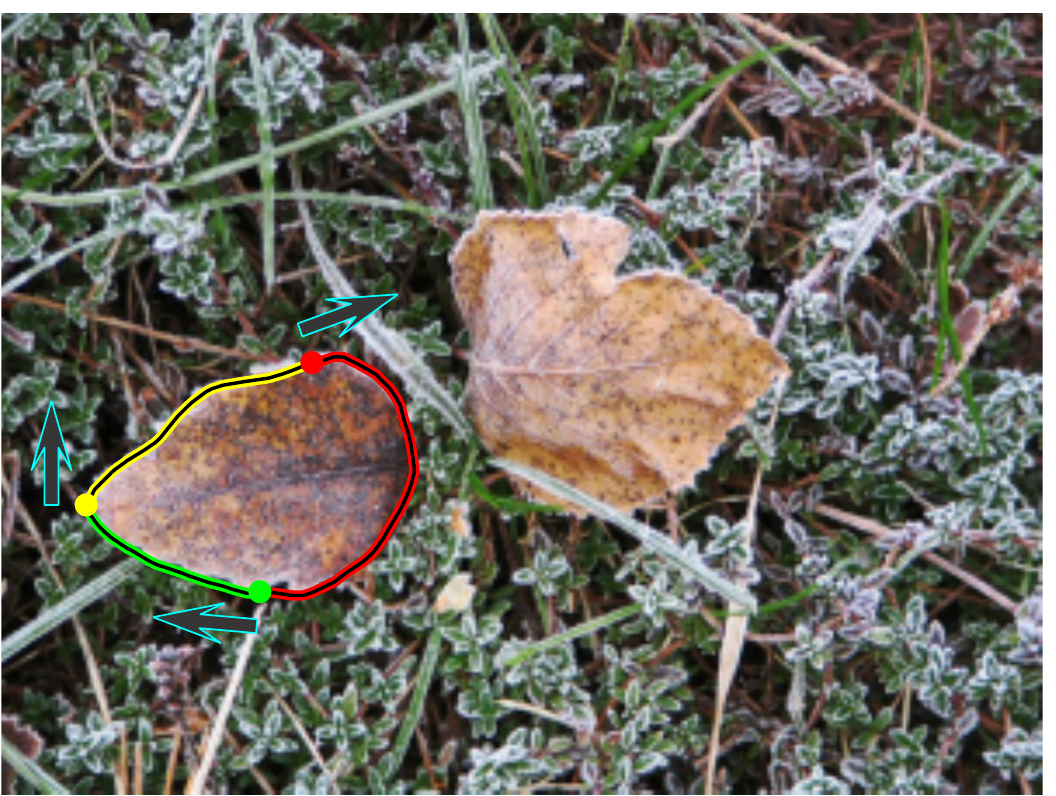}	\\	
\includegraphics[width=3.4cm]{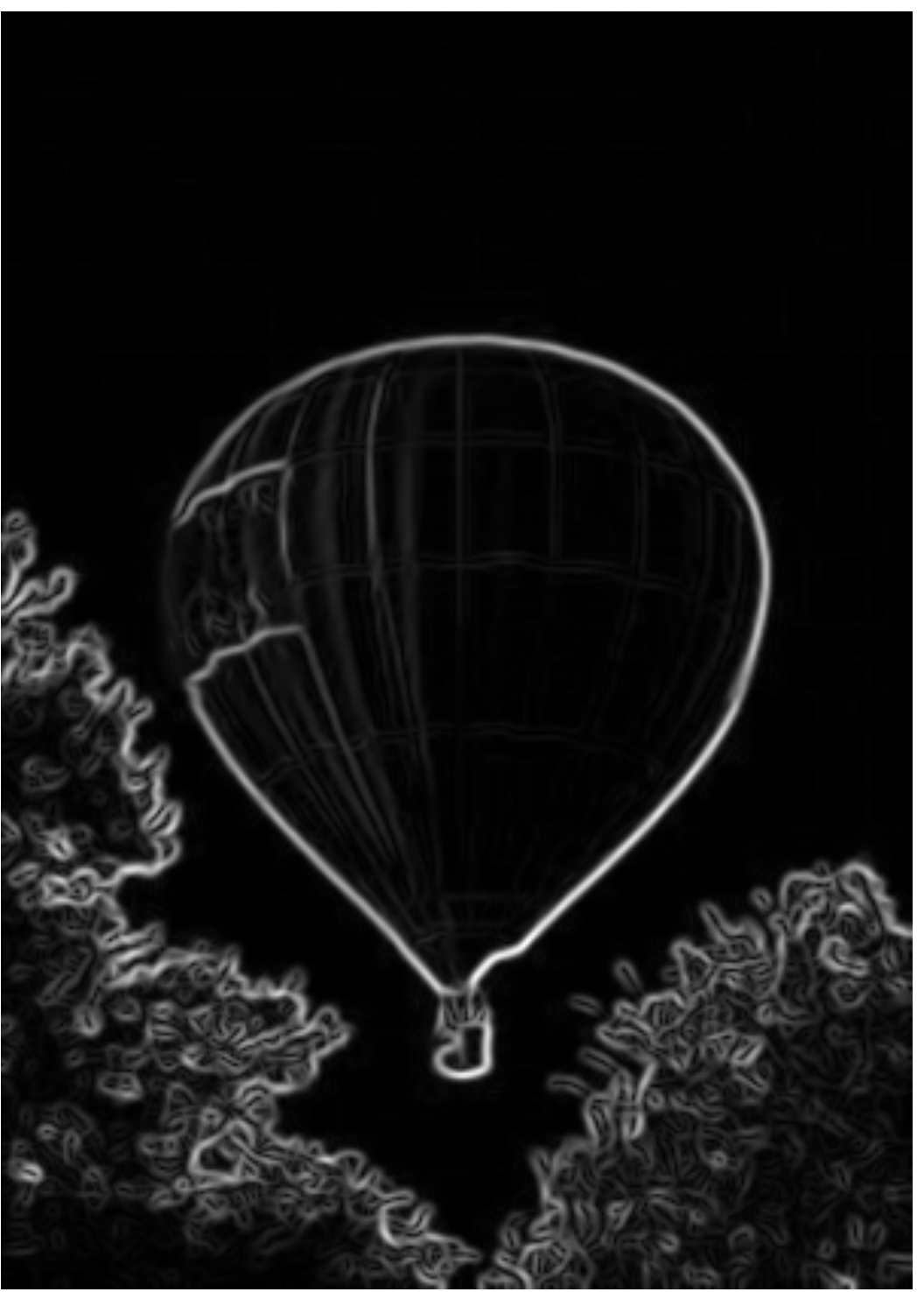}
\includegraphics[width=3.4cm]{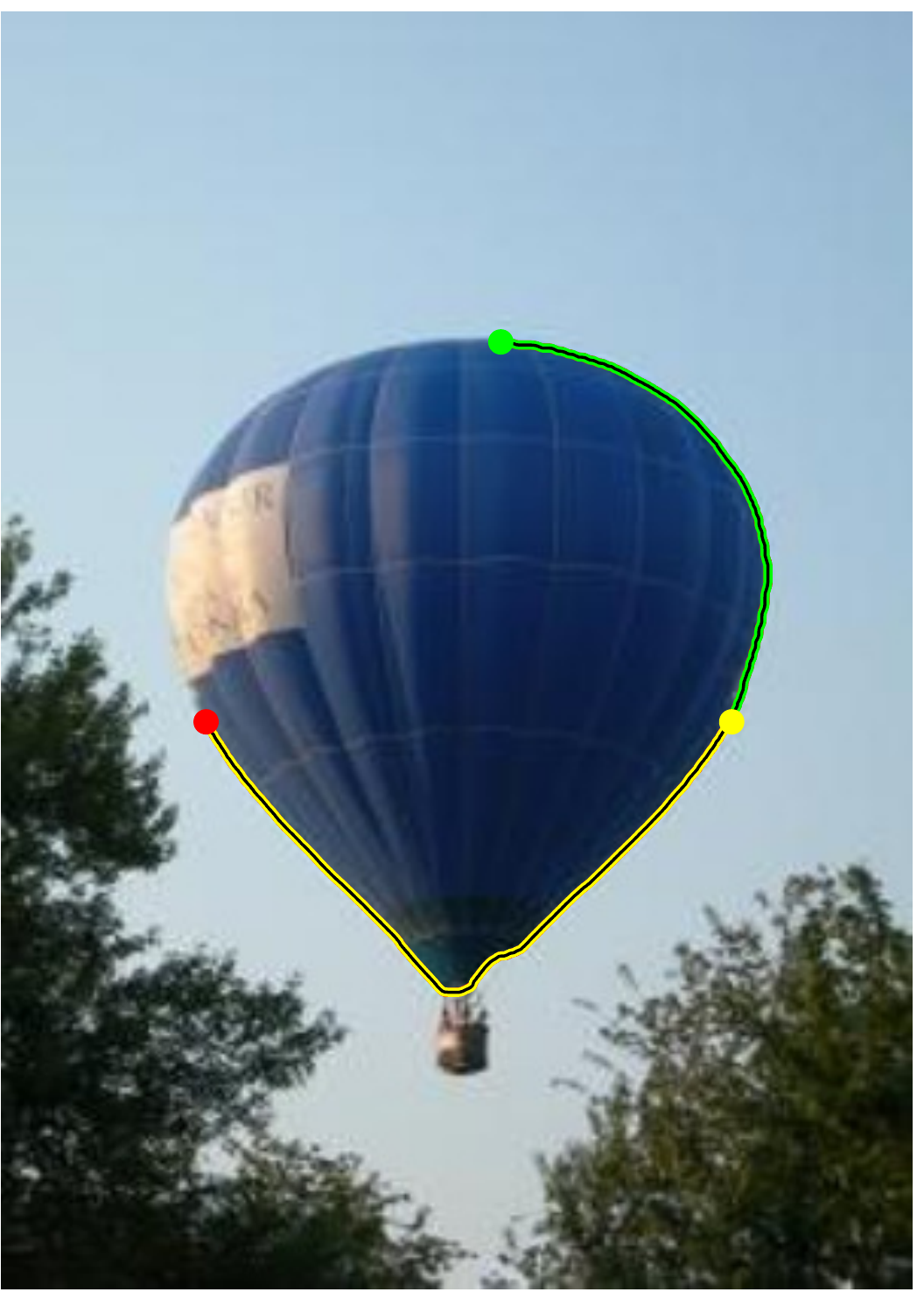}
\includegraphics[width=3.4cm]{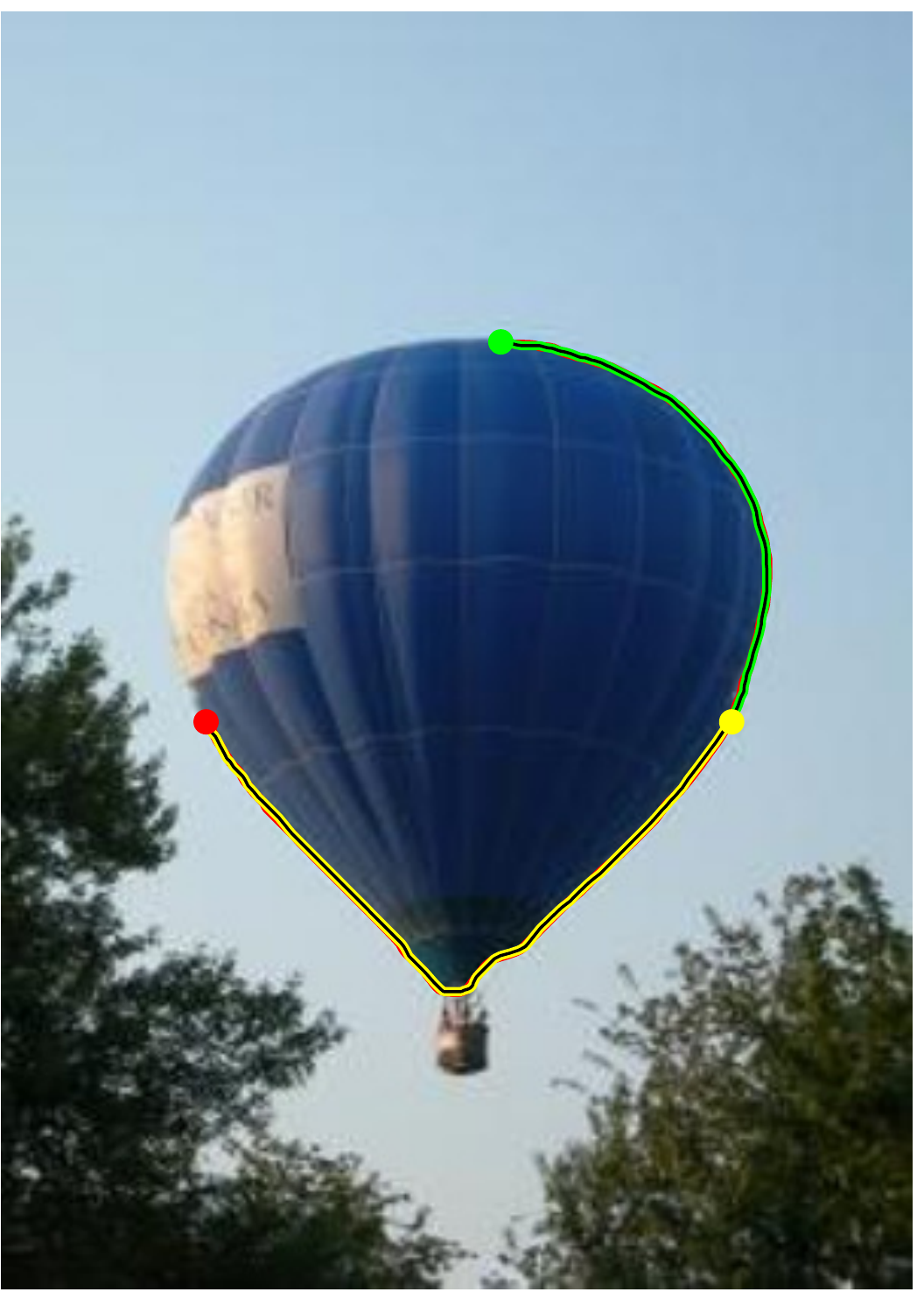}
\includegraphics[width=3.4cm]{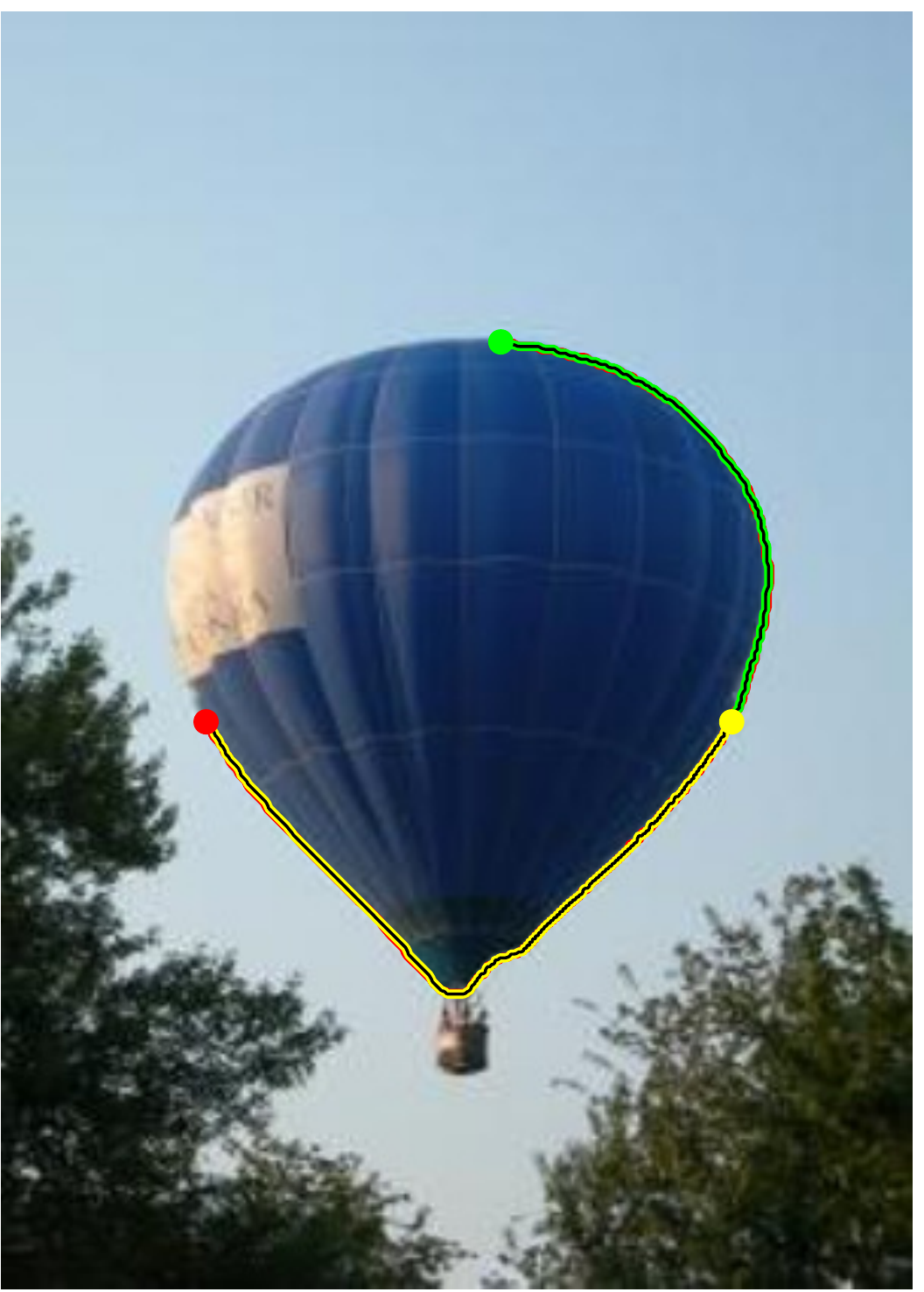}
\includegraphics[width=3.4cm]{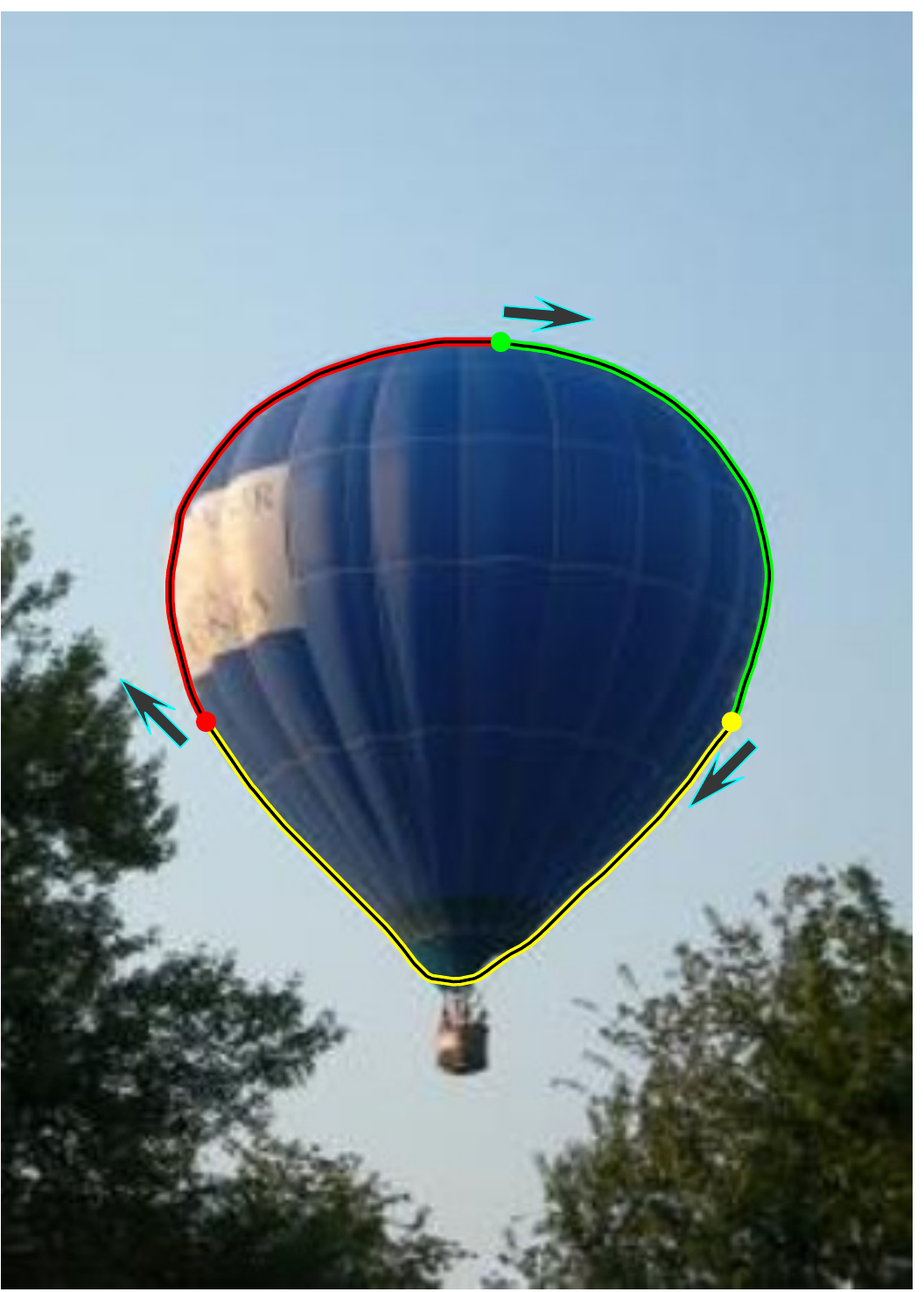}	\\	
\caption{Comparative closed contour detection results obtained by using different metrics. \textbf{Column 1} Edge saliency map. \textbf{Columns 2-5} The closed contour detection results from the IR  metric, the AR metric, the IOLR  metric and the Finsler elastica metric, respectively.}
\label{fig:ContourDetection}
\end{figure*}

\subsection{Closed Contour Detection and Image Segmentation}
\label{subsec:ContourDetectionExperiments}
Fig.~\ref{fig:ContourDetection} shows the closed contour detection results with three  prescribed physical positions using  different metrics, where each position is assigned two opposite orientations~\footnote{For the IR metric and the AR metric, only the physical positions of these orientation-lifted vertices are used.}.  In this experiment, we first  build  the collection $\chi$~\eqref{eq:set_pairs}  by the proposed contour detection procedure using the Finsler elastica metric as described  in Section \ref{subsec:ClosedContour}, where the detection results are shown in column 5. Columns 2-4 show the closed contour detection results using the IR metric, the AR metric, and the IOLR metric, respectively. The minimal paths shown in columns 2 to 4 are obtained by simply linking each pair of vertices by the respective metrics involved in $\chi$. The red, yellow, and green dots are the physical positions of the vertices $\bq_1$, $\bq_2$, and $\bq_3$, respectively. The arrows shown in column 5 indicate the  tangents of the geodesics at the corresponding positions. We assign each geodesic  the same color as its source position. In these images, most parts of the desired boundaries appear to be weak edges, which can be observed from the edge saliency maps shown in column 1. The detected contours associated to the Finsler elastica metric succeed at catching the desired boundaries due to the curvature penalization and asymmetric property. In contrast, the three Riemannian metrics without curvature penalization fail to extract the expected boundaries.  The images used in this experiment are from the Weizmann dataset.

In Fig.~\ref{fig:TwoPointsContourDetaction}, we show the closed contour detection results obtained by the proposed  method with only two given physical positions  and the corresponding orientations. One can see that the proposed method can indeed reduce the user intervention  at least for objects with smooth boundaries.

For the Finsler elastica metric $\cP$ defined in \eqref{eq:FinslerWeighted},  the curvature penalization depends on the parameter $\alpha$ ($\lambda$ is fixed to $100$). 
In Fig.~\ref{fig:ContourWeakParameter}, we show the closed contour detection results by varying $\alpha$ to demonstrate the influence of the curvature term in our approach. In  column 1, we show the closed contour detection results with suitable values of $\alpha$, say $\alpha_0$. In  columns 2 and 3, the closed contour detection results using $\alpha_0/10$ and $5\alpha_0$ are demonstrated. One can see that it could lead to shortcuts by using small values of $\alpha$ in rows 1-3 of column 2. In contrast, with a larger $\alpha$, the detected closed contour can  catch the optimal boundaries of the objects, which  supports the effect of using curvature penalization. The edge saliency maps for each image in this experiments can be found from the first column  of  Fig.~\ref{fig:ContourDetection}.

\begin{figure}[!t]
\centering
\includegraphics[width=8cm]{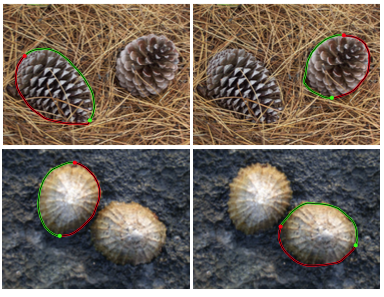}
\caption{Closed contour detection results by using only two given physical positions and the corresponding orientations.}
\label{fig:TwoPointsContourDetaction}
\end{figure}

\begin{figure}[!htb]
\centering
\includegraphics[width=8cm]{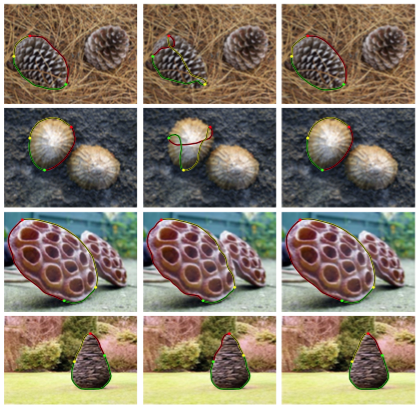}
\caption{Closed contour detection results with different values of  $\alpha$. \textbf{Column 1} Results by the suitable values of $\alpha$ and $\lambda$. \textbf{Columns 2}  Results by small values of $\alpha$. \textbf{Column 3} Results by large values of  $\alpha$.}
\label{fig:ContourWeakParameter}
\end{figure}

\subsection{Perceptual Grouping}
The perceptual grouping result on a synthetic noisy image is shown in Fig.~\ref{fig:PerceptualGrouping1}. In  Fig.~\ref{fig:PerceptualGrouping1}a, we demonstrate the original image, which consists of a set of edges. The red and blue dots with arrows are the orientation-lifted points provided by the user as initializations, where the red dot is the selected initial physical position.  Fig.~\ref{fig:PerceptualGrouping1}b shows the perceptual grouping results of the proposed method. The identified orientation-lifted points in the set $\cD_1$ are denoted by red dots with arrows.  The red curves that links the ordered orientation-lifted vertices in $\cD_1$ indicate the expected closed curves. 

Fig.~\ref{fig:PerceptualGrouping2} illustrates the capacity of the proposed  method to handle the  grouping problem with spurious  points. Different initializations are shown in Figs.~\ref{fig:PerceptualGrouping2}a and \ref{fig:PerceptualGrouping2}c, where the red dots  are the selected initial physical positions. Figs.~\ref{fig:PerceptualGrouping2}b and \ref{fig:PerceptualGrouping2}d are the grouping results. 
The red curves indicate the detected closed curves. 

The proposed perceptual grouping method can detect multiple  closed curves by  specifying the number of expected curves. In Fig.~\ref{fig:PerceptualGrouping3}, three  curves are detected by the proposed method.   Column 1 shows different initializations,  where the red dots are the selected initial physical positions. Columns 2 to 4 demonstrate the intermediate grouping results that correspond to different initializations shown in column 1.
The final perceptual grouping results  are shown in column 5. One can claim that our algorithm indeed has the ability to detect curves intersecting with one another. The vertices shown  in columns 2 to 4, denoted by  red dots with arrows, make up the corresponding collections $\cD_1$ to $\cD_4$.

\begin{figure}[!t]
\centering
\includegraphics[width=8cm]{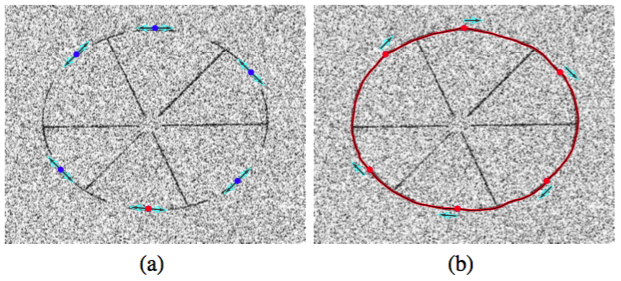}	
\caption{(\textbf{a}) Initialization. The red and blue dots indicate the physical positions, where the red dot is the initial position. (\textbf{b}) Perceptual grouping result.  The arrows denote the tangents for each physical position.}
\label{fig:PerceptualGrouping1}
\end{figure}

\begin{figure*}[!htb]
\centering
\includegraphics[width=17cm]{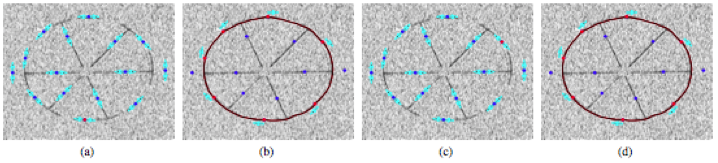}
\caption{Perceptual grouping results.  (\textbf{a}) Initialization 1. Red dot is the selected initial position. (\textbf{b}) Perceptual grouping result for initialization 1.  (\textbf{c})  Initialization 2. (\textbf{d}) Perceptual grouping result for initialization 2.}
\label{fig:PerceptualGrouping2}
\end{figure*}

\begin{figure*}[!t]
\centering
\includegraphics[width =17cm]{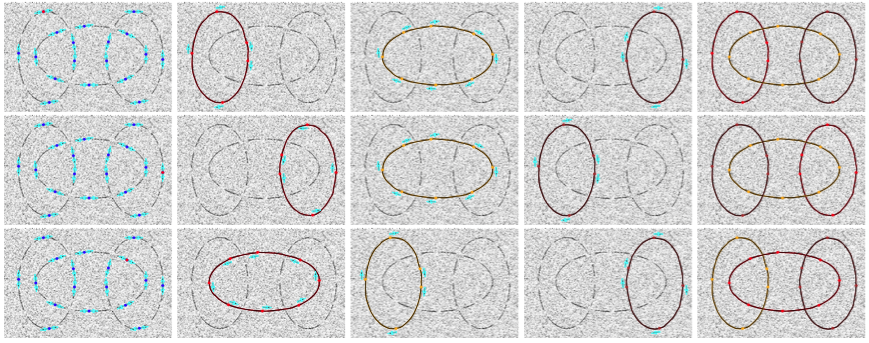}
\caption{Perceptual grouping results by the proposed method where three groups are identified.  \textbf{Column 1}  Initializations.  The red dots denote the selected initial physical positions. \textbf{Columns 2-4} Intermediate  grouping results for the corresponding initializations. \textbf{Column 5} Final grouping results.}
\label{fig:PerceptualGrouping3}
\end{figure*}

\begin{figure*}[!t]
\centering	
\includegraphics[width=17cm]{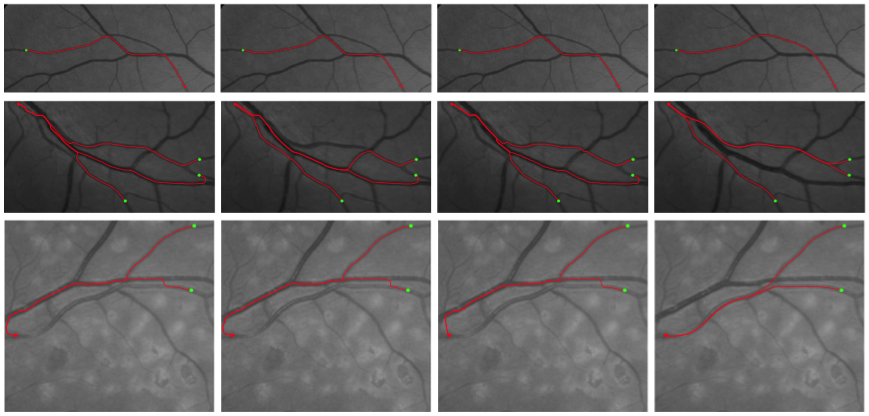}
\caption{Comparative blood vessel extraction results on retinal images. \textbf{Columns 1-4}  The extracted minimal paths using the IR metric, the ARLR metric, the IOLR metric, and the proposed Finsler elastica metric, respectively.}
\label{fig:Retina}
\end{figure*}

\begin{figure*}[!t]
\centering
\includegraphics[width=17cm]{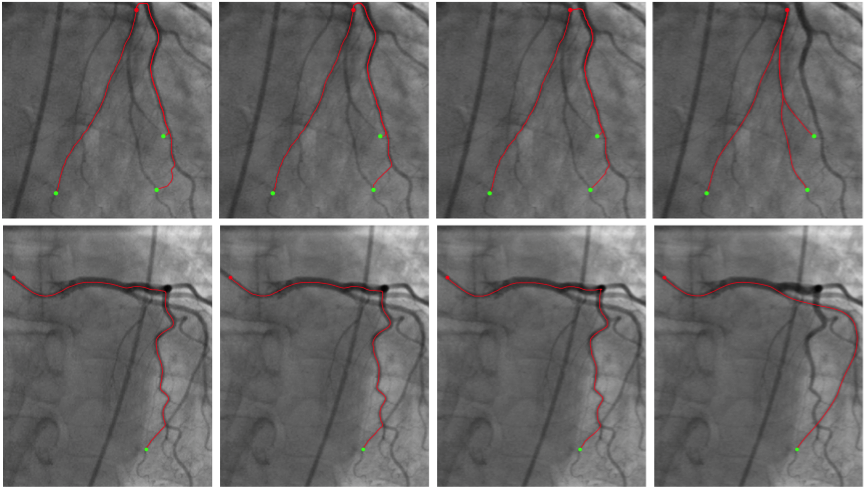}
\caption{Comparative blood vessel extraction results on fluoroscopy images. \textbf{Columns 1-4}  The extracted minimal paths using the IR metric, the ARLR metric, the IOLR metric, and the proposed Finsler elastica metric, respectively.}
\label{fig:Fluoroscopy}
\end{figure*}

\subsection{Tubular Structure Extraction}

In this section, we show the tubular structure  extraction results, where the  source and end positions are indicated by red and green dots, respectively.  In Figs.~\ref{fig:Retina} to \ref{fig:RetinaBlurred}, only the physical positions are provided manually. The corresponding orientations of these physical positions are computed automatically by~\eqref{eq:OptimalOrientation}. In Fig.~\ref{fig:Road}, the corresponding tangents for the physical positions are provided manually because high noise could lead to  failure of  the optimal orientation detection  using \eqref{eq:OptimalOrientation}.
We use the extraction strategy described in Section \ref{subsec:tubular_structure} for the Finsler elastica metric.

In  Fig.~\ref{fig:Retina}, the retinal vessels are extracted by the IR metric, the ARLR metric, the IOLR metric and the Finsler elastica metric as shown in columns 1 to 4, respectively.  In columns 1 to 3,  the minimal paths suffer from the short branches combination problem. In other words, these paths pass through some unexpected vessel segments.  In contrast, the minimal paths obtained by the proposed metric can determine correct combinations of  vessel branches.  

Similar extraction results are observed in Fig.~\ref{fig:Fluoroscopy}. Again, the short branches combination occurs in columns 1 to 3,  where the  minimal paths in these columns are extracted by the IR metric, the ARLR metric and the IOLR metric, respectively. Instead,   the Finsler elastica minimal paths can follow the correct vessel segments as shown in column 4.

In Fig.~\ref{fig:AV_Retina}, we present  the extraction results of the retinal artery centerlines   in three patches of retinal images\footnote{Many thanks to Dr. Jiong Zhang to provide us these images.}.  The centerline of a retinal artery usually appears as a smooth curve. In column 1, we show the retinal artery-vein ground truth maps, where the red and blue regions indicate the arteries and  veins, respectively. Note that  the small vessels have been removed from  the ground truth maps.  Columns 2 to 5 show the centerlines  extraction results  by using the IR metric, the ARLR metric, the IOLR metric and the Finsler elastica metric, respectively.  One can see that the minimal paths demonstrated in columns 2 to 4 pass through  the wrong vessels due to the low gray-level contrast of the retinal arteries. The proposed model can obtain the correct artery centerlines as shown in column 5, thanks to the curvature penalization.

\begin{figure*}[!t]
\centering	
\includegraphics[width=17cm]{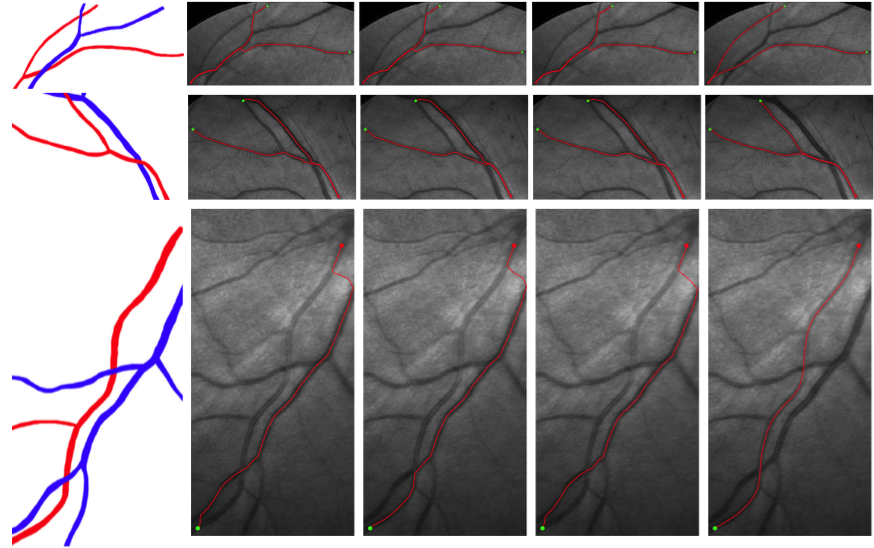}
\caption{Comparative blood vessels extraction results on retinal images. \textbf{Column 1} The retinal artery-vein vessels ground truth maps. \textbf{Columns 2-5}  The extracted minimal paths by the IR metric, the ARLR metric, the IOLR metric, and the  Finsler elastica metric, respectively.}
\label{fig:AV_Retina}
\end{figure*}

\begin{figure*}[!t]
\centering
\includegraphics[width=17cm]{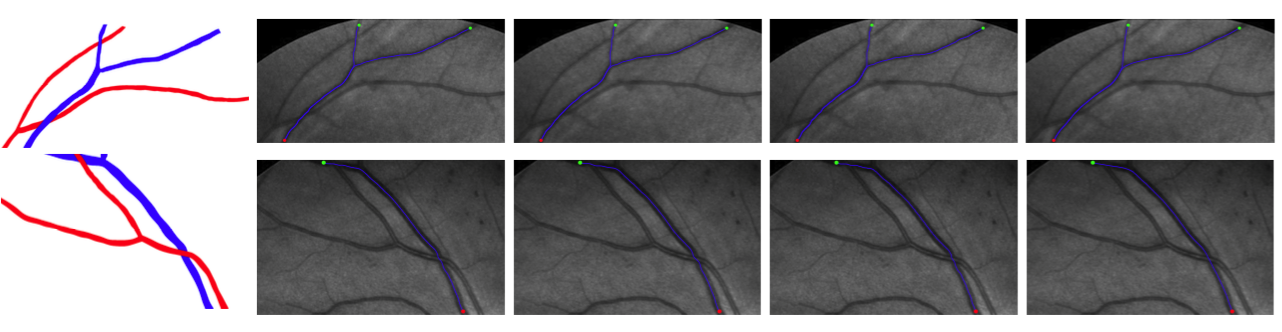}
\caption{Retinal veins extraction results from different metrics. \textbf{Column 1} The retinal artery-vein vessels ground truth maps. \textbf{Columns 2-5}  The extracted minimal paths (blue curves) by the IR metric, the ARLR metric, the IOLR metric, and the  Finsler elastica metric, respectively.}
\label{fig:AV_Veins}
\end{figure*}

\begin{figure*}[!t]
\centering	
\includegraphics[width=17cm]{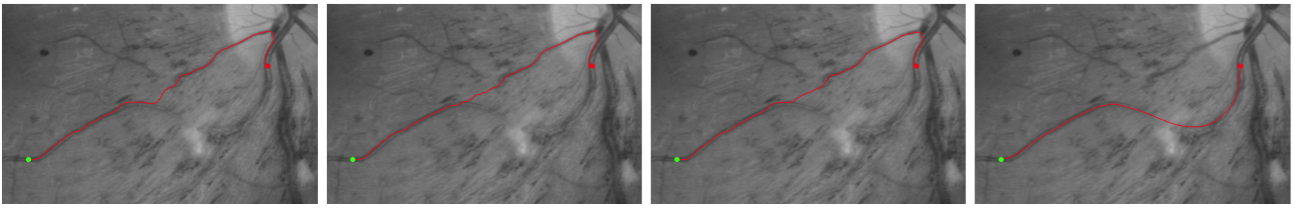}
\caption{Comparative blood vessel extraction results on a blurred retinal image. \textbf{Columns 1-4}  The extracted minimal paths by  the IR metric, the ARLR metric, the IOLR metric, and the  Finsler elastica metric, respectively.}
\label{fig:RetinaBlurred}
\end{figure*}

\begin{figure*}[!t]
\centering
\includegraphics[width=17cm]{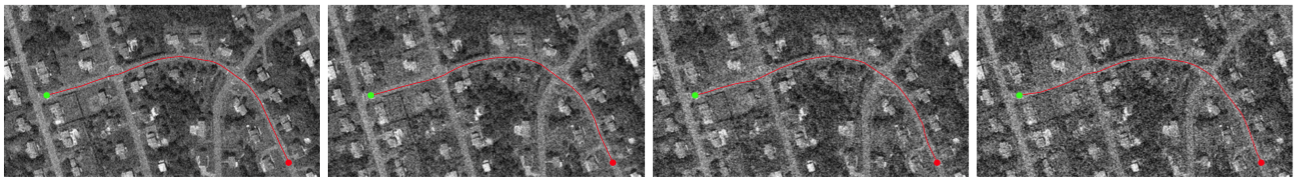}	
\caption{ Roads extraction results by the proposed Finsler elastica metric in an aerial image blurred by Gaussian noise.}
\label{fig:Road}
\end{figure*}

In Fig.~\ref{fig:AV_Veins}, we demonstrate the retinal vein extraction results (blue curves) in the same patches which are used in rows 1-2 of Fig.~\ref{fig:AV_Retina}. The extracted minimal paths (blue curves) are shown in columns 2 to 5 by using  the IR metric, the ARLR metric, the IOLR metric and the  Finsler elastica metric, respectively. From this figure, we can see that all of the extracted minimal paths can successfully follow the retinal veins as indicated by red regions in the artery-vein ground truth maps shown in column 1. 

In  Fig.~\ref{fig:RetinaBlurred}, the vessel extraction results on a patch of a retinal image are demonstrated, where the vessels are blurred by other tissues. The minimal paths  using the three Riemannian metrics fail to extract the desired vessel as shown in columns 1 to 3. In contrast, the Finsler elastica minimal paths can successfully delineate  the targeted vessel as shown in column 4,  thanks to the curvature penalization.

In Fig. \ref{fig:Road}, we show the road segmentation results on an aerial image by the proposed Finsler elastica metric. The road images are blurred by  Gaussian noise with different variances. One can claim that our method is able to obtain smooth and accurate minimal paths on noisy images.

\section{Conclusion}
 \label{sec:Conclusion}
The core contribution of this paper lies at the introduction of  the curvature penalization to the Eikonal PDE-based minimal path framework, which is accomplished by  establishing the  connection between the Euler elastica bending energy and the geodesic energy via a family of orientation-lifted Finsler elastica metrics. Solving the Eikonal PDE with respect to the Finsler elastica metric, the minimal path model thus can determine globally minimizing curves that blend the benefits of the orientation lifting and the curvature penalization. Given a set of the user-provided orientation-lifted points, we have demonstrated the ability of the proposed model for interactive image segmentation, perceptual grouping and tubular structure extraction. The experimental results on both synthetic and real images demonstrate the advantages of the Finsler elastica minimal paths approach.

\begin{acknowledgements}
The authors would like to thank all the anonymous reviewers for their detailed remarks that helped us improve the presentation of this paper. This work was partially supported by ANR grant NS-LBR,  ANR-13-JS01-0003-01.	
\end{acknowledgements}

\section*{Appendix A: Fast Marching Method}

\begin{appendix}
\label{AppendixA}
Basically, the fast marching method introduces a boolean map $b:Z\to\{\text{\emph{Trial}}, \text{\emph{Accepted}}\}$ and tags each grid point $\bx\in Z$ either \emph{Trial} or \emph{Accepted}. The \emph{Trial} points are defined as grid points which have been updated at least once by Hopf-Lax operator \eqref{eq:Hopf-Lax} but not frozen. The \emph{Accepted} points are these points that have been updated and frozen. This method can solve the fixed point system \eqref{eq:FixedPoint} in a single pass manner.

\begin{algorithm}[h]
\caption{Fast Marching Method}
\label{alg:FM}
\begin{algorithmic}
\renewcommand{\algorithmicrequire}{ \textbf{Input:}}
\Require
Metric $\cF$, stencil map $S$, initial source points collection $\cW$.
\renewcommand{\algorithmicensure }{ \textbf{Output:}}
\Ensure
Minimal action map $\cU$.
\renewcommand{\algorithmicrequire}{ \textbf{Initialization:}}
\Require
\State For each point $\x\in Z$, set $\cU(\x)\gets+\infty$ and $b(\x)\gets\text{\emph{Trial}}$.
\State For each point $\y\in\cW$, set $\cU(\y)\gets0$.\end{algorithmic}
\begin{algorithmic}[1]
\renewcommand{\algorithmicrequire}{ \textbf{While:} at least one grid point is tagged as \emph{Trial}}
\Require
\State Find $\x_{\rm min}$, the \emph{Trial} point which minimizes $\cU$.
\State $b(\x_{\rm min})\gets$ \emph{Accepted}. 
\For{All $\y$ such that $\x_{\rm min} \in S(\y)$ and $b(\y)\neq \text{\emph{Accepted}}$ }
\State Compute $\cU_{\rm new}(\y)$ using~\eqref{eq:Hopf-Lax}.
\State Set $\cU(\y)\gets\min\{\cU_{\rm new}(\y),\,\cU(\y)\}$.
\EndFor
\end{algorithmic}
\end{algorithm}
\end{appendix}

\section*{Appendix B: Elements of Proof of Finsler Elastica Minimal Paths Convergence}  
\begin{appendix}
Let $X\subset\bR^d$ be a compact domain and $\Im$ be the collection of compact convex sets of $\bR^d$. We will later specialize to $d=3$ for the application to the Euler elastica curves. The set $\Im$ is metric space, equipped with the Haussdorff distance which is defined as follows.
\begin{defi}
The Euclidean distance map from a set $A\subseteq\bR^d$ is
\begin{equation}
\label{eq:EuclideanDistance}
d(A,\x):=\inf_{\y\in A}\|\x-\y\|.
\end{equation}	
The Haussdorff distance between sets $A_1$, $A_2\subseteq\bR^d$ is
\begin{equation}
\label{eq:Haussdorff}
H(A_1,A_2):=\sup_{\x\in\bR^d}|d(A_1,\x)-d(A_2,\x)|.	
\end{equation}
\end{defi}

\begin{defi}
Let $\gamma\in C^0([0,1],X)$ be a path and $\cB\in C^0(X,\Im)$ be a collection of controls on $X$. The path $\gamma$ is said $\cB$-admissible iff it is locally Lipschitz and $\gamma^\prime(t)\in \cB(\gamma(t))$ for a.e. $t\in[0,1]$.
\end{defi}

\begin{defi}
A collection of controls  on $X$ is a map $\cB\in C^0(X,\Im)$. Its diameter $diam(\cB)$ and modulus\footnote{%
We actually use a slight variant of the classical modulus of continuity because the latter one, obtained with the hard cutoff function ${\mathfrak C(t)} = 1$ if $t\leq 1$, $0$ otherwise, lacks continuity in general.%
} of continuity $\Xi(\cB,\epsilon)$, are defined by $\Xi(\cB,0)=0$ and for $\epsilon>0$
\begin{align}
\label{eq:Diam}
&diam(\cB):=\sup\{\|\fu\|; \fu\in\cB(\x),\,\forall\x\in X\},\\
\label{eq:Modulus}
&\Xi(\cB,\epsilon):=\sup_{\x,\y\in X} H(\cB(\x),\cB(\y)) {\mathfrak C}(\|\x-\y\|/\epsilon),
\end{align}
where ${\mathfrak C}$ is the continuous cutoff function defined by 
\begin{equation*}
	{\mathfrak C}(t) = 
	\begin{cases}
	1 & \text{if } t \leq 1\\
	2-t & \text{if } 1 \leq t \leq 2\\
	0 & \text{if } t \geq 2.
	\end{cases}
\end{equation*}
\end{defi}
Clearly $\cB\mapsto diam(\cB)$ and $(\cB,\epsilon )\mapsto \Xi(\cB,\epsilon)$ are continuous functions of $\cB\in C^0(X,\Im)$  and $\epsilon\in\bR^+$. In addition $\Xi(\cB,\epsilon)$ is increasing w.r.t. $\epsilon$. Here and below, if $A_1$, $A_2$ are metric spaces, and $A_1$ is compact, then $C^0(A_1,A_2)$ is equipped with the topology of uniform convergence. This applies in particular to the space of paths $C^0([0,1],X)$ and of controls $C^0(X,\Im)$.

\begin{lema}
If $\gamma$ is $\cB$-admissible, then its Lipschitz constant is at most $diam(\cB)$. A necessary and sufficient condition for $\gamma$ to be $\cB$-admissible is: for all $0\leq p\leq q\leq 1$,
\begin{equation}
\label{eq:AdmissibleCondition}
d\left(\cB(\gamma(p)),\frac{\gamma(q)-\gamma(p)}{q-p}  \right)\leq \Xi(\cB,(q-p)\,diam(\cB)),	
\end{equation}

\begin{proof}
Assume that $\gamma$ is $\cB$-admissible. Then for any 	$0\leq p\leq q\leq 1$ one has
\begin{equation}
\|\gamma(p)-\gamma(q)\|\leq \int_{p}^{q}\|\gamma^\prime(\varrho)\|d\varrho\leq |p-q|\,diam(\cB),
\end{equation}
hence $\gamma$ is $diam(\cB)$-Lipschitz as announced. 
Denoting $w_\varrho= p+(q-p)\varrho$, for all $\varrho\in[0,1]$, one obtains
\begin{equation}
\frac{\gamma(q)-\gamma(p)}{q-p} = \int_0^1 \gamma^\prime(w_\varrho) d\varrho.
\end{equation}
Hence by Jensen's inequality and the convexity of $d(\cB(\gamma(p)),\cdot)$, which follows the convexity of $\cB(\gamma(p))$, we obtain
\begin{align}
\label{eq:Ad1}
d\left(\cB(\gamma(p)),\frac{\gamma(q)-\gamma(p)}{q-p}\right)&\leq \int_0^1 d(\cB(\gamma(p)),\gamma^\prime(w_\varrho))d\varrho\\
\label{eq:Ad2}
&\leq \int_0^1 H(\cB(\gamma(p)),\cB(\gamma(w_\varrho)))d\varrho\\
\label{eq:Ad3}
&\leq\Xi(\cB,(q-p)diam(\cB)).	
\end{align}
which establishes half of the announced characterization.  The inequality \eqref{eq:Ad2} follows from the admissibility property $\gamma^\prime(w_\varrho)\in\cB(\gamma(w_\varrho))$ and the definition of the Haussdorff distance \eqref{eq:Haussdorff}. The inequality \eqref{eq:Ad3} follows from the above established Lipschitz regularity of $\gamma$ and the definition of the modulus of continuity \eqref{eq:Modulus}.

Conversely, assume that $\gamma$ and $\cB$ obey \eqref{eq:AdmissibleCondition}. Then 
\begin{equation*}
\left\|\frac{\gamma(q)-\gamma(p)}{q-p} \right\|\leq diam(\cB)+\Xi(\cB,diam(\cB)),	
\end{equation*}
for any $0\leq p\leq q\leq 1$. Thus $\gamma$ is a Lipschitz path as announced, and therefore it is almost everywhere differentiable. If $p\in[0,1]$ is a point of differentiability, then letting $q\to p$ we obtain 
\begin{equation*}
d(\cB(\gamma(p)),\gamma^\prime(p))\leq \Xi(\cB,0)=0,	
\end{equation*} 
which is the announced admissibility property $\gamma^\prime(p)\in\cB(\gamma(p))$.\qed
\end{proof}
\end{lema}
The characterization \eqref{eq:AdmissibleCondition} is written in terms of continuous functions of the path $\gamma$ and control set $\cB$, hence it is a closed condition, which implies the following two corollaries. We denote $$T\cB(\x):=\{T\fu;\fu\in\cB(\x),\forall \x\in X \}.$$
\begin{coro}
\label{col:Closed}
The set $$\{(\gamma,\cB)\in C^0([0,1],X)\times C^0(X,\Im);
\gamma \text{ is }  \cB\text{-admissible}\}$$ is closed.
\end{coro}
\begin{coro}
\label{cor:Convergence}
Let $\x_n$, $\y_n$, $\cB_n$ and $T_n$ be converging sequences in $X$, $X$, $C^0(X,\Im)$ and $\bR^+$, with limits $\x_\infty$, $\y_\infty$, $\cB_\infty$, and $T_\infty$ respectively. Let $\gamma_n\in C^0([0,1],X)$ be a $(T_n+1/n)\cB_n$-admissible path with endpoints $\x_n$ and $\y_n$. Then the sequence of paths $(\gamma_n)$ is equip-continuous, and the limit $\gamma_\infty$ of any converging subsequence\, is a $T_\infty\cB_\infty$-admissible path $\gamma\in C^0([0,1],X)$ with endpoints $\x_\infty$ and $\y_\infty$.

\begin{proof}
Note that  the map $(T,\cB)\mapsto T\cB$ is continuous on\, $\bR^+\times C^0(X,\Im)$, hence the controls $\tilde \cB_n:=(T_n+1/n)\cB_n$ converge to $\tilde\cB_\infty:=T_\infty \cB_\infty$. Defining $E:=\sup\{diam(\tilde \cB_n);n>0\}$ which is finite by continuity of $diam(\cdot)$ \eqref{eq:Diam} and convergence of $\tilde \cB_n$, we find that the paths $(\gamma_n)_{n>0}$ are simultaneously $E$-Lipschitz, hence that a subsequence uniformly converges to some path $\gamma_\infty$. The $\tilde\cB_\infty$-admissibility of $\gamma_\infty$ then follows from Corollary \ref{col:Closed}.\qed
\end{proof}
\end{coro}
We next introduce the minimum-time optimal control problems. The minimum of \eqref{eq:MimimumTime} is attained by Corollary \ref{cor:Convergence}, which also immediately implies the Corollary \ref{coro:LowerSemi}.
\begin{defi}
For all $\x$, $\y\in X$, and $\cB\in C^0(X,\Im)$, we let
\begin{align}
\label{eq:MimimumTime}
\cT_\cB(\x,\y):=min\{ &T>0;\exists\gamma\in C^0([0,1],X),\nonumber\\
&\gamma(0)=\x,\gamma(1)=\y,\gamma \text{ is } T\cB\text{-admissible}\}.	
\end{align} 	
\end{defi}

\begin{coro}
\label{coro:LowerSemi}
The map $(\x,\y,\cB)\mapsto\cT_\cB(\x,\y)$ is lower semi-continuous on $X\times X\times\cC^0(X,\Im)$. 
In other words, whenever $(\x_n,\y_n,\cB_n)\to(\x_\infty,\y_\infty,\cB_\infty)$ as $n\to\infty$ one has
\begin{equation}
\label{ea:liminfT}
	\lim \inf \cT_{\cB_n}(\x_n,\y_n)\geq \cT_{\cB_\infty}(\x_\infty,\y_\infty).
\end{equation}
\end{coro}
\begin{proof}
For each $n>0$ let $T_n = \cT_{\cB_n}(\x_n,\y_n)$ and let $T_\infty = \lim \inf T_n$ as $n\to \infty$. 
Up to extracting a subsequence, we can assume that $T_\infty = \lim T_n$ as $n\to \infty$. 
Denoting by $\gamma_n$ a path as in Corollary \ref{cor:Convergence} for all $n \geq 0$, we find that there is a converging subsequence which limit $\gamma_\infty$ is $T_\infty \cB_\infty$ admissible and obeys $\gamma_\infty(0)=\x_\infty$ and $\gamma_\infty(1)=\y_\infty$. Thus $T_\infty \geq \cT_{\cB_\infty}(\x_\infty,\y_\infty)$ as announced.\qed
\end{proof}

\begin{defi}
Let 	$\cB_1$, $\cB_2\in C^0(X,\Im)$. These collections of controls are said included $\cB_1\subseteq \cB_2$ iff $\cB_1(\x)\subseteq\cB_2(\x)$ for all $\x\in X$. 
\end{defi}
The property $\cB_1\subseteq \cB_2$ clearly implies, for all $\x$, $\y\in X$
\begin{equation}
\label{eq:grterT}
\cT_{\cB_1}(\x,\y)\geq \cT_{\cB_2}(\x,\y).
\end{equation}

\begin{coro}
Assume that one has a converging sequence of controls $\cB_n\to\cB_\infty$ obeying the inclusions $\cB_n\supseteq \cB_\infty$ for all $n\geq 0$.
Then 
\begin{equation}
\label{eq:DistCV}
	\lim  \cT_{\cB_n}(\x,\y)= \cT_{\cB_\infty}(\x,\y)
\end{equation}
for all $\x,\y \in X$.
Let $T_n := \cT_{\cB_n}(\x,\y)$ for all $n\in {\mathbb N} \cup \{\infty\}$, and let $\gamma_n$ be an arbitrary $(T_n+1/n)\cB_n$-admissible path from $\x$ to $\y$. If there exists a unique $T_\infty \cB_\infty$-admissible path $\gamma_\infty$ from $\x$ to $\y$, then $\gamma_n \to \gamma_\infty$ as $n \to \infty$.
\end{coro}

\begin{proof}
The identity \eqref{eq:DistCV} follows from inequalities \eqref{ea:liminfT} and \eqref{eq:grterT}. 
By Corollary \ref{cor:Convergence} the sequence of paths $\gamma_n$ is equi-continuous, and any converging subsequence tends to a $T_* \cB_\infty$ admissible path $\gamma_*$ from $\x$ to $\y$, with $T_* := \lim T_n$. Since $T_*=T_\infty$ and by uniqueness we have $\gamma_*=\gamma_\infty$, hence $\gamma_n \to \gamma_\infty$ as announced.
\end{proof}

\noindent\textbf{Application to Finsler elastica geodesics convergence problem.}
\newline
Consider an orientation-lifted Finsler metric $\cF:X\times \bR^3\to \bR^+$  where $X:=\bar\Omega \subset \bR^3$.  For any $\bx\in X$, let $$\cB(\bx):=\{\bu\in\bR^3; \cF(\bx,\bu)\leq 1\}$$ be the unit ball of $\cF$. $\cB(\bx)$ is compact and convex, due to the positivity, continuity and convexity of the metric $\cF$ and  the map $\cB:X\to\Im$ is  continuous. Furthermore, using the homogeneity of the metric, one obtains for all $\bx$, $\by\in X$:
\begin{equation*}
\cL^*_\cF(\bx,\by)=\cT_\cB(\bx,\by),
\end{equation*}
where $\cL^*(\bx,\by)$ is the minimal curve length between $\bx$ and $\by$ with respect to metric $\cF$.

In the case of Finsler elastica problem,  one has  $\cB_\infty(\bx):=B^\infty_{\bx} $ and $\cB_\lambda(\bx):=B_{\bx}^\lambda$, $\forall \bx\in X$, where $B^\infty_{\bx} $ and $B_{\bx}^\lambda$ are defined in equations \eqref{eq:Ball_infty} and \eqref{eq:Ball_lambda} respectively. The Finsler elastica metrics $\cF^\lambda$ on $X$ pointwisely tend to the metric $\cF^\infty$ as $\lambda\to\infty$. Fortunately, the associated control sets $\cB_\lambda(\bx) \to\cB_\infty(\bx)$ uniformly in $C^0(\Omega,\Im)$, as can be seen from \eqref{eq:Charac_lambda}. Hence one has
\begin{align*}
\liminf\cL^*_{\cF^\lambda}(\bx,\by)&=\liminf \cT_{\cB_\lambda}(\bx,\by)\\
&\geq \cT_{\cB_\infty}(\bx,\by)=\cL^*_{\cF^\infty}(\bx,\by),
\end{align*}  
as $\lambda\to\infty$ for all $\bx,\by\in X$. To show that equality holds, it suffices to prove that sequence $\cB_\lambda$ obeys $\cB_{\lambda}\supseteq \cB_\infty$, equivalently to prove that $\cF^{\lambda}(\bx,\bu)\leq\cF^\infty(\bx,\bu)$ for all $\bx\in X$ and any vector $\bu\in\bR^3$. Indeed,  let $\bx=(\x,\theta)\in X$, and $\bu=(\fu,\nu)\in\bR^2\times \bR$:
\begin{align}
\cF^\lambda(\bx,\bu)&=\sqrt{\lambda^2\|\fu\|^2+2\lambda|\nu|^2}-(\lambda-1)\langle\fu, \vt\rangle\nonumber\\
&=\lambda \|\fu\| \left( -1+\sqrt{1+\frac{2|\nu|^2}{\lambda\|\fu\|^2}} \right)\nonumber\\
&\qquad+\lambda \|\fu\| -(\lambda-1)\langle \fu,\vt \rangle.\nonumber\\
&=\|\fu\|+\frac{|\nu|^2}{\|\fu\|}\left( \frac{2}{1+\sqrt{1+\frac{2|\nu|^2}{\lambda\|\fu\|^2}}} \right)\nonumber\\
&\qquad+(\lambda-1)(\|\fu\|-\langle \fu,\vt \rangle).\label{eq:denominator}\\
&\leq \cF^\infty(\bx,\bu).\nonumber
\end{align}
The last inequality holds because the denominator $1+\sqrt{1+\frac{2|\nu|^2}{\lambda\|\fu\|^2}}$ in \eqref{eq:denominator} is greater than 2, and $\|\fu\|\geq \langle \fu,\vt \rangle$ for any vector $\fu$ and any orientation $\theta$.

By Corollary \ref{cor:Convergence}, minimal paths $\cC^\lambda$ with endpoints $\bx$ and $\by$  for geodesic distance $\cL^*_{\cF^\lambda}(\bx,\by) $ converge as $\lambda\to\infty$ to a minimal path $\cC^\infty$ for $\cL^*_{\cF^\infty}(\bx,\by) $. We finally point out that $\cL^*_{\cF^\infty}(\bx,\by)<\infty $ for all $\bx$, $\by$ in the interior of $X$, provided this interior is connected, due to a classical  controllability result for the Euler elastica problem.

\end{appendix}

\bibliographystyle{spbasic}      
\bibliography{IJCV-Finsler.bib}   

%
%

\end{document}